\documentclass{article}

\usepackage[margin=1in]{geometry}
\usepackage{graphicx}%
\usepackage{multirow}%
\usepackage{amsmath,amssymb,amsfonts}%
\usepackage{amsthm}%
\usepackage{mathrsfs}%
\usepackage[title]{appendix}%
\usepackage{xcolor}%
\usepackage{textcomp}%
\usepackage{manyfoot}%
\usepackage{booktabs}%
\usepackage{algorithm}%
\usepackage{algorithmicx}%
\usepackage{algpseudocode}%
\usepackage{listings}%
\usepackage{mathtools}
\usepackage{tikz}
\usepackage{xspace}
\usepackage[most]{tcolorbox}
\usepackage{empheq}
\usepackage[nocompress]{cite}

\usepackage[colorlinks=true, allcolors=black]{hyperref}

\newtcbox{\mymath}[1][]{%
    nobeforeafter, math upper, tcbox raise base,
    enhanced, colframe=blue!30!black,
    colback=blue!30, boxrule=1pt,
    #1}

\makeatletter
\newcommand*\rel@kern[1]{\kern#1\dimexpr\macc@kerna}
\newcommand*\widebar[1]{%
  \begingroup
  \def\mathaccent##1##2{%
    \rel@kern{0.8}%
    \overline{\rel@kern{-0.8}\macc@nucleus\rel@kern{0.2}}%
    \rel@kern{-0.2}%
  }%
  \macc@depth\@ne
  \let\math@bgroup\@empty \let\math@egroup\macc@set@skewchar
  \mathsurround\z@ \frozen@everymath{\mathgroup\macc@group\relax}%
  \macc@set@skewchar\relax
  \let\mathaccentV\macc@nested@a
  \macc@nested@a\relax111{#1}%
  \endgroup
}
\makeatother

\newtheorem{problem}{Problem}
\newtheorem{theorem}{Theorem}

\newcommand{\spara}[1]{\paragraph{#1}}

\newcommand{\MSE}{\mathrm{MSE}}
\newcommand{\Polarization}{\mathrm{Polarization}}

\newcommand{\figspace}{\vspace{-1.0em}}

\newcommand{\RndTree}{\ensuremath{\mathsf{RndTree}}\xspace}
\newcommand{\RndCommunities}{\ensuremath{\mathsf{RndCommunities}}\xspace}
\newcommand{\GNP}{\ensuremath{\mathsf{GNP}}\xspace}
\newcommand{\Star}{\ensuremath{\mathsf{Star}}\xspace}
\newcommand{\Grid}{\ensuremath{\mathsf{Grid}}\xspace}

\newcommand{\bx}{\boldsymbol{x}}
\newcommand{\balpha}{\boldsymbol{\alpha}}
\newcommand{\bs}{\boldsymbol{s}}

\newcommand{\com}[1]{\hfill {\color{blue} /\!/ #1}}

\def\figwidth{0.4\linewidth}

\begin{document}


\title{Wiser than the Wisest of Crowds: \\ The Asch Effect and Polarization Revisited}

\author{Dragos Ristache \\ dragosr@bu.edu \\ Boston University \and Fabian Spaeh \\ fspaeh@bu.edu \\ Boston University \and  Charalampos E. Tsourakakis \\ ctsourak@bu.edu \\ Boston University}

\date{June 2024}

\maketitle

\begin{abstract}
In 1907, Sir Francis Galton independently requested 787 villagers to estimate the weight of an ox. Although none of them guessed the exact weight, the average estimate was remarkably accurate. This phenomenon is also known as wisdom of crowds. In a clever experiment, Asch employed actors to demonstrate the human tendency to conform to others' opinions. The question we ask is the following: what would Sir Francis Galton have observed if Asch had interfered by employing actors? Would the wisdom of crowds become even wiser or not? The problem becomes intriguing when considering the inter-connectedness of the villagers. This is the central theme of this work. Specifically, we examine a scenario where $n$ agents are interconnected and influence each other. The average of their innate opinions provides an estimator of a certain quality for an unknown quantity $\theta$. How can one improve or reduce the quality of the original estimator in terms of the mean squared error (MSE) by utilizing Asch's strategy of hiring a few stooges?

We present a new formulation of this problem, assuming that the nodes adjust their opinions according to the Friedkin-Johnsen opinion dynamics with susceptibility parameters~\cite{abebe2018opinion}. We demonstrate that selecting $k$ stooges for both maximizing and minimizing the MSE is NP-hard. Additionally, we demonstrate that our formulation is closely related to either maximizing or minimizing polarization~\cite{musco18} and present an NP-hardness proof. We propose a greedy heuristic that we implement efficiently, enabling it to scale to large networks. We test our algorithm on various synthetic and real-world datasets collected from Twitter against various baselines.
Although the MSE and polarization objectives differ, we find in practice that maximizing polarization often yields solutions that are nearly optimal for minimizing the wisdom of crowds in terms of MSE. Lastly, our analysis of real-world data reveals that even a small number of stooges can significantly influence the conversation surrounding the controversial topic of the war in Ukraine, resulting in a relative increase of the MSE of 207.80\% (maximization) or a decrease of 50.62\% (minimization).

\end{abstract}

\section{Introduction}
\label{sec:intro}
Opinion dynamics is a fascinating area of research that seeks to understand how individual beliefs and attitudes evolve over time within a social network. At its core, this field is driven by the desire to unravel the complex interplay between individual psychology, social influence, and network structure. By studying opinion dynamics, researchers aim to shed light on how opinions spread and converge within communities, how consensus is reached or polarization occurs, and how external factors such as media or political events can impact collective opinion. The insights gained from this research have profound implications for various domains, including politics, marketing, and public policy, as they can help predict and influence societal trends and behaviors. Researchers have developed a wide range of opinion dynamics models~\cite{proskurnikov2018tutorial} that aim to model how humans change their expressed opinions over time and have used them to optimize various objectives~\cite{biondi2023dynamics,zhu2022nearly,sun2023opinion,tang2021susceptible,zhu2021minimizing,gaitonde2020adversarial,musco18,abebe18,chen22,DBLP:journals/corr/abs-2206-08996}.

The concept of the wisdom of crowds highlights the collective intelligence that emerges when diverse groups of individuals come together to make decisions~\cite{surowiecki2005wisdom}. This phenomenon was first observed by Sir Francis Galton, who noted that the average guess (or according to some, the median guess) of 787 villagers regarding the weight $\theta$ of an ox was remarkably close to the actual value, as documented in the literature~\cite{wisdomofthecrowds2020}. At the same time it is also well known that humans influence each other. Solomon Asch conducted a set of seminal experiments in the study of social psychology and opinion dynamics~\cite{asch1955opinions}. These experiments aimed to investigate the extent to which social pressure from a majority group could influence an individual to conform. In a series of controlled experiments, participants were asked to match the length of a line on one card with one of three lines on another card. Unbeknownst to the subject, the other participants in the room were confederates of the experimenter and instructed to give incorrect answers unanimously in certain trials. The results showed that a significant proportion conformed to the majority's incorrect answer, even when the correct answer was obvious. Asch's experiments highlight the powerful influence of group conformity on individual judgment and decision-making, shedding light on the dynamics of opinion formation and change within social groups. At the same time, another phenomenon of interest to our work is polarization~\cite{haidt2022social} that represents a significant challenge in social systems, characterized by the division of opinions into distinct and often opposing camps. In social media, echo chambers refer to situations in which individuals are exposed primarily to opinions and information that reinforce their existing beliefs, leading to a narrowing of their perspectives and an increase in polarization. This phenomenon is often observed in online social networks and media platforms, where algorithms and social dynamics create filter bubbles that isolate users from diverse viewpoints~\cite{baumann2020modeling,chen2022polarizing}.

The goal of our work is to study the wisdom of crowds and polarization from an optimization perspective using a popular opinion dynamics model as the underlying opinion formation mechanism. Specifically, we examine the generalized Friedkin-Johnsen (FJ) model (c.f. Equation~\eqref{eq:fj} and \cite{abebe18}) that accommodates varying levels of susceptibility to persuasion among nodes. The selection of these two phenomena is not arbitrary; in fact, we demonstrate through experiments that according to the generalized Friedkin-Johnsen opinion dynamics model, maximizing polarization results in a significant decrease in the collective wisdom of the network agents, which is comparable to a near-optimal scenario. We also provide theoretical evidence why this is the case. Overall, in this work we make the following contributions:

\begin{itemize}
        \item We propose four new optimization problems that allow for the inclusion of $k$ stooges. The first two aim to maximize and minimize the mean squared error (MSE) of the average opinion value estimator at equilibrium, relative to the unbiased estimator of the average of innate opinions. The final two problems are dedicated to optimizing the polarization of opinions at equilibrium. Following the approach of Abebe et al.~\cite{abebe18}, our optimization variables are the resistance parameters, which encapsulate Asch's scenario of agents acting as stooges, with the flexibility to accommodate $k$ stooges in our model.

    \item We prove that all formulations are NP-hard. 
    
    \item We show that the problem is neither sub- nor supermodular and therefore not amenable to standard methods of convex optimization as other formulations in this line of research. Nonetheless, we provide an efficient greedy heuristic. To enable scalability, we provide a method for fast recomputation of the equilibirum opinions together with lazy evaluations.
    
    \item We run several experiments on synthetic and real-world data against natural baselines. This shows the merit of our heuristic, and provides interesting insights into the
    behavior of MSE and polarization, and their relationship under the FJ opinion dynamics. We observe that even for a small number of stooges, the relative increase or decrease achieved is substantial for both objectives across real-world datasets. For instance, in the case of 50 stooges on the Twitter network discussing real-world opinions about the Ukraine war, the mean squared error (MSE) can experience a relative increase of up to 207.80\% and a decrease of up to 50.62\%. 
    
\end{itemize}

\section{Related work} 
\label{sec:rel}
\spara{Models of Opinion Dynamics}~\cite{proskurnikov2018tutorial} explores how individual beliefs and attitudes evolve over time within a social network~\cite{acemouglu2013opinion}. These models aim to understand how interactions among individuals lead to the formation of collective opinions, consensus~\cite{de}, or polarization~\cite{musco18}. They often consider factors such as peer influence, social pressure~\cite{bindel2015bad}, and individual stubbornness or resistance~\cite{ghaderi2014opinion,abebe18} to change. By studying these dynamics, researchers can gain insights into the mechanisms behind opinion formation and change in various social contexts.

Friedkin and Johnsen~\cite{friedkin1990social} expanded the simple averaging French-DeGroot model~\cite{de,french1956formal} to include individuals' inherent beliefs and biases. Each node $v \in V$ represents an individual with an \emph{innate opinion} $s_u$ and an \emph{expressed opinion}. For each node $v$, the innate opinion $s_v \in [0,1]$ remains constant over time and is private, while the expressed opinion $x_v(t) \in [0,1]$ is public and changes over time $t \in \mathbb{N}_0$ due to social influence. Lower/higher values correspond to less/more favorable opinions on a given topic, and we use the convention that an opinion is neutral if it is equal to $\frac{1}{2}$. Initially, $x_v(0) = s_v$ for all users $v \in V$. At each subsequent time step $t > 0$, all users $v \in V$ update their expressed opinion $x_v(t)$ as the 
average of their innate opinion and the expressed opinions of their neighbors, as $x_v(t+1) = \frac 1 {1 + \deg(v)} (s_v + \sum_{v \in N(v)} x_u(t))$. As the number of rounds $t$ tends to $\infty$, the expressed opinions reach an equilibrium
$x^\star$ 
under mild conditions that are well understood~\cite{proskurnikov2018tutorial}. A popular variant of the FJ model is the generalized FJ model that allows for each node to have a varying degree of susceptibility to persuasion~\cite{cialdini2001science}. This model has been used by several papers~\cite{abebe18,ghaderi2014opinion,das2014modeling} and we discuss it in full detail in Section~\ref{sec:prelim}. Numerous other models of opinion dynamics exist, and interested readers can refer to the comprehensive survey by Proskurnikov and Tempo~\cite{proskurnikov2018tutorial} for an in-depth exploration.

\spara{Optimizing Objectives within the Friedkin-Johnsen Model} Although the concept of influence maximization has a rich history in discrete models, beginning with the groundbreaking work of Kempe, Kleinberg, and Tardos~\cite{kkt} (see also~\cite{richardson02,sadeh20}), the application of influence maximization in continuous opinion dynamics models has been largely overlooked, with the emphasis primarily on developing the models themselves. The seminal work of Gionis, Terzi, and Tsaparas inaugurated a new line of research~\cite{gionis13}. They tackle the challenge of optimizing the aggregate sum of opinions at equilibrium by selecting $k$ nodes and consistently fixing their expressed opinion to 1. Although their approach employs the conventional Friedkin-Johnsen (FJ) model as outlined in the previous paragraph, it inherently incorporates the concept of stubbornness, as the opinions of the chosen $k$ nodes remain unchanged.  Musco, Musco, and Tsourakakis~\cite{musco18} considered optimizing an objective that balances between disagreement and polarization at equilibrium. The disagreement of an edge is the squared difference of the equilibrium opinions of its endpoints, i.e., $d_{uv}(\bx^\star) = (x_u^\star - x_v^\star)^2.$ The total disagreement is defined as the sum of pairwise disagreements over all edges, i.e., $D_{G,\bx^\star} = \sum_{(u,v) \in E} d_{uv}(\bx^\star).$ Since then numerous other formulations along with algorithmic solutions have been derived~\cite{biondi2023dynamics,zhu2022nearly,sun2023opinion,tang2021susceptible,zhu2021minimizing}. For example, Gaitonde, Kleinberg, and Tardos~\cite{gaitonde2020adversarial}, along with Chen and Rácz et al. \cite{chen22,DBLP:journals/corr/abs-2206-08996}, explored budgeted adversarial interventions on inherent opinions. Their studies unveiled significant connections between spectral graph theory and opinion dynamics. Closest to our work lies the work of Abebe et al.~\cite{abebe18} who pioneered adversarial interventions at the susceptibility level using the generalized FJ model. They consider an unbudgeted optimization problem that is neither convex nor concave and designed a local search algorithm that reveals a remarkable structure on the local optima of the objective functions. Chan and Shan
proved the NP-hardness of the same problem subject to budget constraints~\cite{chan2021hardness}. 


\spara{Wisdom of Crowds (WoC)} is a concept popularized by James Surowiecki~\cite{surowiecki2005wisdom}, suggesting that under certain conditions, the aggregated judgment of a diverse group of individuals can be surprisingly accurate, even outperforming expert opinions. The phenomenon of the WoC has been recognized since the time of Galton~\cite{wisdomofthecrowds2020} and is considered the starting point for formal recognition of the WoC in the modern era despite earlier references in Aristotle's ``Politics," discussing the idea of collective judgment. Research has shown that the spread of misinformation online can significantly influence the WoC \cite{delVicario2016misinformation}. Additionally, exposure to ideologically diverse news and opinions on social platforms like Facebook has been studied to understand its impact on collective decision-making \cite{bakshy2015exposure}. Furthermore, experimental evidence suggests that social conventions can reach tipping points, beyond which collective behaviors can change dramatically \cite{centola2018experimental}. Several studies have examined how interventions impact the WoC, both from a theoretical perspective~\cite{golub2010naive,das2013debiasing,eger2016opinion,becker2017network,buechel2015opinion} and through empirical analysis using real-world data~\cite{becker2017network,madirolas2015improving}. Our work is most closely related to the recent study by Tian et al.~\cite{tian2023social}, which also examines the impact of social influence on Wisdom of Crowds (WoC) but employs a distinct statistical model with different optimization variables.

\section{Preliminaries} 
\label{sec:prelim}

\spara{Generalized Friedkin-Johnsen (FJ) Model.}
We consider the following generalized FJ model. There exists a set $V$ of \emph{agents}, where each agent $v \in V$ is associated with an \emph{innate opinion} $s_v \in [0,1]$. Each agent has a different level of propensity for conforming, captured by the {\it resistance} parameter $\alpha_v \in [0,1]$, which is also known as the stubbornness of node $v$~\cite{abebe18}.
Higher resistance values signify agents who are less susceptible to changing their opinion. The agents interact with one another in discrete time steps.
Consider an undirected graph $G = (V,E,W)$ 
represented by the column stochastic matrix $W=[w_{uv}]_{u,v\in V}$ with $w_{uv} \ge 0$.
We shall refer to $W$ as the {\it random walk matrix} or the {\it influence structure}, as $w_{uv}$ 
corresponds to how much node $u$ influences node $v$. 
The opinion dynamics evolve in discrete time $t=0,1,2,\ldots$ according to the following model~\cite{friedkin1990social,abebe2018opinion,das2014modeling}: 
\begin{equation}
\label{eq:fj}
\boxed{ x_v(0) = s_v, \quad
\textstyle
x_v(t+1) = \alpha_v s_v + (1-\alpha_v) \sum_{u} w_{uv} x_u(t).}
\end{equation}
for all $v \in V$.
It is known that this process converges
to a unique equilibrium \cite{abebe18} given some well understood conditions \cite{proskurnikov2018tutorial,ghaderi2014opinion},
and we therefore assume the existence
of a unique equilibrium throughout our work.
We denote the opinion of $v$
at equilibrium as
$x^\star_v = \lim_{t \to \infty} x_v(t)$.
We describe an instance of the
generalized FJ model via
the tuple $(W, \balpha, \bs)$
for $\balpha = (\alpha_v)_{v \in V}$
and $\bs = (s_v)_{v \in S}$.

\spara{Absorbing Random Walk Interpretation.} 
Analogously to \cite{gionis13}, we can interpret
the opinion formation process of Equation~\eqref{eq:fj}
via an absorbing random walk.
Let us use the random variable $X(t) \in V$ to
denote the state of the random walk at step $t$.
We start the random walk in state $X(0)$.
In each step $t$, we randomly decide
to stop the random walk and remain in state
$X(t)$
with probability $\alpha_{X(t)}$.
In this case, we write $X(\infty) = X(t)$
and stop the random process.
If we do not end the random walk, we transition
to another state $u$ with probability $w_{u, X(t)}$.
After the transition, we repeat the process
of deciding to end the random walk or continue
transitioning. We can use this interpretation to
characterize the
equilibrium opinion of each state $v \in V$ as
$x^\star_v = \mathbb E[s_{X(\infty)} \mid X(0) = v]$.
%

\section{Problem Definitions} 
\label{sec:dfn}

Let $\hat \theta = \frac 1 n \sum_{v \in V} s_v$
be the average innate opinion. Inspired by Galton's findings~\cite{wisdomofthecrowds2020} we assume that it is an unbiased estimator of the true quantity $\theta$. Let $\hat \theta^\star = \frac 1 n \sum_{v \in V} x_v^\star$
be the average opinion at equilibrium under the FJ dynamics. 
The \emph{mean squared error} (MSE) and the \emph{polarization}~\cite{musco18} for the equilibrium $\bx^\star$ are defined respectively as 
\begin{align}
    \label{eq:mse} \MSE &= \frac 1 n \sum_{v \in V} \big(x^\star_v - \hat \theta\big)^2,  \\
    \label{eq:polarization} \mathrm{Polarization} &= \frac 1 n \sum_{v \in V} \big(x^\star_v - \hat \theta^\star\big)^2 .
\end{align}
Consider an instance of opinion dynamics $(W, \balpha, \bs)$ under the generalized FJ model (Equation~\eqref{eq:fj}). We can set a node $u$ to be a stooge by setting its resistance value $\alpha_u$ to 0 or 1.
It is worth noting that provably, this may not be an optimal choice to optimize
either the MSE or the polarization, but it gives us a clear interpretation
of the stooge's role.
In our work, we consider these problems:


\begin{tcolorbox}
\begin{problem}[Optimizing MSE]
\label{prob:mse}  
Given $(W, \balpha, \bs)$ and an integer $k$, choose a set of $k$ stooges from the node set and their resistance values so that the MSE (Equation~\ref{eq:mse}) {\it at equilibrium} is maximized or minimized.
\end{problem}
\end{tcolorbox}

\begin{tcolorbox}
\begin{problem}[Optimizing Polarization]
\label{prob:polar}  
Given $(W, \balpha, \bs)$ and an integer $k$, choose a set of $k$ stooges from the node set and their resistance values so that polarization (Equation~\ref{eq:polarization}) {\it at equilibrium} is maximized or minimized.
\end{problem}
\end{tcolorbox}

\noindent
In other words, Problems~\ref{prob:mse} and~\ref{prob:polar} aim to identify a set of nodes in the network that will become stooges either for the common good (minimization) or for adversarial reasons (maximization) with respect to the altered equilibrium.
Upon a closer look on Equations~\eqref{eq:mse} and~\eqref{eq:polarization}, we can see how similar the
two measures are; the former is centered around the original estimator $\hat \theta$
while the latter around the equilibrium average $\hat\theta^\star$.
We will establish this similarity further for our hardness results
and in our experimental evaluation.
Recall also that we assume that the generalized FJ dynamics reach an equilibrium
both before and after selecting stooges,
as the problem of determining the existence of an equilibrium is already well-established in the
literature~\cite{proskurnikov2018tutorial,ghaderi2014opinion}.

\section{Hardness} 
\label{sec:hardness}


All four optimization formulations from Problems~\ref{prob:mse} and ~\ref{prob:polar} are NP-hard. 
Furthermore, we show in Section~\ref{sec:proposed} that they 
are neither sub- nor supermodular.
We first show the hardness of maximizing the polarization; the remaining proofs
follow with small adaptations.
In our problem, we pick a subset $S \subseteq V$ of
$|S| \le k$ stooges for which
we can change the resistance values $\alpha_u$
arbitrarily.
We use $\mathrm{Polarization}^\dagger$ to denote
the polarization of the equilibrium opinions
after changing the resistance
values of vertices in $S$.
We formulate a decision version of this problem:
Given an instance $(W, \balpha, \bs)$, an integer $k$ and a
threshold $\tau$, is there
a subset $S \subseteq V$ of stooges with $|S| \le k$
such that $\mathrm{Polarization}^\dagger \ge \tau$.
Analogous decision versions can be formulated
for the other three problems.
%

\begin{theorem}
\label{thm:pol-hardness}
Maximizing and minimizing the polarization (Prob.~\ref{prob:polar}) is NP-hard.
\end{theorem}

\begin{proof}
Let us first consider the maximization.
We extend the proof of \cite{gionis13} to our setting,
as we also reduce the decision version of
vertex cover on regular graphs to our problem:
Given an instance of vertex cover
consisting of a regular graph $H$ with vertex degree
$d$ and an integer $k$, we
construct an instance to our problem by modifying
$H$ as follows.
We introduce a special vertex ${\tilde u}$
which is totally absorbing (i.e. $\alpha_{\tilde u} = 1$)
and has innate opinion $s_{\tilde u} = 0$.
We further add a large empty
graph $\widebar K_{\ell}$
where every node $u \in \widebar K_\ell$
has innate opinion $s_u = 0$
and resistance $\alpha_u = 1$.
In the original graph $H$, we set
$\alpha_{u} = 0$,
$s_u = 1$, and
$w_{\tilde u u} = \frac 1 {d+1}$
for all nodes $u \in V(H)$.
We further set
$w_{uv} = \frac 1 {d+1}$ for
all $(u, v) \in E(H)$ and $w_{uv} = 0$, otherwise.
As such, our resulting instance
consists of the nodes $V = V(H) \cup K_\ell \cup \{ \tilde u \}$.
Furthermore, we use a threshold $\tau$
defined below,
as a function of $n$, $k$, $\ell$, and $d$.
We can show that $H$ has a vertex cover of
size $k$ if and only if there is a set of
at most $k$ stooges
such that $\mathrm{Polarization}^\dagger \ge \tau$.

$(\Rightarrow)$ Assume first that $S \subseteq V$ is a vertex cover of
$H$ with $|S| \le k$. In this case, we can pick
the vertices $S$ and set $\alpha_u \gets 1$ such that
$x_u = 1$ for all $u \in S$.
We can verify that this is indeed
a solution to our decision problem
by computing the equilibrium opinions
using the absorbing random walk
interpretation for opinion dynamics.
For any vertex $v \in V$,
we consider a random walk on the
modified graph $G$ starting in $v$.
If $v \in S$ then the random walk is
absorbed immediately after a single step
in $v$ itself, which
has innate opinion $s_v = 1$ and thus $x^\star_v = 1$.
If $v \notin S$ then the random
walk may get absorbed after the first step in
$\tilde u$ with probability
$\frac 1 {d+1}$.
If instead, the random walk goes
over an edge in the original graph $H$
(with probability $\frac d {d+1}$),
we must have reached a vertex in $S$ since
$S$ forms a vertex cover of $H$.
Since every vertex $u \in S$ is absorbing,
$x^\star_v = \frac d {d+1}$.
We compute the average equilibrium opinion as
\begin{align*}
    \hat \theta^\star = \frac 1 {|V|} \left(\sum_{u \in S} x^\star_u + \sum_{u \in V(H) \setminus S} x^\star_u +
        \sum_{u \in \widebar K_\ell} x^\star_u + x^\star_{\tilde u}\right)
    = \frac{k + (n-k) \frac d {d+1} + 0}{n + \ell + 1} \le \epsilon
\end{align*}
where $\epsilon$ can
be made arbitrarily small by choosing a large enough $\ell$.
This means that the polarization is at least
\begin{multline*}
    \mathrm{Polarization}^\dagger = \frac 1 {|V|} \sum_{u \in V} \big(x^\star_u - \hat \theta^\star\big)^2 \ge
    \frac 1 {|V|} \sum_{u \in V(H)} ((x_u^\star)^2 - 3 \epsilon)  =
    \frac 1 {|V|} \left( k + (n-k) \left(\frac{d}{d+1}\right)^2 - 3 n \epsilon \right)
    \eqqcolon \tau
\end{multline*}
and we define $\tau$ to be exactly at this lower bound.

$(\Leftarrow)$
Conversely, assume that $H$ has no vertex cover
of size $k$.
Let now $S$ be any set of stooges.
We still have that $x^\star_u \le 1$
for all $u \in S$ and
and $x^\star_u \le \frac d {d+1}$ for all $u\in V(H) \setminus S$.
However, because there is no vertex cover of
size $k$, there must exist a vertex $v \in V(H)$
from which we can take two steps without getting
absorbed in a node $u \in S$ with opinion $x^\star_u = 1$.
Thus, $x^\star_v \le \frac d {d+1} - ( \frac 1 {d+1} )^2$.
Since the total opinion could not have increased
compared to before,
we still have $\hat \theta^\star \le \epsilon$.
We compute the polarization as
\begin{align*}
    \mathrm{Polarization}^\dagger  &=
    \frac 1 {|V|} \sum_{u \in V} \big(x^\star_u - \hat\theta^\star\big)^2 \\[-3pt] &\le
    \frac 1 {|V|} \Bigg( \sum_{u \in V(H)} ((x^\star_u)^2 + \epsilon^2) + (\ell+1) \epsilon^2 \Bigg) \\ &=
    \frac 1 {|V|} \sum_{u \in V(H)} (x^\star_u)^2 + \epsilon^2 \\ &\le
    \frac 1 {|V|} \Bigg(k + (n-k-1) \left( \frac d {d+1} \right)^2 +
        \left( \frac d {d+1} - \left( \frac 1 {d+1} \right)^2 \right)^2 \Bigg) + \epsilon^2
    <\tau 
\end{align*}
which is true for sufficiently small $\epsilon$,
which we can achieve by choosing a large $\ell$.

To show NP-hardness
for minimization,
we only need to change
the innate opinions of the vertices
$u \in V(H)$ to $s_u = 0$ and
$s_{\tilde u} = 1$.
The MSE is large because
of the disparity of opinions
between vertices in the empty graph
$\widebar K_\ell$ and
vertices in $H$.
A vertex cover of $H$ now allows
us to effectively reduce the opinions
of vertices in $H$ and
we use this to define our
threshold.
Analogously to the above, we can also
show that if there is
no vertex cover in $H$, we
cannot reduce the opinions
below the threshold
for sufficiently large $\ell$.
\end{proof}

\noindent Since the number of stooges 
corresponds exactly to the number of
vertices in the vertex cover, we also
inherit the inapproximability of vertex
cover~\cite{khot08}.
An identical result holds for the MSE:

\begin{theorem}
\label{thm:mse-hardness}
Maximizing and minimizing the MSE (Problem~\ref{prob:mse}) is NP-hard.
\end{theorem}

\medskip

\begin{proof}
    Establishing the NP-hardness of the MSE is
    easier compared to the polarization 
    the equilibrium $x^\star$ is centered around $\hat{\theta}$ (see Equation~\eqref{eq:mse})
    rather than $\hat{\theta}^\star$. 
    Let us first consider the maximization of the MSE.
    We proceed as in the proof of Theorem~\ref{thm:pol-hardness}.
    Note that the average
    opinion of the innate opinions is
    \[
        \hat \theta = \frac 1 {|V|} \left(
        \sum_{u \in V(H)} s_u +
        \sum_{u \in \widebar K_\ell} s_u +
        s_{\tilde u}
        \right)
        = \frac 1 {n + \ell + 1} < \epsilon
    \]
    where $\epsilon$ serves the
    same purpose as in
    the proof of Theorem~\ref{thm:pol-hardness}.
    We can thus again make $\epsilon$
    arbitrarily small by choosing a large
    value $\ell$, and
    carry out the same proof.
    The adaption to the minimization works analogously
    to the adaption made for the proof of
    Theorem~\ref{thm:pol-hardness}.
\end{proof}

\section{Proposed methods} 
\label{sec:proposed}

The work of Gionis et al.~\cite{gionis13},
which launched a research stream of optimization under the FJ opinion dynamics,
uses that their objective can be formulated as a
monotone submodular function. 
A set function $f \colon 2^V \to \mathbb R^+$ is called submodular
if it has diminishing returns, i.e. $f(A \cup \{v\}) - f(A) \ge f(B \cup \{v\}) - f(B)$
for all sets $A \subseteq B \subseteq V$ and $v \in V$.
Monotone submodular functions can be approximately
maximized to a $(1-1/e)$-factor subject to a cardinality constraint \cite{fisher1978},
and minimized in polynomial time~\cite{fujishige2005submodular}.
Particularly, Gionis et al. leverage the property of diminishing returns
in their objective to obtain
an $(1-1/e)$-approximate solution.
We demonstrate that our problems lacks submodularity, and is thus is less structured
and not amenable to standard techniques of convex optimization;
it is also straight-forward to establish the absence of
supermodularity, which we defer to the Appendix.
Here, we detail the lack of submodularity
for the polarization (Problem~\ref{prob:polar}).
Problem~\ref{prob:mse} also lacks sub- and supermodularity, which
can be easily shown via a simplification
of our arguments.

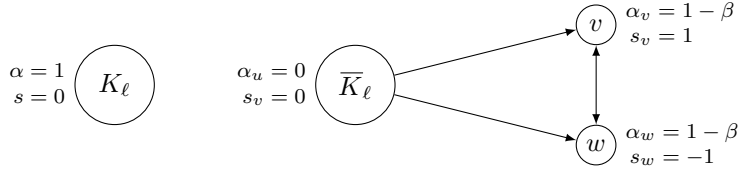
\begin{figure}[t]
    \centering
    \scalebox{1.0}{
    \begin{tikzpicture}[scale=1.6]
        \node[draw,circle,minimum size=30pt,inner sep=0pt,label={left,align=right:\footnotesize $\begin{aligned} \alpha &= 1 \\[-3pt] s &= 0 \end{aligned}$}] at (-2,0) {$K_\ell$};
        \node[draw,circle,minimum size=30pt,inner sep=0pt,label={left,align=right:\footnotesize $\begin{aligned} \alpha_u &= 0 \\[-3pt] s_v &= 0 \end{aligned}$}] (u) at (0,0) {$\widebar K_\ell$};
        \node[draw,circle,minimum size=15pt,inner sep=0pt,label={right,align=left:\footnotesize $\begin{aligned} \alpha_v &= 1-\beta \\[-3pt] s_v &= 1 \end{aligned}$}] (v) at (2,0.5) {$v$};
        \node[draw,circle,minimum size=15pt,inner sep=0pt,label={right,align=left:\footnotesize $\begin{aligned} \alpha_w &= 1-\beta \\[-3pt] s_w &= -1  \end{aligned}$}] (w) at (2,-0.5) {$w$};
        \draw[-latex] (u) to (v);
        \draw[-latex] (u) to (w);
        \draw[latex-latex] (v) to (w);
    \end{tikzpicture}}
    \vspace{-10pt}
    \caption{An example graph that shows that our objective is
    not submodular. We use $\overline K_\ell$ to denote the
    empty graph on $\ell$ nodes. Resistances and innate opinions
    of all nodes are as stated in the figure (for convenience,
    we allow negative innate opinions). We use a sufficiently
    large $\ell$ such that the contribution
    of $v$ and $w$ to the polarization vanishes.}
    \label{fig:non-submod}
\end{figure}%

\spara{Absence of Submodularity.}
Let us consider the example
network described in Figure~\ref{fig:non-submod}.
To simplify calculations,
we assume that $\ell$ is large enough
so that we may ignore the contributions of $v$ and $w$
to the polarization.
The example is based on the following idea:
Selecting $v$ as a stooge and changing
its resistance to $\alpha_v = 1$ raises the polarization
as opinions in the empty graph increase,
while they remain $0$ in the complete graph.
However, selecting $w$ as an additional stooge
and changing its resistance to $\alpha_w = 0$
raises the opinions in the empty graph even more
and the additional increase in polarization
is therefore larger.
\begin{itemize}
\item
\emph{Adding no stooges:}
In the original graph, we have
$x^\star_u = 0$ for all $u \in \widebar K_\ell$
since it is equally likely to assume the innate
opinion of $v$ and $w$ in a random walk starting
from $u$. Therefore, $\hat\theta^\star=0$ as well as
$\Polarization^\star = 0$.

\item
\emph{Adding $v$ as stooge:}
We select $v$ as a stooge and change its
resistance to $\alpha_v = 1$. This changes
the equilibrium opinions to
$x^\dagger_v = 1$,
$x^\dagger_w = 2\beta - 1$,
and therefore also
$x^\dagger_u = \frac 1 2 (x^\dagger_v + x^\dagger_w) = \beta$
for all $u \in \widebar K_\ell$.
Thus, $\hat\theta^\dagger = \frac 1 2 x^\dagger_u$
and
$\Polarization^\dagger = \frac 1 2 (x^\dagger_u - \hat\theta^\dagger)^2 + \frac 1 2 \hat\theta^{\dagger 2}
= \frac 1 4 \beta^2$.

\item
\emph{Adding $v$ and $w$ as stooges:}
We additionally select $w$ as a stooge. Note
that changing its resistance to $0$ would
result in a decrease in polarization to $0$, so we
change its resistance to $\alpha_w = 0$.
We obtain
$x^\ddagger_v = x^\ddagger_w = 1$,
as well as
$x^\ddagger_u = 1$ for all $u \in \widebar K_\ell$.
Therefore, $\hat\theta^\ddagger = \frac 1 2$ and
$\Polarization^\ddagger = \frac 1 4$.
\end{itemize}
We therefore have
$\Polarization^\dagger - \Polarization^\star = \frac 1 4 \beta^2$
which is smaller than
the increase
$\Polarization^\ddagger - \Polarization^\dagger = \frac 1 4 (1 - \beta^2)$
for $\beta < \sqrt{1/2}$.
This proves that our objective does not
show diminishing returns.
Note also that the ratio between 
the two increases is unbounded since
we can select $\beta$ arbitrarily small, which defies
applications of approximate submodularity. 
This alludes to the hardness of our
problem as it shows
non-convexity,
and is related to \cite{chen22}.

\spara{Absence of Supermodularity.}
\label{subsec:supermod} 

A set function $f \colon 2^V \to \mathbb R$
is supermodular if $-f$ is submodular.
%
We show that the  Problems~\ref{prob:mse}
and \ref{prob:polar}
are not supermodular. 
We can easily see this on a lollipop
graph consisting of a complete graph $K_n$ connected
to a path $P_m$. Assume that the resistance of
every vertex $u \in K_n$ is non-zero,
and the innate opinions are all $s_u = 0$.
The resistance of every vertex $v \in P_m$
is $\alpha_v = 0$.
Let $\tilde v \in P_m$ be the vertex that
is connected to the complete graph and
$\hat v \in P_m$ be another arbitrary vertex.
The idea is to select $\hat v$ and $\tilde v$ as
potential stooges in our example. As such, we 
use innate opinions $s_{\tilde v} = s_{\hat v} = 1$
and $s_v = 0$ for all
other vertices $v \in P_m \setminus \{\tilde v, \hat v\}$.
By this definition,
the equilibrium opinion of each vertex
$w \in K_n \cup P_m$ is $x^\star_w = 0$.
We therefore have that
$\mathrm{Polarization} = 0$ and
$\MSE = 0$ as
$\tilde v$ and $\hat v$ are the
only vertices with non-zero
innate opinion and their contribution
therefore diminishes for large enough
$n$ and $m$.

Now, let us consider the scenario
where we choose a vertex $v \in \{\tilde v, \hat v\}$ as a
stooge and make it fully resistant.
This means all vertices $w \in P_m$
that are farther away from the complete graph
than $v$ now have
equilibrium opinion $x_w^\dagger = 1$.
The equilibrium opinions for all other
vertices also increase and we thus
obtain an increase in MSE and polarization.
Let us designate
the vertex $\tilde v \in P_m$
connected to the complete graph
as a stooge.
Now, adding $\hat v$
as a second stooge does
not result in an any increase in MSE.
This shows that in this scenario,
the problem does exhibit strictly
diminishing returns and is therefore not 
supermodular.




\begin{algorithm}
\caption{Greedily Selecting Stooges}\label{alg:algo1}
\begin{algorithmic}[1]
\Function{ApproximateLazyGreedy}{$G, \phi, \epsilon$}
    \State Let $(W, \balpha, \bs) = G$
    \State $S \gets \emptyset$
    \com{initialize stooges}
    \State $\delta_{v,\beta} \gets \infty$ for all $v \in V$ and $\beta \in \{0, 1\}$
    \com{initialize marginal gains}
    \State $x \gets$ \Call{ApproximateEquilibrium}{$G, V, s$}
    \While{$|S| \le k$}
        \State $\delta_{\textrm{prev}} \gets 0$
        \State $\delta_{\textrm{max}} \gets 0$
        \State $(v_{\textrm{max}}, \beta_{\textrm{max}}) \gets (\bot, \bot)$
        \For{$(v,\beta) \in (V \setminus S) \times \{0, 1\}$ in order of decreasing $\delta_{v,\beta}$}
            \If {$\delta_{\textrm{prev}} > \phi \cdot \delta_{v,\beta}$}
                \com{marginal gain too small}
                \State {\bf break}
            \EndIf
            \State Let $\alpha^{(v,\beta)}_u = \alpha_u$ for all $u \in V \setminus \{v\}$ and $\alpha^{(v, \beta)}_{v} = \beta$
            \State Let $G^{(v, \beta)} = (V, W, \alpha^{(v, \beta)})$
            \com{opinion dynamic with $v$ as stooge}
            \State $x^{(v, \beta)} \gets$ \Call{ApproximateEquilibrium}{$G^{(v, \beta)}, \{v\}, x, \epsilon$}
            \State $\delta_{v, \beta} \gets \MSE(x^{(v, \beta)}) - \MSE(x)$
            \If {$\delta_{v, \beta} > \delta_{\textrm{max}}$}
                \State $\delta_{\textrm{max}} \gets \delta_{v, \beta}$
                \State $(v_{\textrm{max}}, \beta_{\textrm{max}}) \gets (v, \beta)$
            \EndIf
            \State $\delta_{\mathrm{prev}} \gets \delta_{v, \beta}$
        \EndFor
        \State $S \gets S \cup \{v_{\max}\}$
        \com{add $v_{\max}$ as stooge and update opinion dynamic}
        \State $\alpha \gets \alpha^{(v_{\max}, \beta_{\max})}$
        \State $x \gets x^{(v_{\max}, \beta_{\max})}$
    \EndWhile
    \State \Return $S$
\EndFunction
\vspace{0.5em}
\Function{ApproximateEquilibrium}{$G, U, x, \epsilon$}
    \State Let $(W, \balpha, \bs) = G$
    \While {$U \not= \emptyset$}
        \com{as long as there are active nodes}
        \State $x' \gets x$
        \For{$v \in U$}
            \State $x'_v \gets \alpha_v s_v + (1 - \alpha_v) \sum_{u \in V} w_{uv} x_u$
            \If{$|x'_v - x_v| < \epsilon$}
                \State $U \gets U \setminus \{v\}$
                \com{deactivate $v$}
            \Else
                \State $U \gets U \cup \{ u \in V : w_{uv} > \epsilon \}$
                \com{activate influential neighbors of $v$}
            \EndIf
        \EndFor
        \State $x \gets x'$
    \EndWhile
    \State \Return $x$
\EndFunction
\end{algorithmic}
\end{algorithm}

\spara{Greedy Heuristic.} Despite the lack of submodularity, 
we design a greedy algorithm to optimize the MSE and polarization.
We describe it in terms of maximizing the MSE in detail as Algorithm~\ref{alg:algo1}.
However, the algorithm can be easily adapted
to minimizing the MSE or optimizing the polarization.
We split the selection of stooges in two parts: computing
an approximation to the equilibrium opinions, and an
approximate version of the lazy greedy algorithm.
The former is implemented by repeated application of
Equation~\ref{eq:fj} until convergence. Additionally,
we deactivate nodes whose opinion is changing less
than a small threshold $\epsilon \ge 0$ and stop updating
their opinions. This avoids unnecessary computation
where the equilibrium opinions of large parts of the
graph are already reached.
We implement the greedy selection by considering
all pairs of vertices and extreme-value resistances $\beta \in \{0, 1\}$.
As in the well-known lazy greedy algorithm, we keep
track of the marginal gain of the MSE (or polarization) for each pair.
For a submodular objective,
the marginal gains decrease in each
subsequent iteration of the
outer while-loop due to diminishing returns.
After adding a stooge, we can now order
the remaining pairs in decreasing order
of their prior marginal gains before the addition.
We can stop recomputing the marginal gains
as soon as all prior marginal gains fall
below the current maximum, since we know they
can only decrease and will therefore never
exceed the current maximum.
Since our function is however not submodular and
we cannot guarantee diminishing returns, we
do not know whether the marginal gains are
decreasing. We account for it
by introducing some multiplicative slack $\phi \ge 1$ and
only stop recomputing marginal gains once the slack
times an updated marginal gain fall below the
current maximum.

\section{Experimental Evaluation} 
\label{sec:exp}

We evaluate
Algorithm~\ref{alg:algo1} on
Problems~\ref{prob:mse} and \ref{prob:polar}
against several intuitive baselines.
We show further results
in Appendix~\ref{sec:app-exp}.
We used
Python~3 and ran our methods on a 2.9 GHz
Intel Xeon Gold 622R processor
with 384GB RAM.
Our code is available online\footnote{\url{https://github.com/MKLOL/asch-effect-polarization}}

\begin{figure}
    \centering
    \includegraphics[width=\figwidth]{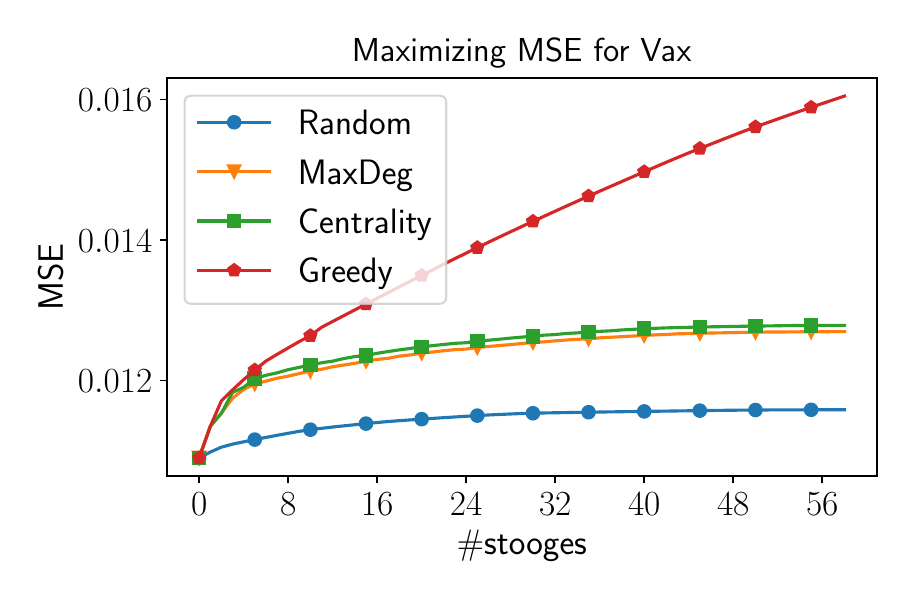}~
    \includegraphics[width=\figwidth]{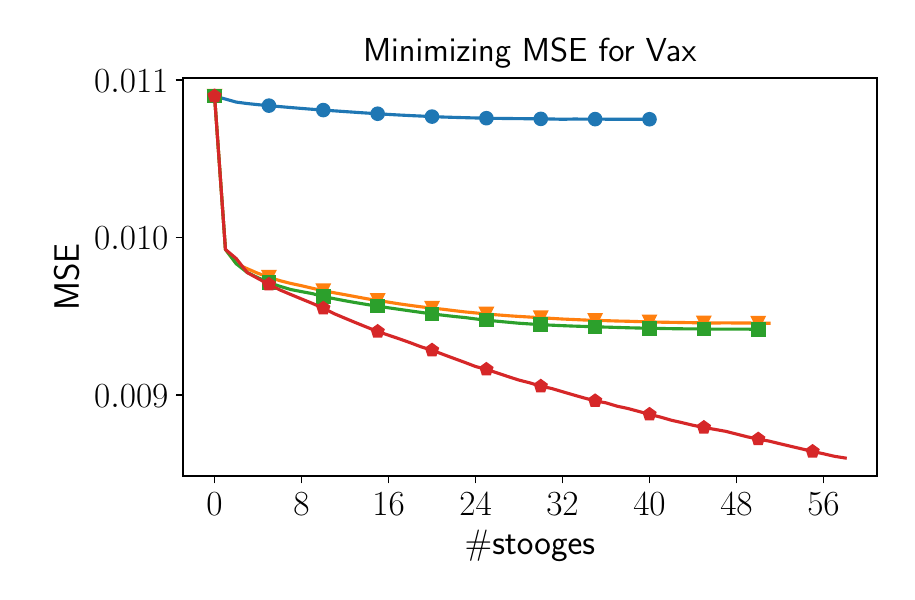}
    
    \vspace{-6pt}
    
    \includegraphics[width=\figwidth]{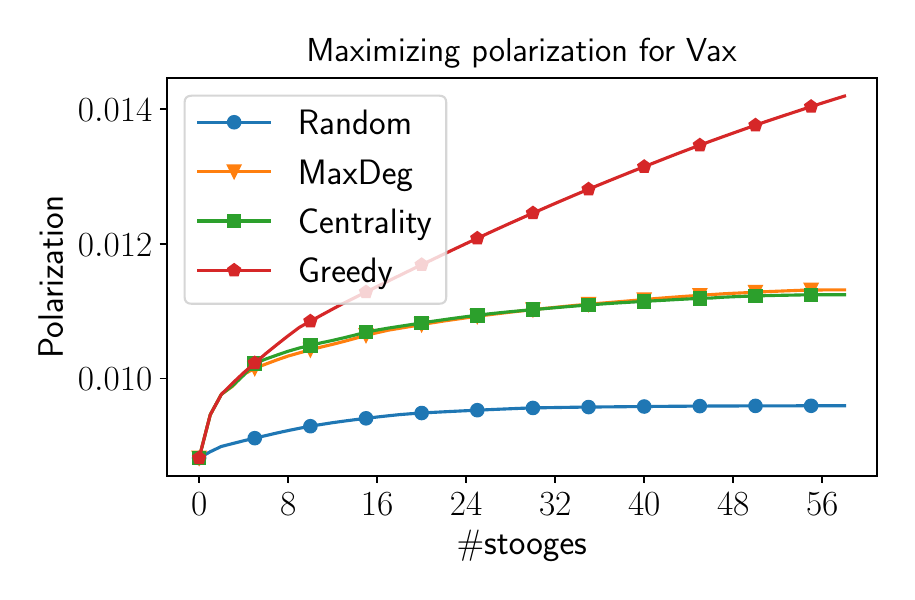}~
    \includegraphics[width=\figwidth]{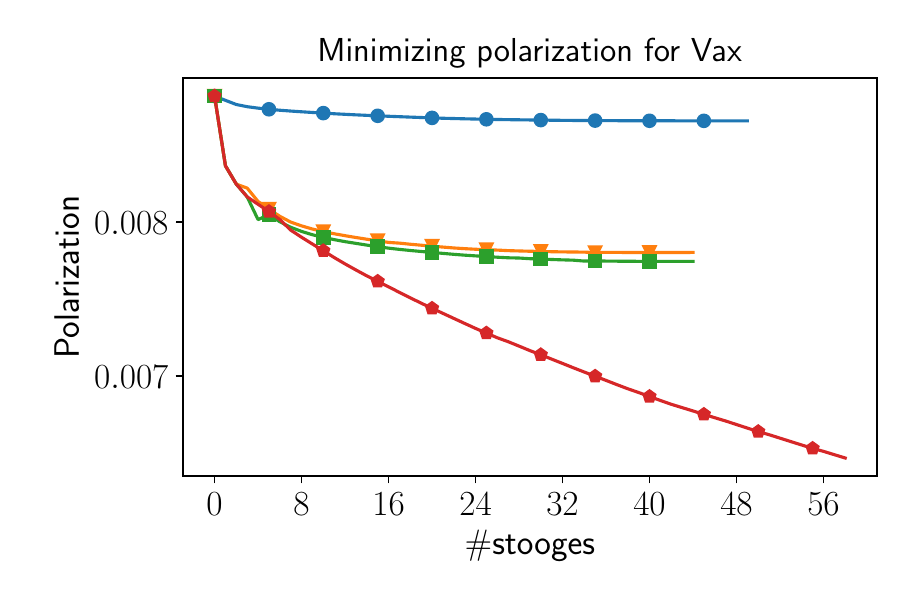}
    \figspace
    \caption{Maximization (left) and minimization (right) of the MSE (top) and polarization (bottom) on the \textsf{Vax} instance.}
    \label{fig:vax}
    \label{fig:vax-pol}
\end{figure}

\subsection{Setup}

\spara{Synthetic Instances.}
We consider different random graphs for our
synthetic experiments.
We use a $\mathsf{GNP}(n, p)$ model where each possible edge
on a graph of size $n$
is included independently with probability $p$.
We also consider a more refined stochastic block model
which we refer to as \RndCommunities.  
It has a total of $n=300$ nodes of which $200$ form
a big community and the remaining $100$ nodes form $10$ small
communities of size $10$, each.
Small communities are only connected to the big community, not to each other.
Each intra-community edges exists independently with probability $0.5$
and each inter-community edge with probability $0.3$. We choose this model as it simulates the existence of several disconnected small communities connected to a large, expander-like core as described in ~\cite{leskovec2008statistical}.
We also consider random trees $\RndTree(n)$
where we pick a tree uniformly at random from the set
of all trees on $n$ vertices by choosing a random Prüfer sequence~\cite{prufer1918neuer} and converting it into a tree (see also~\cite{broder1989generating}). 
We also run experiments on a star graph with $150$
leaves (denoted $\Star$) and a $10 \times 5$ grid
graph (denoted $\Grid$).
Unless otherwise noted,
we sample
innate opinions independently from
$\mathcal N(\mu=\frac 1 2, \sigma^2=\frac 1 2)$
which we clip for convenience to $s_v \in [0,1]$.
Throughout, we use resistance values of $\alpha_v = \frac 1 2$
for all nodes.

\spara{Real-World Instances.}
We use a new Twitter dataset~\cite{ukrainevax} compiled using twAwler \cite{pratikakis2018twawler} before Twitter changed its data collection policy.
We create two instances \textsf{Vax}
and \textsf{War}, consisting of
a network of users which posted
about the Ukraine war and vaccinations.
We include an undirected edge between
two users if there is at least one follow or a mention
relationship between them.
The \textsf{Vax} and \textsf{War} datasets
are on the same graph of $3\,409$ nodes and $11\,508$ edges,
consisting of users 
who have expressed opinions on both topics.
The opinions are obtained by parsing $12\,257$ tweets posted by these users.
The sentiment of the tweets was analyzed using ChatGPT-3.5.
We provide the prompts we used in Appendix~\ref{sec:app-exp}. 
The sentiment scores ranged from 0 to 10 (and were subsequently normalized to the interval [0,1]), where 0 indicates complete opposition to vaccination, 10 represents full support, and 5 denotes neutrality.
For this dataset, we assume that the number of tweets
act as a proxy on how opinionated a user is.
Specifically, we give higher resistance values to users with many tweets,
and we detail the exact formula in
Appendix~\ref{sec:app-exp-desc}.
%

We use further instances generated
from Twitter by 
Garimella et al. \cite{garimella18}. 
Each instance captures users which
posted with a specific hashtag, and
we infer an extreme innate opinion
$s_v \in \{0, 1\}$
from the sentiment expressed in that post.
Users form a graph where the edges
express the retweet or follow relationship
between the users. Note that the resulting
graph is also undirected even though these
relationships are directed.
We use resistances 
$\alpha_v = \frac 1 2$ for all nodes.

\begin{figure}
    \centering
    \includegraphics[width=\figwidth]{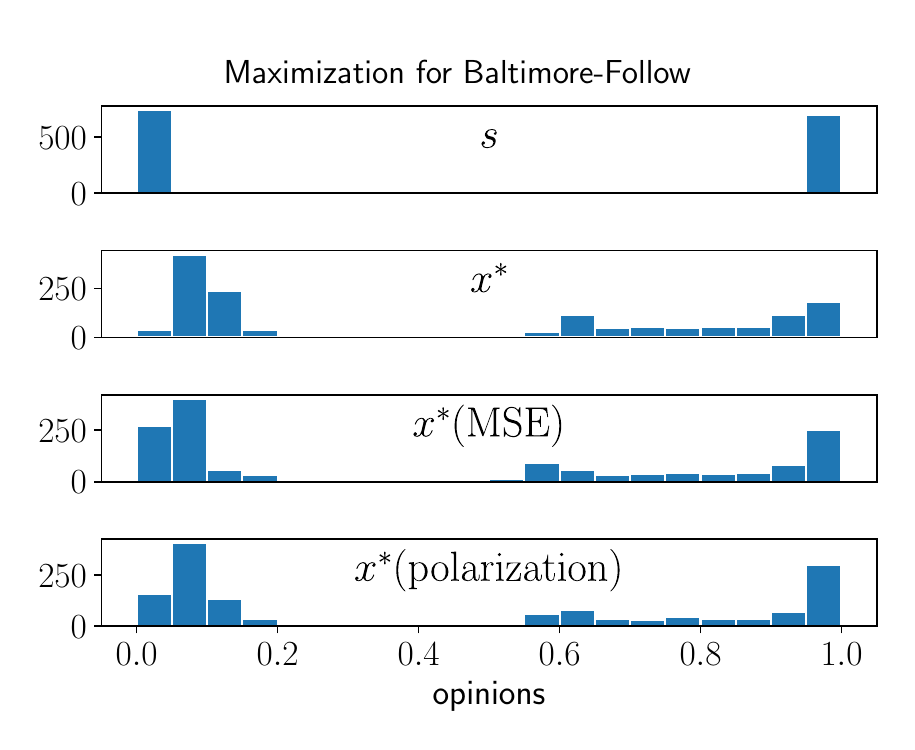}~
    \includegraphics[width=\figwidth]{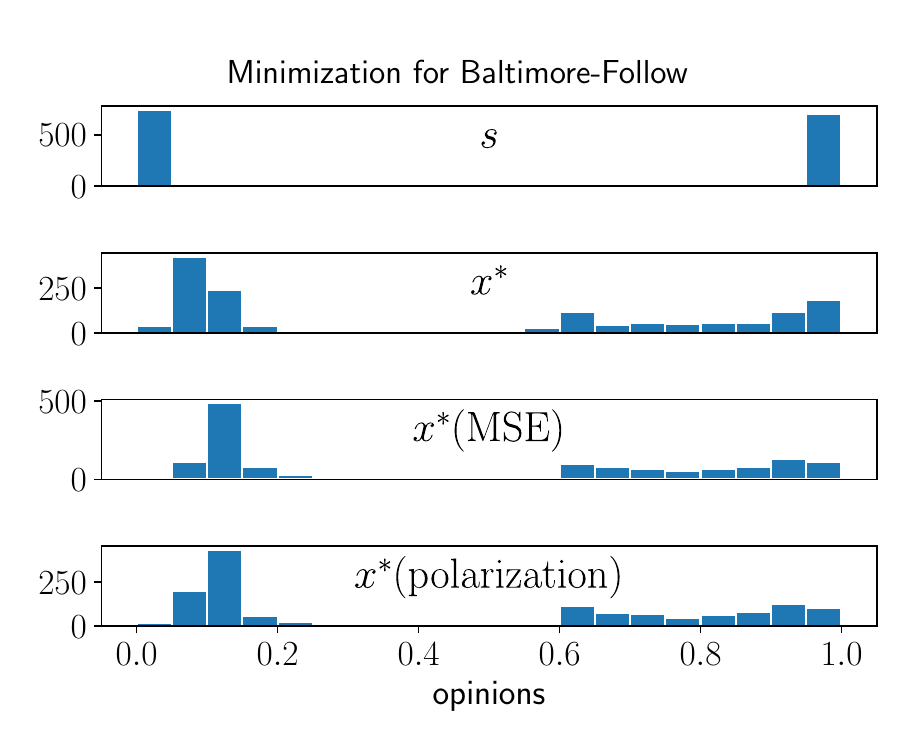}
    \figspace
    \caption{Innate opinions ($s$), equilibrium opinions before
    ($x^\star$) and after adding stooges for optimizing the MSE ($x^\star(\mathrm{MSE})$)
    and the polarization ($x^\star(\mathrm{Polarization})$).
    We show maximization (left) and minimization
    (right) on the \textsf{Baltimore-Follow}
    dataset.}
    \label{fig:baltimore-opinions}
\end{figure}


\spara{Algorithms.}
We use the greedy heuristic described in Algorithm~\ref{alg:algo1}
with parameters $\epsilon=10^{-5}$ and $\phi=1.1$,
which we refer to as $\mathsf{Greedy}$.
For baselines, we consider selecting
stooges randomly ($\mathsf{Random}$),
according to descending maximum degree ($\mathsf{MaxDegree}$),
or descending betweenness-centrality ($\mathsf{Centrality}$).
For each of these baselines, we also follow a greedy approach
to decide the modified resistance values:
From the set of selected stooges $S$,
we pick the stooge $u \in S$ such that its change
in resistance to either $\alpha_u = 0$
or $\alpha_u = 1$
maximizes (or minimizes) the objective.
We continue this for all remaining stooges.
All algorithms use the fast approximate calculation
of the equilibrium opinions as stated in
Algorithm~\ref{alg:algo1} with $\epsilon=10^{-5}$.
We stop an algorithm early when additional
stooges do not lead to any further
improvement.
As an additional baseline, we use
a brute-force approach ($\mathsf{Brute\ Force}$) which
enumerates all possible subsets
of stooges.


\spara{Evaluation Metrics.}
For our experiments, we report
the MSE and Polarization as
defined in Equations~\ref{eq:mse}
and \ref{eq:polarization}.
To compare two sets of stooges
$A, B \subseteq V$,
we use the Jaccard similarity index
$J(A,B) = \frac{|A \cap B|}{|A \cup B|}$.

\begin{figure}
    \centering
    \includegraphics[width=\figwidth]{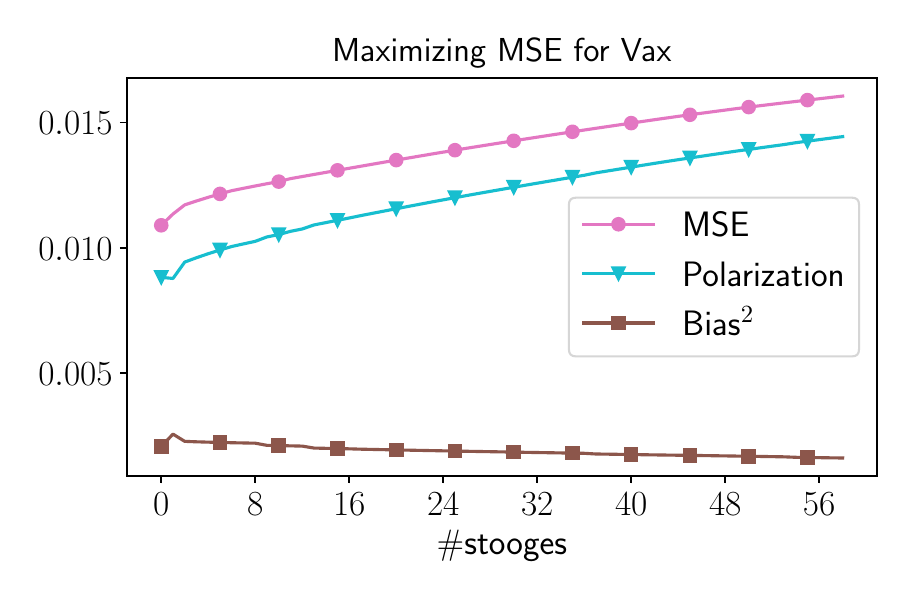}~
    \includegraphics[width=\figwidth]{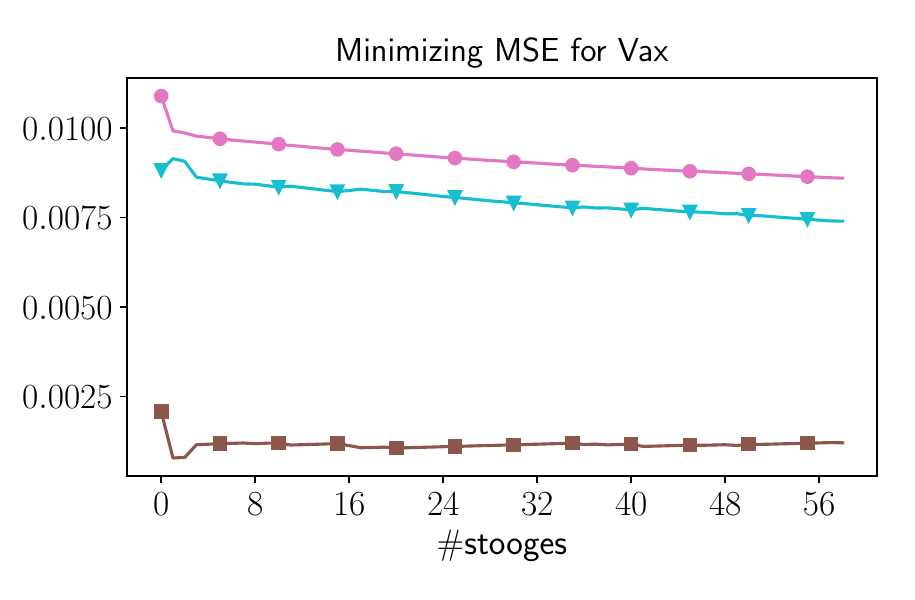}
    \figspace
    \caption{Maximization (left) and minimization (right) of the MSE on the \textsf{Vax} dataset.
    We show $\mathrm{Bias}^2 = \big(\hat \theta - \frac 1 {|V|}\sum_u x^\star_u\big)^2$,
    variance (polarization), and the MSE.}
    \label{fig:decomp}
\end{figure}

\subsection{Real-World Experiments}

We now evaluate our algorithm and baselines
on real-world datasets obtained
from Twitter (cf. Appendix~\ref{sec:app-exp} for the full results).
We further note that results on synthetic
data follow similar trends as on
real-world data, but the advantage
of our greedy approach over baselines
is even more pronounced.

We show
the result of the maximization and minimization objectives
on the \textsf{Vax} instance
in Figure~\ref{fig:vax}.
Our greedy approach
consistently outperforms the
other baselines.
Selecting stooges based
on centrality or the maximum degree only
proves sensible for a few stooges.
Figure~\ref{fig:baltimore-opinions} shows the
opinion distribution before and after adding stooges,
on the \textsf{Baltimore-Follow} and \textsf{Vax}
instances using \textsf{Greedy}.
Notice that the overall MSE and polarization for \textsf{Vax} is
small due to the non-polarizing effect of the FJ model~\cite{dandekar2013biased} and also due to the large number
of nodes in the network compared to the stooges.  However, the relative increase we obtain is substantial ($\ge 40\%$ for \textsf{Vax}) as indicated by the plots. 
We can also see how maximization
of the MSE and Polarization increases extreme
opinions that are far
from the true mean $\theta$,
and minimization pushes
opinions closer together.

Figure~\ref{fig:decomp}
shows the decomposition of the MSE
into bias and variance, for
an increasing number of stooges
selected by \textsf{Greedy}.
The bias remains low
throughout which indicates
that both the MSE and Polarization
objectives are similar:

\begin{figure}
    \centering
    \includegraphics[width=\figwidth]{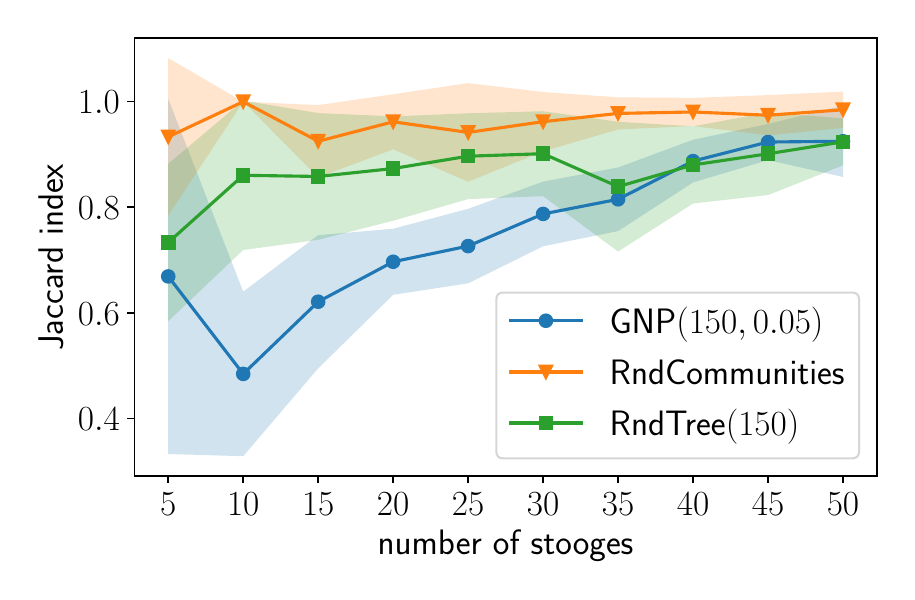}~
    \includegraphics[width=\figwidth]{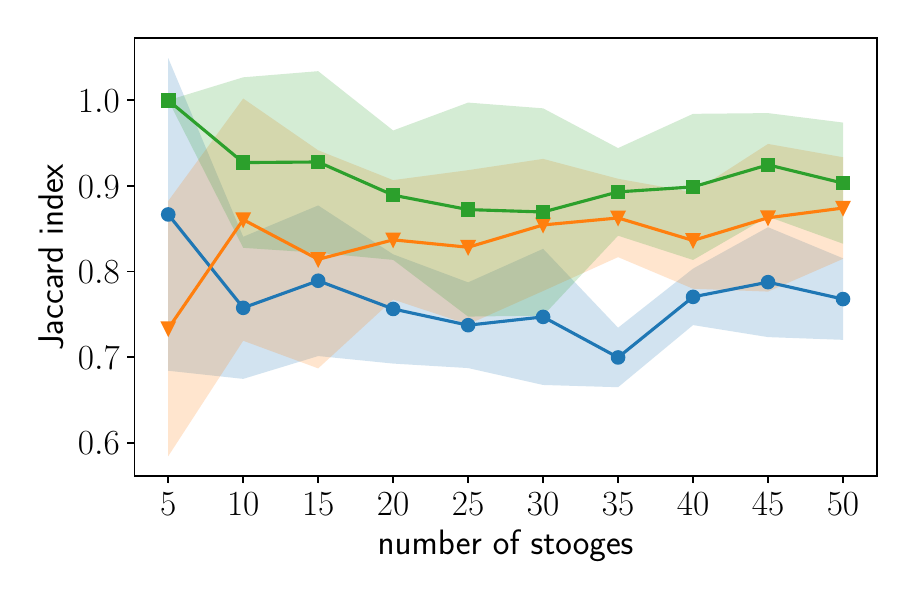}
    \figspace
    \caption{Comparing the set of stooges chosen for optimizing the
    MSE and Polarization. We report the Jaccard similarity on both sets for
    an increasing number of stooges.
    The left plot shows the maximization and the right plot the minimization
    for different random graphs. We report mean and standard deviation across five runs.}
    \label{fig:isect-pol-mse}
\end{figure}

\spara{Comparing MSE and Polarization.}
In order to explore the differences
between optimizing the MSE (Problem~\ref{prob:mse})
and the Polarization (Problem~\ref{prob:polar}),
we showcase the difference in the set of stooges
selected by Algorithm~\ref{alg:algo1} for both
objectives (Figure~\ref{fig:isect-pol-mse}).
We can see that similar stooges are selected
for both objectives. Furthermore, using
the set of stooges selected for the other
objective (i.e. using the stooges selected
for optimizing the MSE for the Polarization and vice versa)
results in a loss in objective of
at most $1\%$ for maximization and $6\%$
for minimization for the
random graphs in Figure~\ref{fig:isect-pol-mse}.

\spara{Scalability.}
Figure~\ref{fig:scalability-n} shows the runtime
of our approach and the three baselines.
Our greedy approach clearly outperforms
the different baselines as it lazily avoids
recomputation. We note, however, that this kind
of optimization would also be possible for
the baseline approaches. $\mathsf{Centrality}$
is the most expensive among the baseline approaches
as it requires identifying the $k$ vertices with
maximum centrality, which can be prohibitive for
large values of $n$.

\begin{figure}
    \centering
    \includegraphics[width=\figwidth]{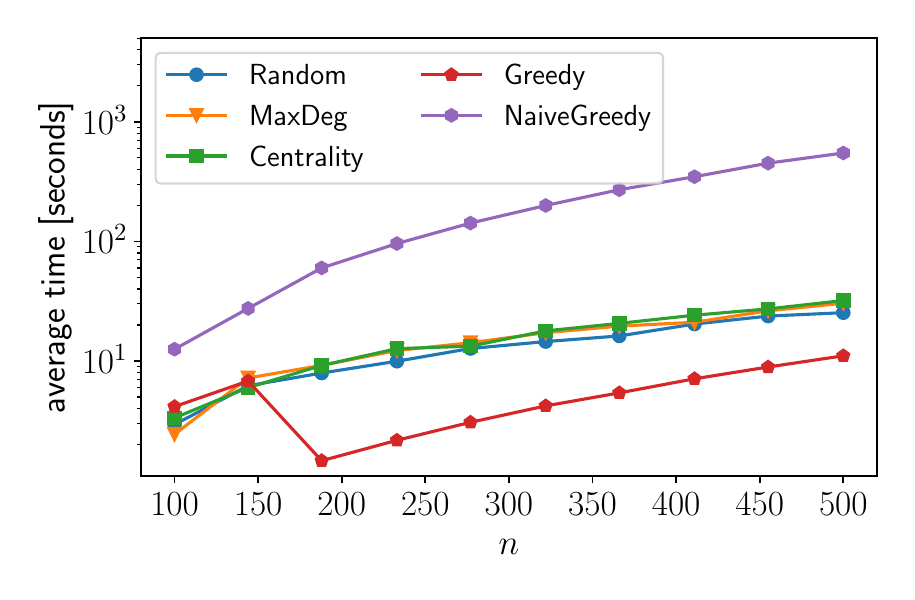}
    \vspace{-5mm}
    \caption{Runtime on $\GNP(n, 0.05)$ for an increasing number of vertices $n$, for selecting $50$ stooges over 5 runs. We omit the standard deviation as runtimes are well concentrated.}
    \label{fig:scalability-n}
\end{figure}


\section{Conclusion}
\label{sec:concl}

We investigate the influence of social conformity using the Friedkin-Johnsen model on the wisdom of crowds by extending the classic experiment of Sir Francis Galton with a twist inspired by Asch's conformity experiments. In particular, we explore the impact of introducing stooges (i.e., biased individuals) on the group's collective estimation accuracy. We demonstrate NP-hardness for both maximizing and minimizing the mean squared error (MSE) of the group's estimate. Our study reveals a connection between optimizing MSE and polarization in opinion dynamics. Despite the computational challenges, we propose a greedy heuristic that works well on synthetic and real-world datasets, suggesting practical implications for understanding and potentially manipulating collective decision-making processes in connected societies. A major open question is the design of approximation algorithms for our objectives.

\bibliographystyle{splncs04}
\bibliography{ref}

\newpage\clearpage

\appendix

\label{sec:appendix}

\section{Additional Experiments}
\label{sec:app-exp-desc}
\label{sec:app-exp}

\subsection{Description of Experimental Setup}

\spara{Resistances and Innate Opinions for \textsf{Vax} and \textsf{War}.}
Here, we detail the formula used to set
the resistance of a user for the real
world datasets \textsf{Vax} and \textsf{War},
which are both compiled from Twitter data.
Let $\mathrm{tw}_u$ be the number of tweets of user $u$.
We use the following formula to set the resistance values
for each node $u$:
\begin{equation}
   \alpha_u \sim 
\begin{cases} 
U(0.4, 0.6) & \text{if } 1 \leq \mathrm{tw}_u \leq 5 \\
U(0.5, 0.7) & \text{if } 5<  \mathrm{tw_u} \leq 10 \\
U(0.6, 0.8) & \text{if } 10< \mathrm{tw_u} \leq 20 \\
U(0.7, 0.9) & \text{if } \mathrm{tw_u}  > 20 .
\end{cases} 
\end{equation}

\spara{ChatGPT-3.5 Prompts.} We used the following prompts to extract an opinion
for each tweet.

\begin{enumerate}
    \item Return only integer rate for the following tweet. The rate represents its opinion towards COVID vaccination with an integer between 0 to 10, with 10 being very positive and supportive of vaccination, 0 being very negative and skeptical about it, and 5 being completely neutral. COVID vaccination is also called vax.
    \item  Return an integer in the range 0 to 10 (without additional comments, no letters at the sentence) for the following tweet. The rate represents its opinion towards a hypothetical war between Ukraine and Russia with an integer between 0 to 10, with 10 being totally supportive of Ukraine, 0 being totally supportive of Russia, and 5 being completely neutral.
\end{enumerate}

\subsection{Experimental Results on Synthetic Data}

\begin{figure}
    \centering
    \includegraphics[width=\figwidth]{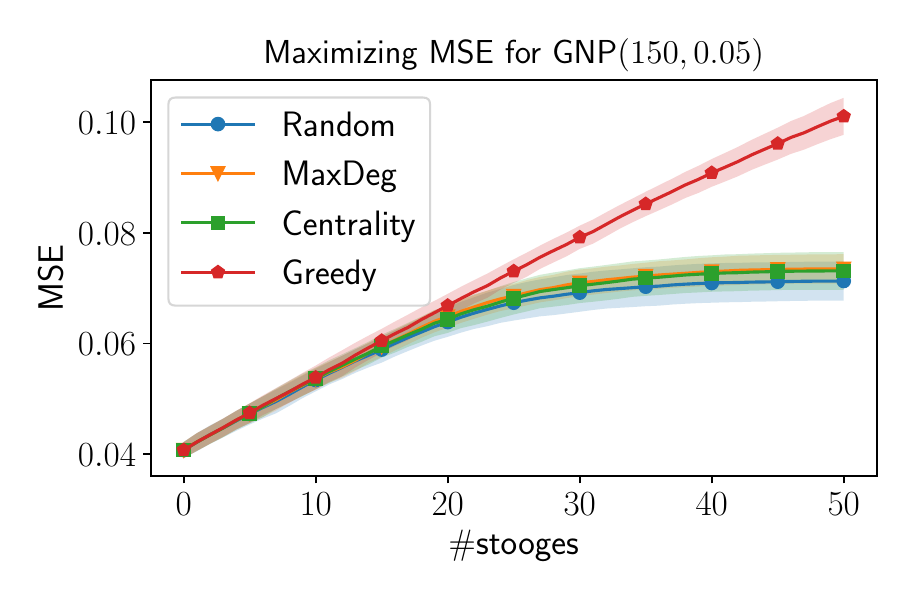}~
    \includegraphics[width=\figwidth]{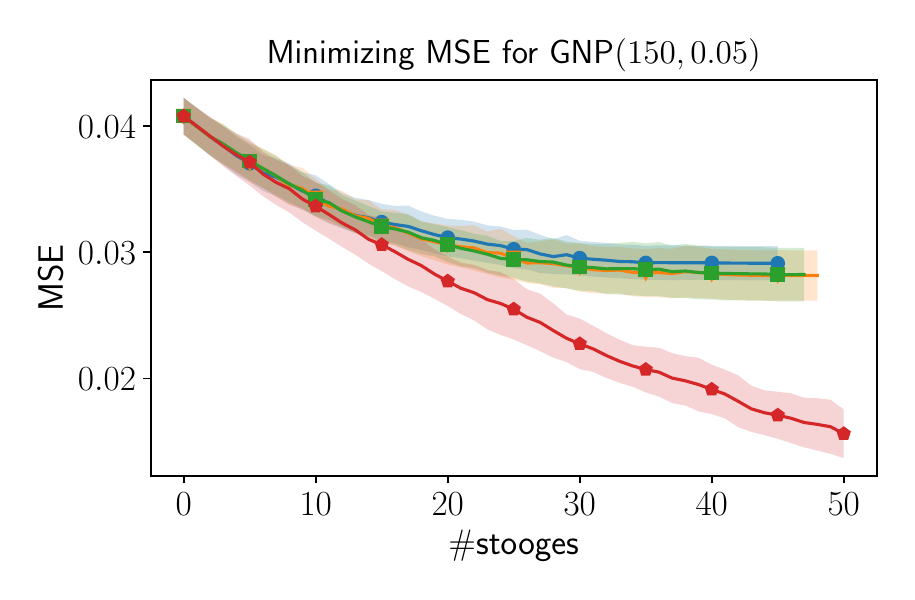}
    
    \includegraphics[width=\figwidth]{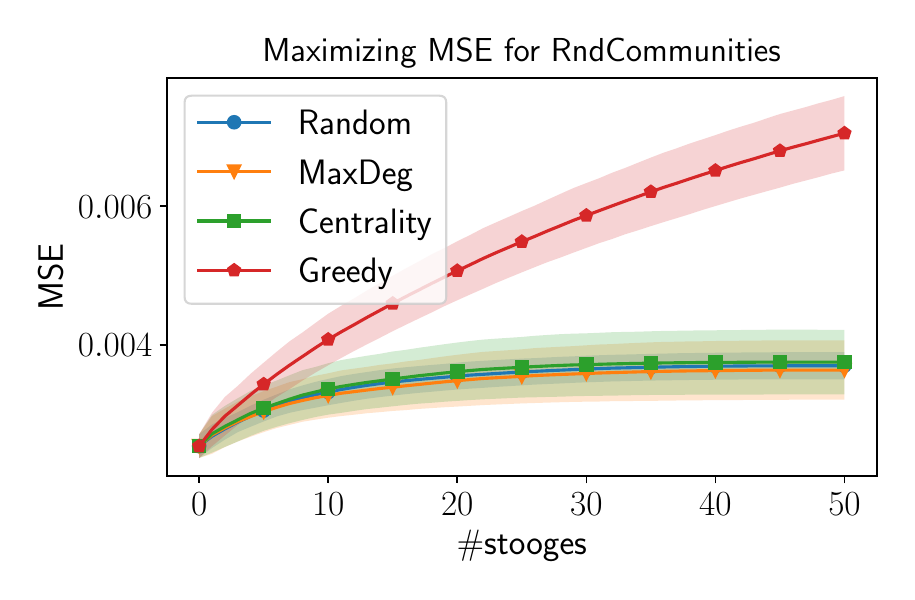}~
    \includegraphics[width=\figwidth]{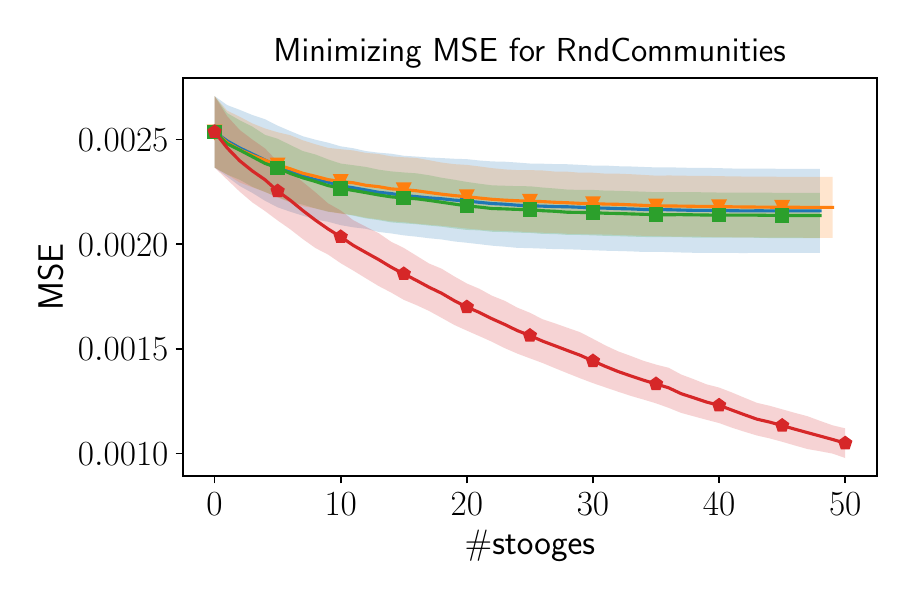}
    
    \includegraphics[width=\figwidth]{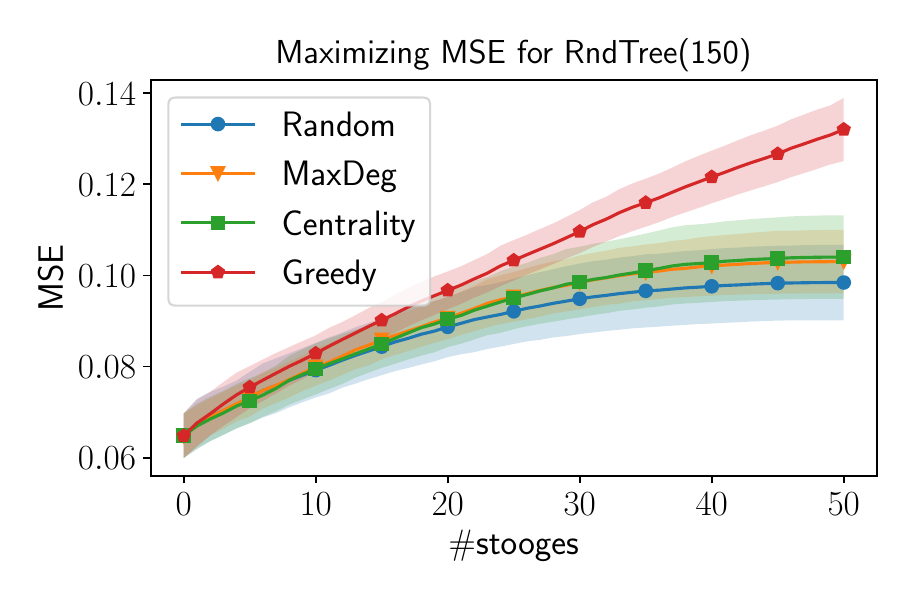}~
    \includegraphics[width=\figwidth]{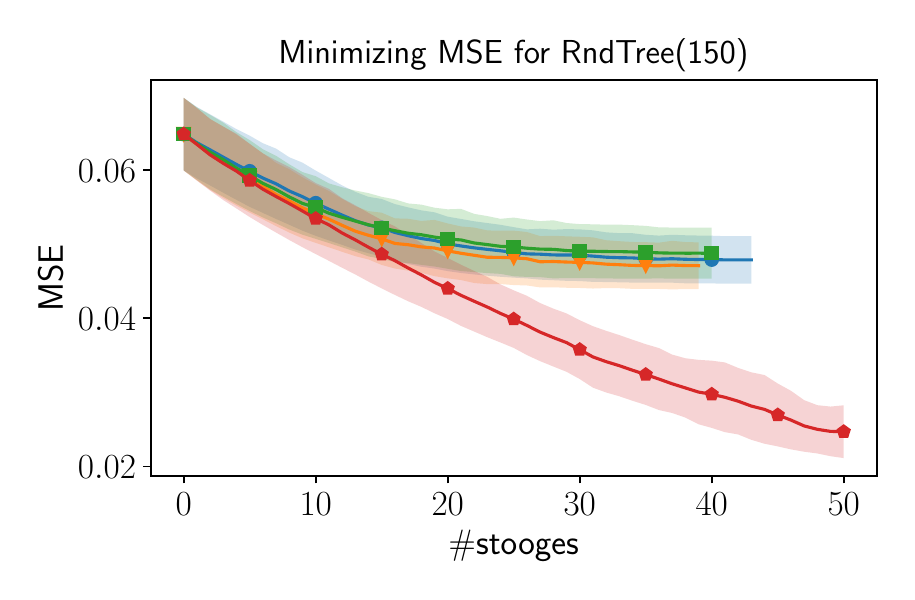}

    \includegraphics[width=\figwidth]{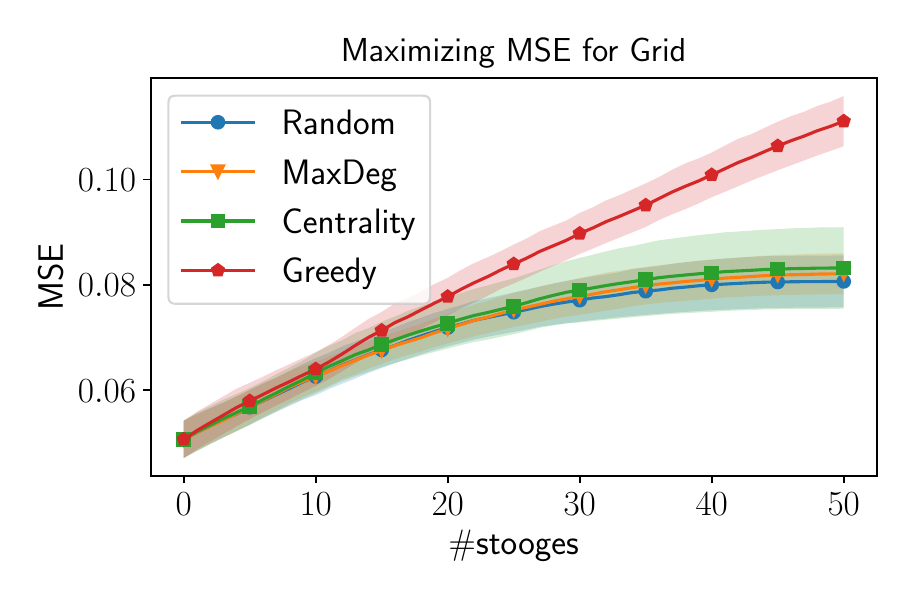}~
    \includegraphics[width=\figwidth]{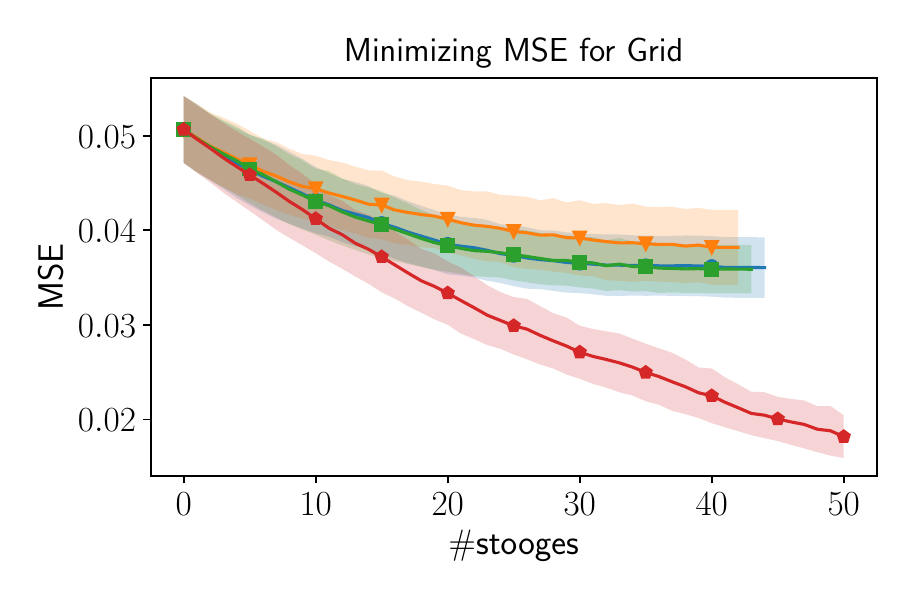}
    
    \caption{Maximization (left) and minimization (right) of the MSE on
    synthetic instances.
    We show the MSE for an increasing number of stooges, and report mean and
    standard deviation across five runs. Note that both the graph structure
    and the innate opinions are randomly generated for synthetic random graphs.}
    \label{fig:synthetic-mse}
\end{figure}

\begin{figure}
    \centering
    \includegraphics[width=\figwidth]{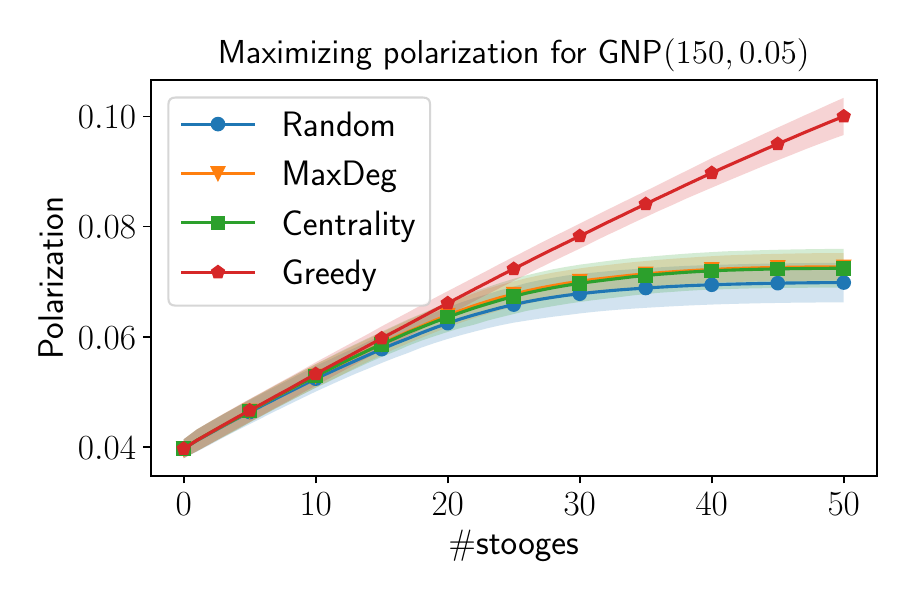}~
    \includegraphics[width=\figwidth]{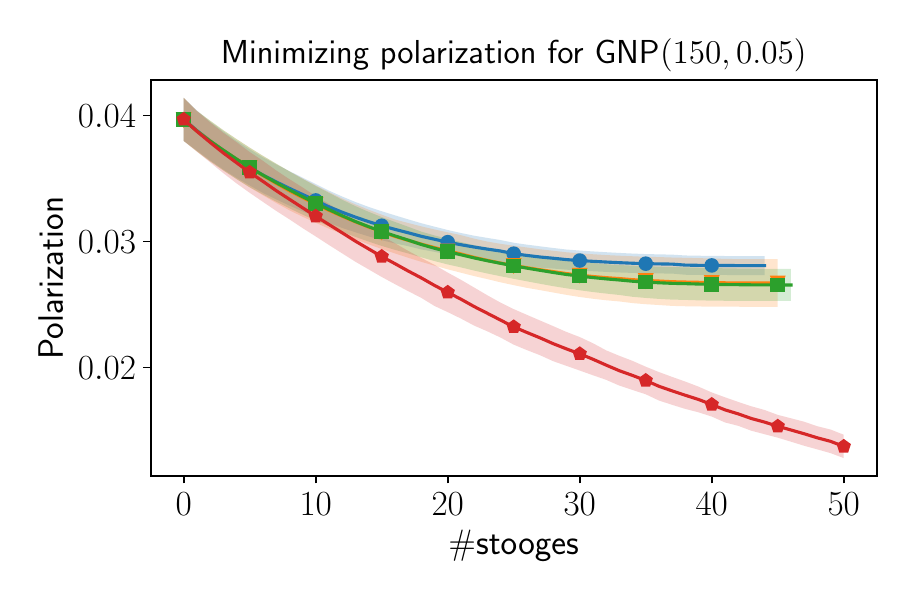}
    
    \includegraphics[width=\figwidth]{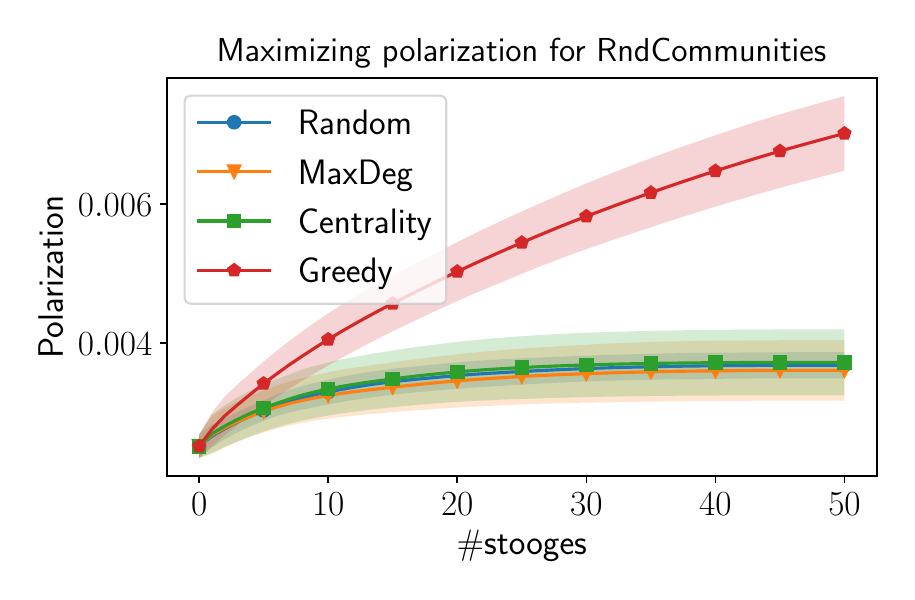}~
    \includegraphics[width=\figwidth]{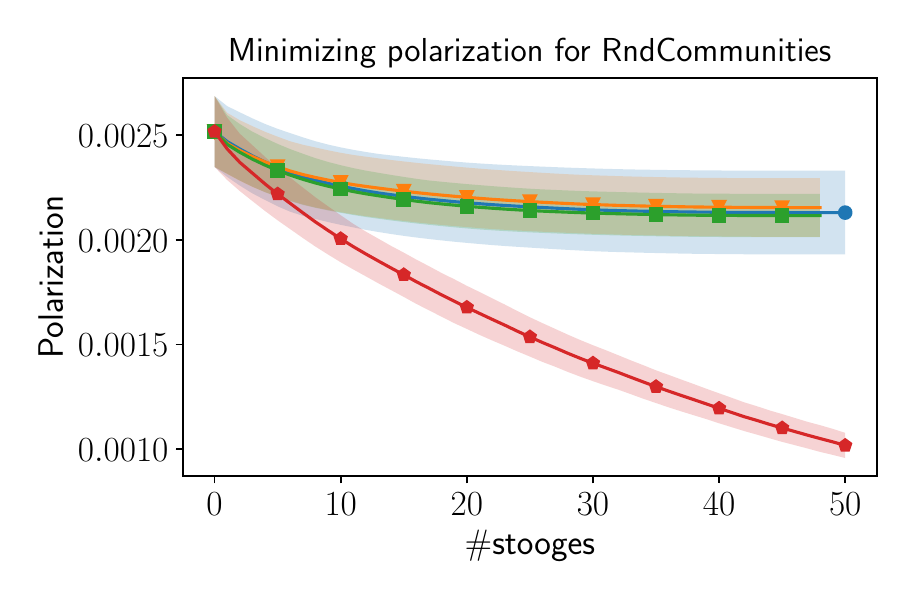}
    
    \includegraphics[width=\figwidth]{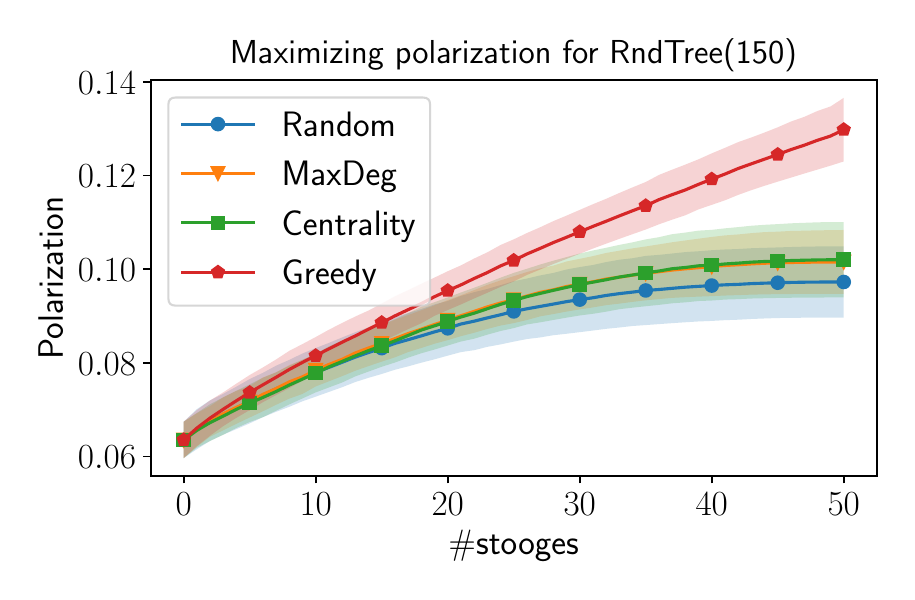}~
    \includegraphics[width=\figwidth]{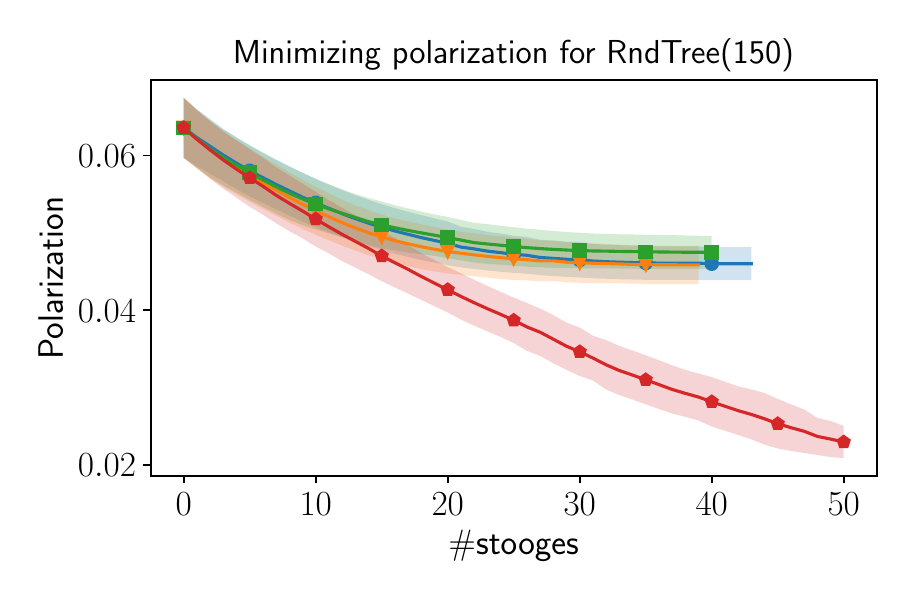}

    \includegraphics[width=\figwidth]{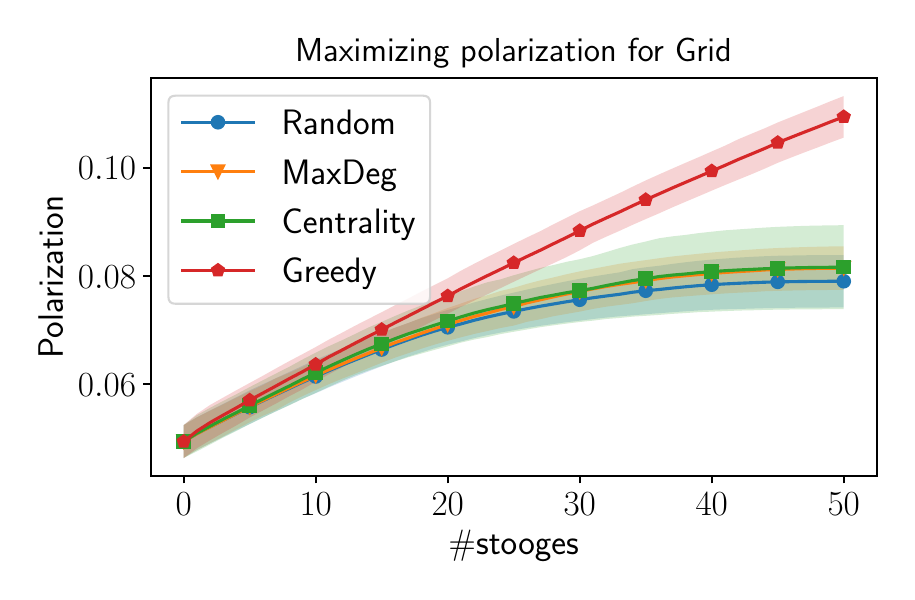}~
    \includegraphics[width=\figwidth]{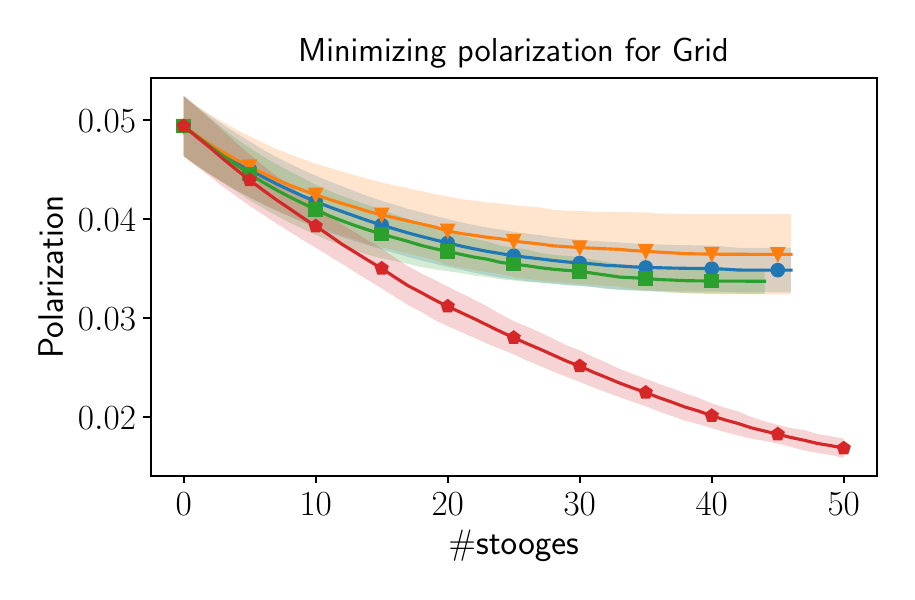}
    \caption{Maximization (left) and minimization (right) of the polarization,
    analogously to Figure~\ref{fig:synthetic-pol}.}
    \label{fig:synthetic-pol}
\end{figure}

\begin{figure}
    \centering
    \includegraphics[width=0.9\linewidth]{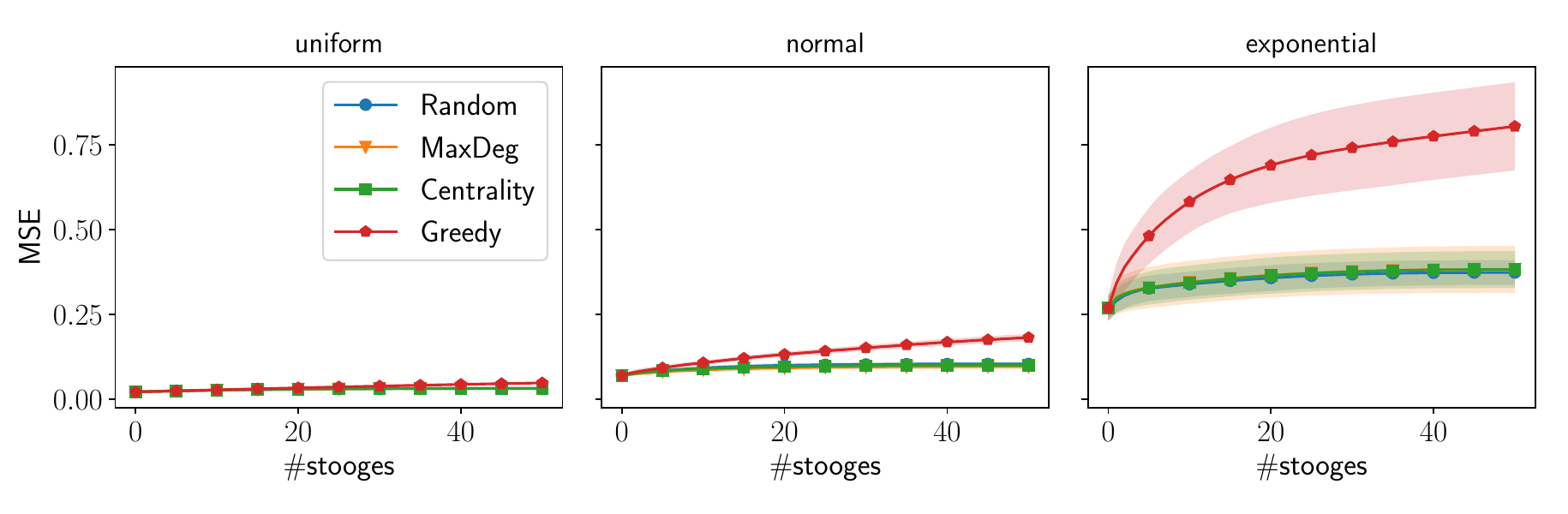}
    \figspace
    \caption{Maximizing the MSE on \RndCommunities
    when innate opinions are sampled from a uniform distribution $s_v \sim \mathcal U([0, 1])$ (left), normal distribution $s_v \sim \mathcal N(0.5, \sigma^2)$ (center), and exponential distribution $s_v \sim \mathrm {Exp}(\sigma^{-1})$ (right) where we choose
    $\sigma^2=\frac 1 {12}$ such that the all variances coincide with the
    variance of the uniform distribution.}
    \label{fig:distributions}
\end{figure}

\begin{figure}
    \centering
    
    \includegraphics[width=\figwidth]{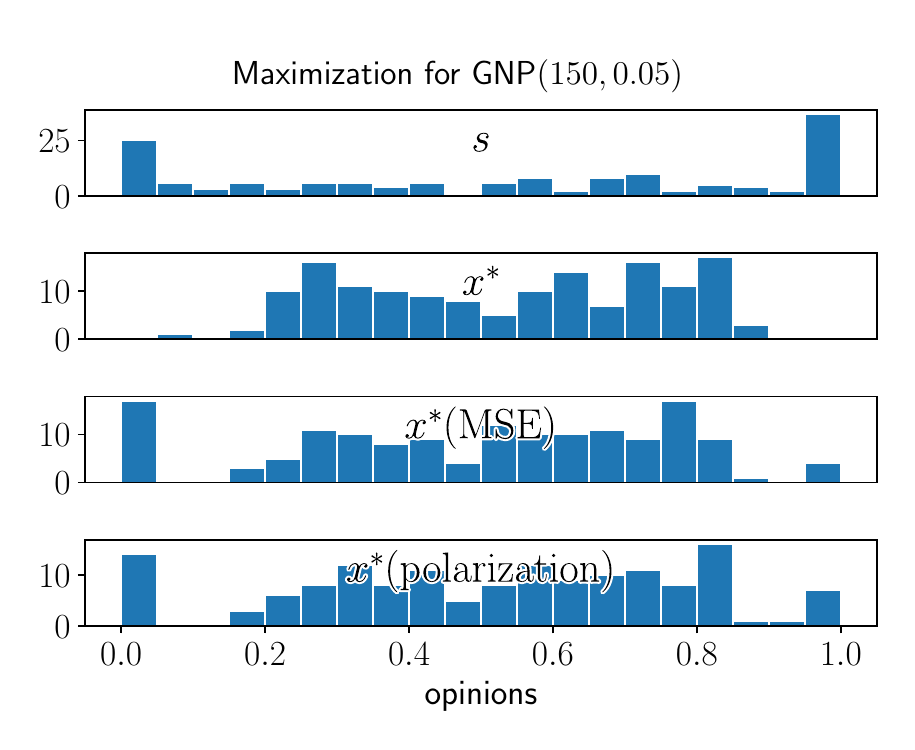}~
    \includegraphics[width=\figwidth]{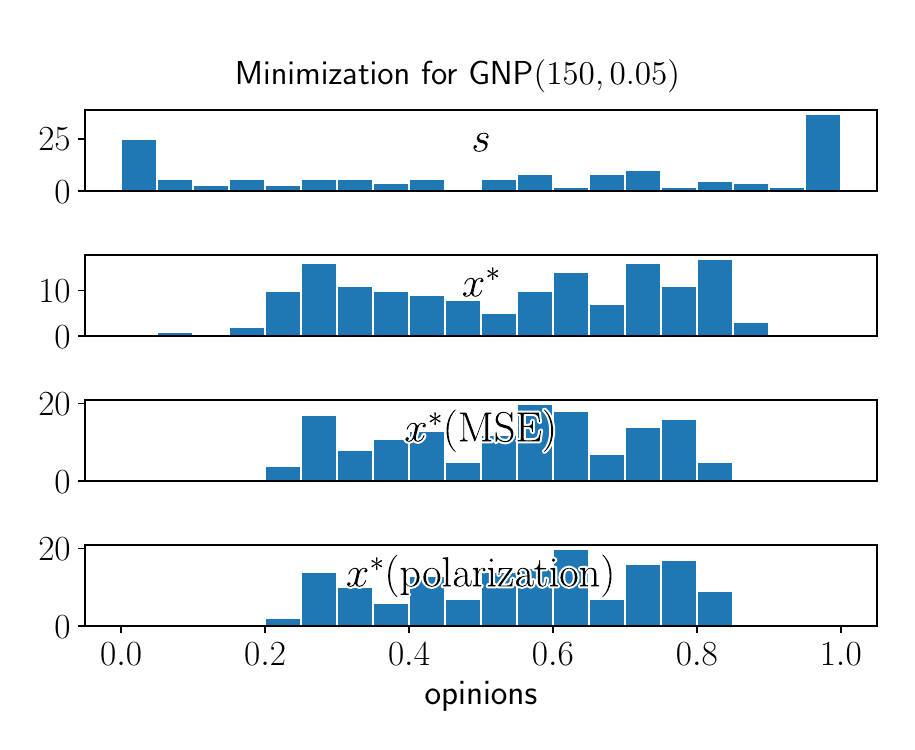}
    
    \vspace{-8pt}

    \includegraphics[width=\figwidth]{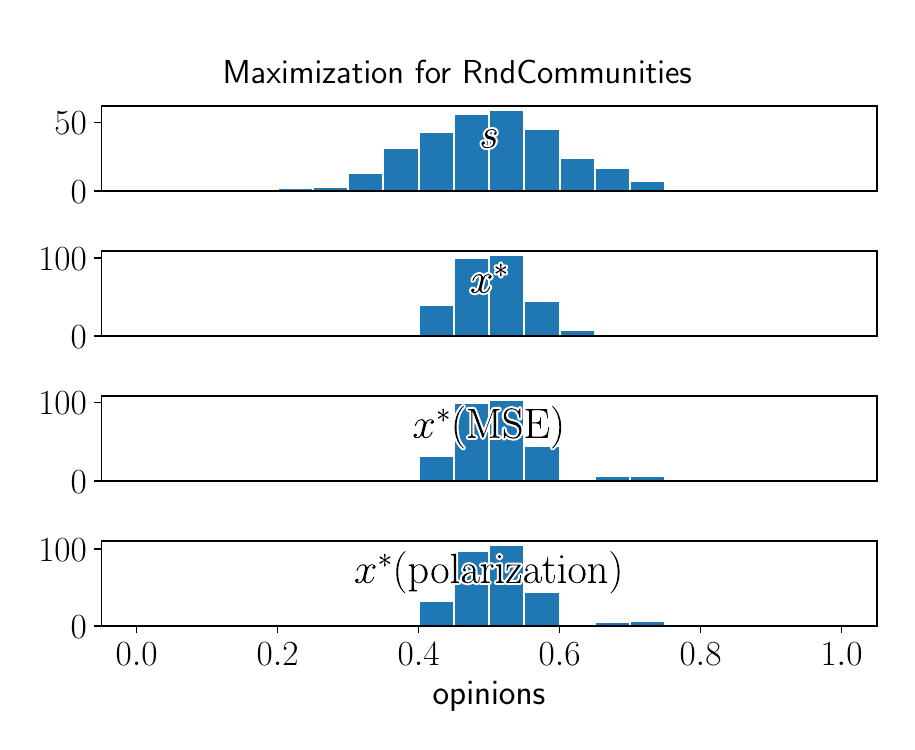}~
    \includegraphics[width=\figwidth]{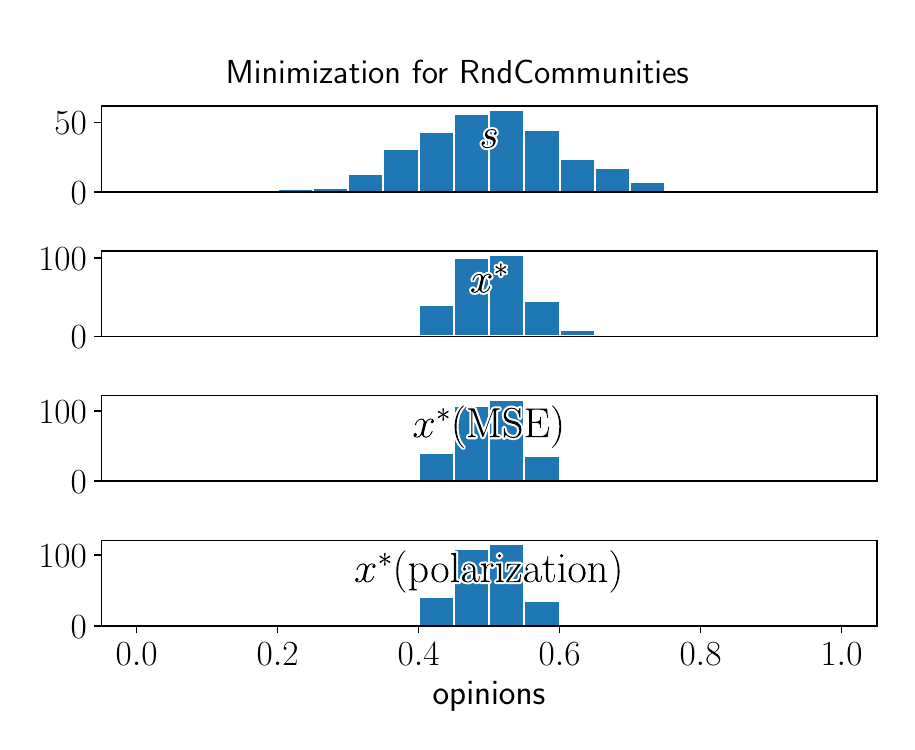}
    
    \vspace{-8pt}
    
    \includegraphics[width=\figwidth]{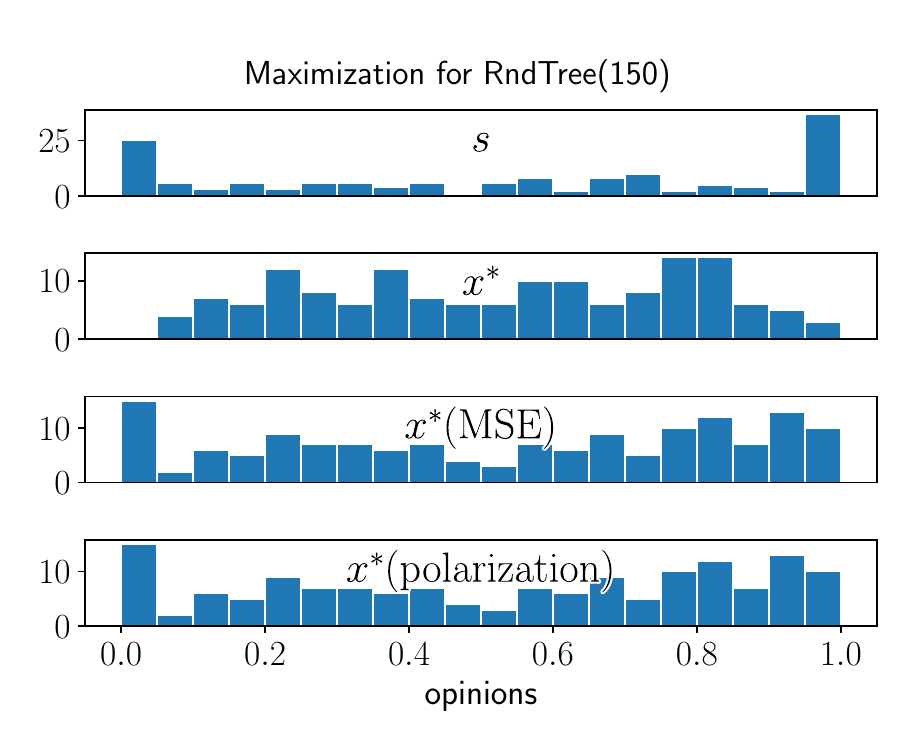}~
    \includegraphics[width=\figwidth]{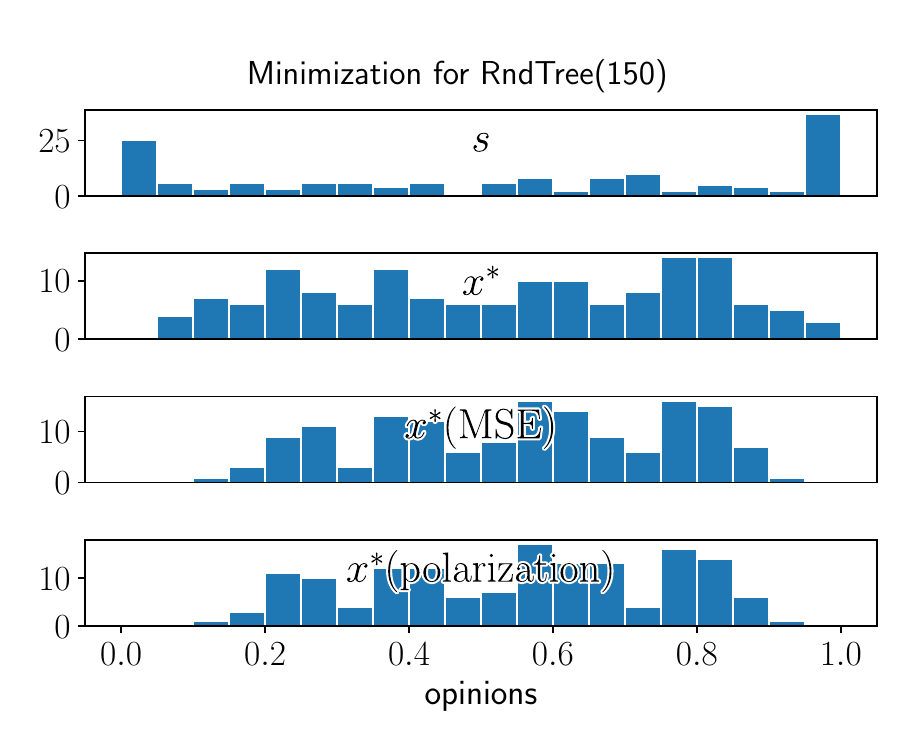}
    
    \vspace{-8pt}
    
    \includegraphics[width=\figwidth]{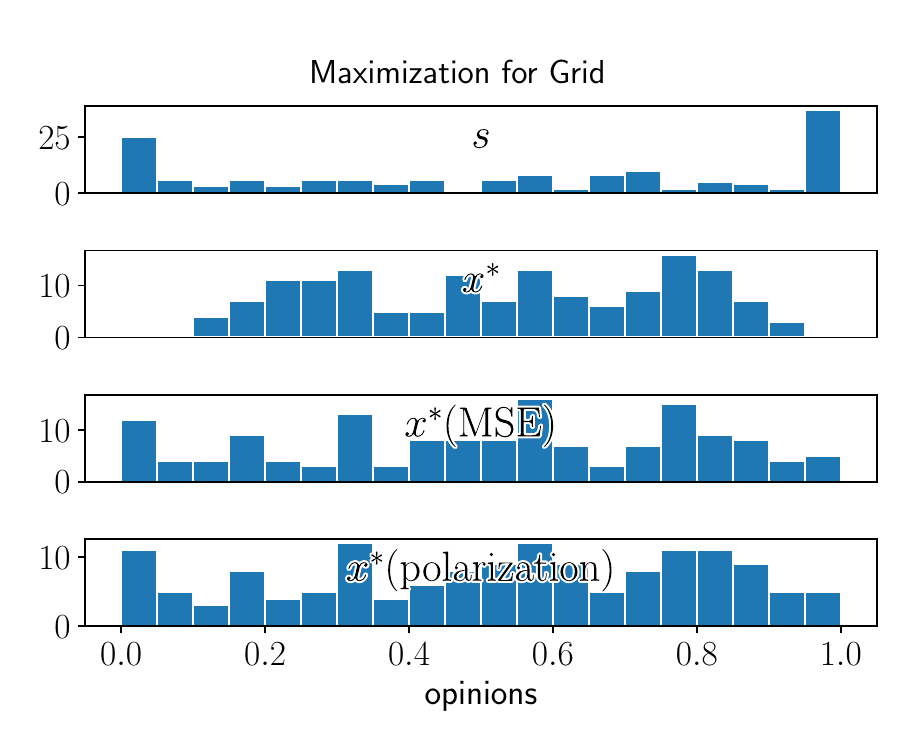}~
    \includegraphics[width=\figwidth]{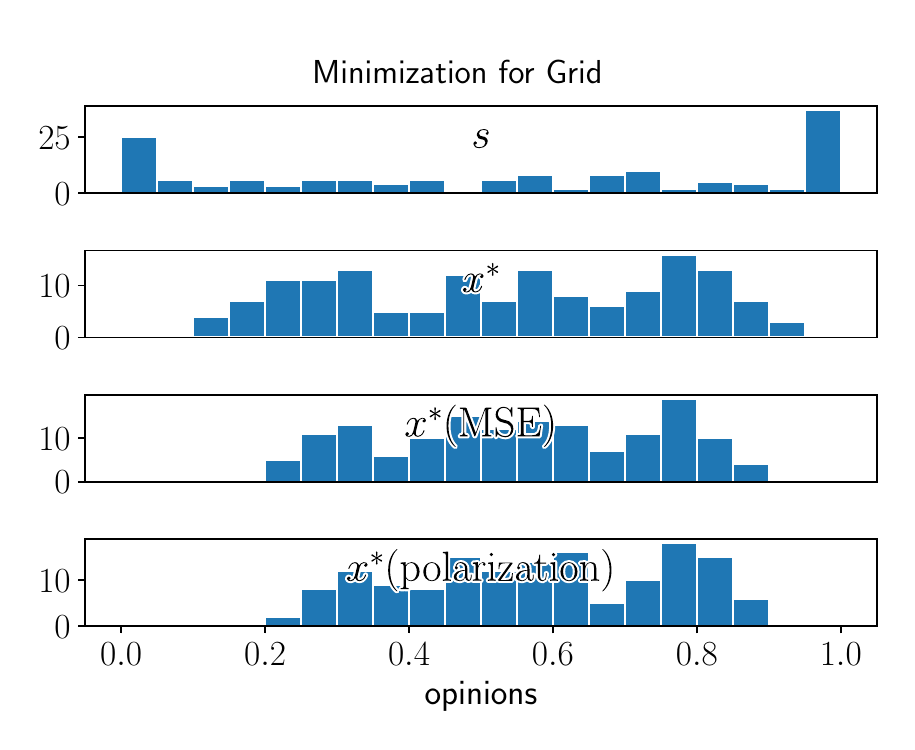}
    
    \vspace{-8pt}
    
    \caption{Distribution of the innate opinions ($s$), equilibrium opinions before
    ($x^*$) and after adding stooges for optimizing the MSE ($x^*(\mathrm{MSE})$)
    and the Polarization ($x^*(\mathrm{Polarization})$).
    We show maximization (left) and minimization (right) on a severa
    synthetic instances.}
    \label{fig:synthetic-opinions}
\end{figure}

\begin{figure}
    \centering
    \includegraphics[width=\figwidth]{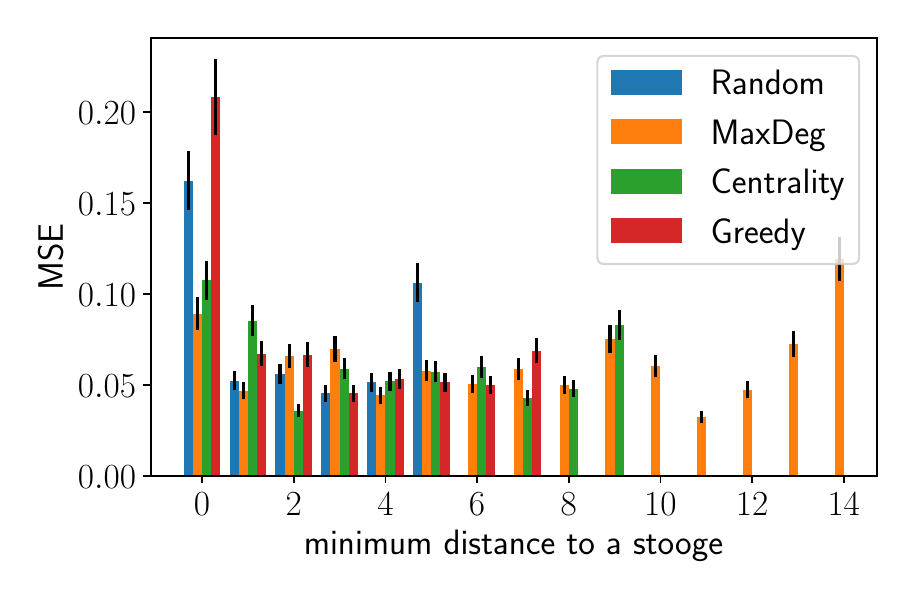}~
    \includegraphics[width=\figwidth]{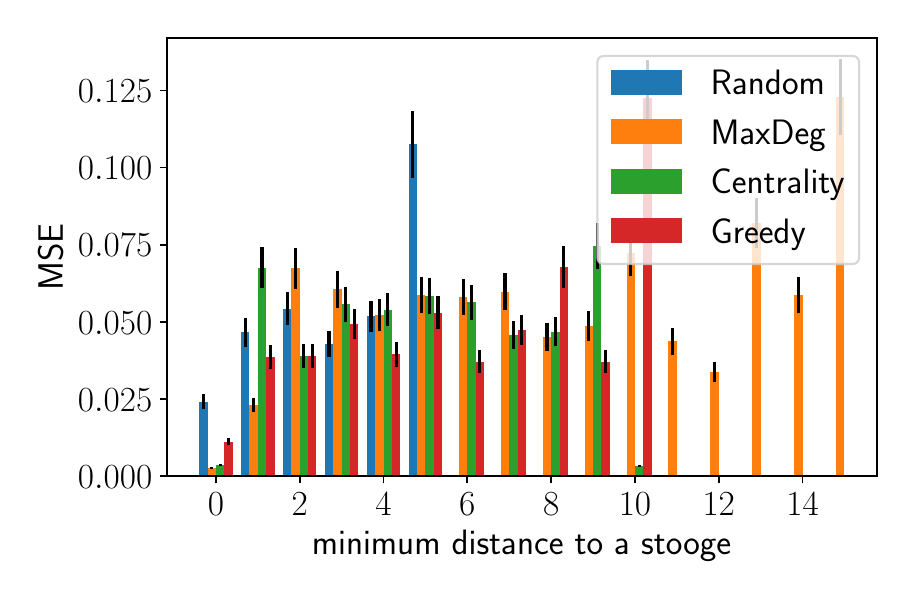}
    \figspace
    \caption{MSE of the set of nodes nodes at a given distance to the stooges, for maximization (left) and minimization (right), on a $\Grid$ instance.
    Error bars show a 10th of a standard deviation.}
    \label{fig:dists-plot}
\end{figure}

\spara{Optimization of MSE and Polarization.}
Figures~\ref{fig:synthetic-mse} and \ref{fig:synthetic-pol}
show MSE and polarization
for an increasing number
of stooges under Algorithm~\ref{alg:algo1}
($\mathsf{Greedy}$) and baselines.
Our algorithm consistently outperforms the baselines, even
the stooge selection of the baseline algorithms
does work well for a small number of stooges.
Interestingly, the $\mathsf{Random}$ baseline performs almost
as well as $\mathsf{MaxDeg}$ and $\mathsf{Centrality}$ on our
synthetic instances.
In Figure~\ref{fig:distributions}, we show the behavior of
our algorithms when the innate opinions are sampled
from different distributions.
Even though the variance of the innate opinions
are kept equal, we can see that more extreme
opinions lead to a higher MSE in the equilibrium
opinions and make the network more susceptible
for an increase in MSE by adding stooges
via \textsf{Greedy}.

Figure~\ref{fig:synthetic-opinions} shows the
innate and equilibrium before
and after adding stooges.
This shows how the maximization and minimization objectives clearly
differ in bringing opinions together or farther apart.
Optimizing for the MSE and the polarization results
in approximately the same distribution over equilibrium
opinions.

\spara{Stooge Location.}
We run experiments to showcase where
our approach and baselines selects stooges.
We run \textsf{Greedy} to select
10 stooges.
In Figure~\ref{fig:dists-plot}, we determine
the minimum distance of a vertex to any stooge,
and group the vertices according to this distance.
We show the MSE among each of these groups.
We observe that for maximization, our greedy approach
selects a set of stooges (distance 0) that itself
experiences a high MSE. This is reversed for minimization.
Furthermore, the MSE remains relatively constant for
groups of other distances, but this is not the
case for the baseline methods.

\begin{table}
    \centering
    \caption{We compute the set of stooges for maximization
    and minimization of the MSE, and report the Jaccard similarity
    between both sets across five runs.}
    \label{tab:isect}
    \medskip
    \begin{tabular}{rccccc}
        \toprule
         Graph type &
            $\GNP(150, 0.05)$ &
            $\RndCommunities$ &
            $\RndTree$(150) &
            $\Grid$ &
            $\Star$ \\
        \midrule
         Jaccard sim. &
            $0.34 \pm 0.07$ &
            $0.83 \pm 0.10$ &
            $0.13 \pm 0.05$ &
            $0.27 \pm 0.05$ &
            $0.03$ \\
        \bottomrule
    \end{tabular}
\end{table}

\spara{Comparing Minimization and Maximization.}
Table~\ref{tab:isect} shows the Jaccard similarity
between the set of stooges selected for the
maximization and minimization objectives.
The size of the intersection varies
with the network structure, as is apparent from
the difference in intersection sizes.
The maximal intersection is achieved
for $\RndCommunities$, indicating that
community structures contain a set
of influential nodes responsible for
controlling the wisdom of crowds.

\begin{figure}
    \centering
    \includegraphics[width=\figwidth]{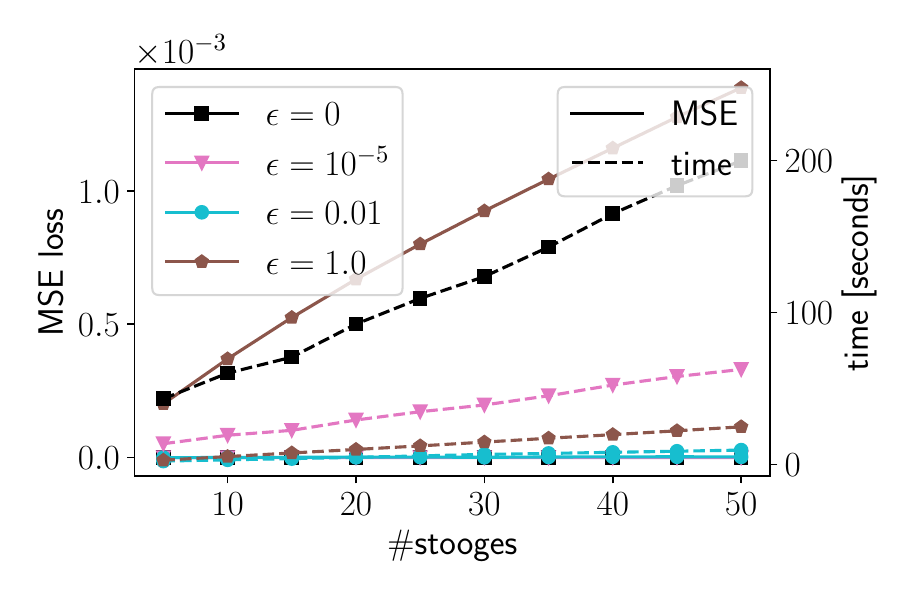}~
    \includegraphics[width=\figwidth]{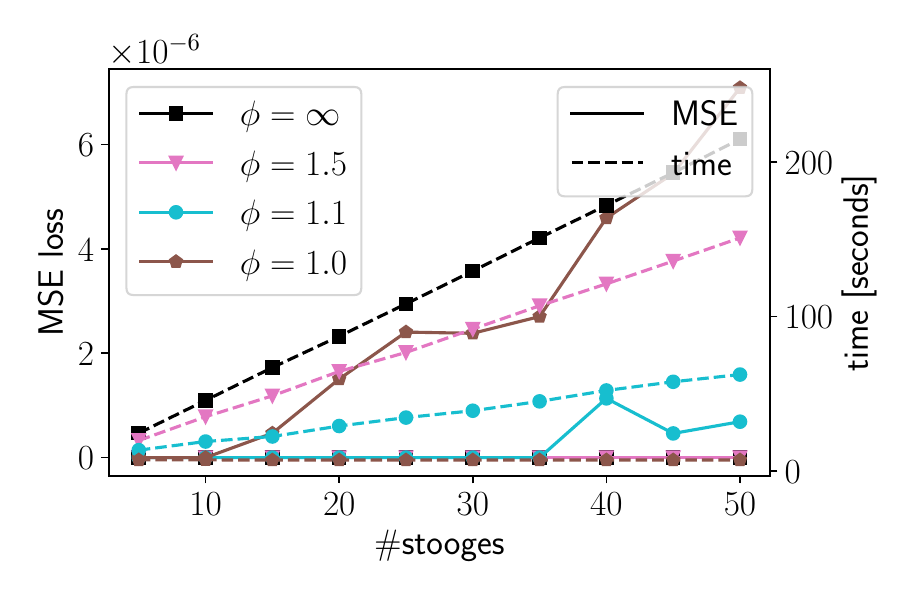}
    \figspace
    \caption{Performance of Algorithm~\ref{alg:algo1}
    for different choices of $\epsilon$ and $\phi$
    for an increasing number of stooges, on $\RndTree(1000)$ instances.
    For the left
    plot, we fix $\phi=1.1$ and vary $\epsilon$. For the right
    plot, we fix $\epsilon=10^{-5}$ and vary $\phi$.
    We show the loss in MSE compared to the
    optimal parameter choices $\epsilon=0$ (left)
    and $\phi=\infty$ (right). }
    \label{fig:algo-eps-phi}
\end{figure}

\begin{figure}
    \centering
    \includegraphics[width=\figwidth]{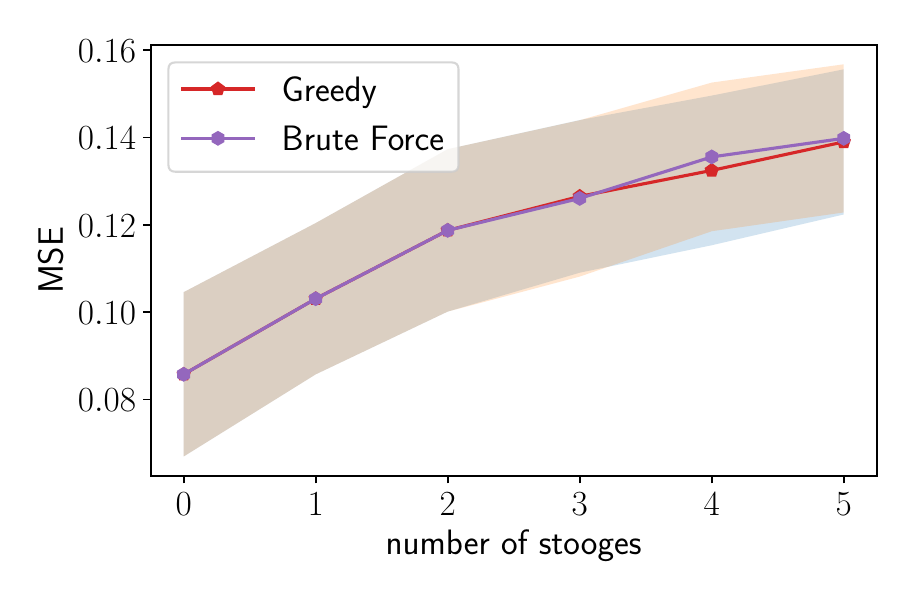}
    \caption{Comparing \textsf{Greedy} with a brute force search
    on $\GNP(150, 0.05)$ instances.}
    \label{fig:brute-force}
\end{figure}

\spara{Parameters of Algorithm~\ref{alg:algo1}.}
We now want to experimentally
motivate the parameter choices
of \textsf{Greedy}.
Figure~\ref{fig:algo-eps-phi} contains
the MSE and runtimes achieved
by Algorithm~\ref{alg:algo1}
for different parameter choices.
We vary the accuracy $\epsilon$ of
the equilibrium opinions and the
slack $\phi$ for the lazy evaluation.
We observe that with
$\epsilon=10^{-5}$ and $\phi=1.1$,
we are able to extract almost as much
values as for a better accuracy
$\epsilon$ or larger slack $\phi$,
while using only a fraction of the
time.
We further investigate the approximation
guarantee obtained via these parameter
choices by comparing with a brute force
approach in
Figure~\ref{fig:brute-force}.
Our greedy approach is able
to extract almost all of the value
of the brute force search.

\newpage

\subsection{Additional Experimental Results on Real-World Data}

\begin{table}[b]
    \centering
    \caption{We report the relative increase and
    decrease for optimizting MSE and polarization
    on real-world instances for $50$ stooges.}
    \medskip
    \label{tab:rel-change}
    \setlength{\tabcolsep}{4pt}
    \begin{tabular}{r|rr|rr}
        \toprule
        & \multicolumn{2}{|l|}{MSE} & \multicolumn{2}{l}{Polarization} \\
        Dataset & Maximization & Minimization & Maximization & Minimization \\
        \midrule
        \textsf{war}                   & $207.80\%$ & $50.43\%$ & $206.16\%$ & $50.62\%$ \\
        \textsf{vax}                   & $43.25\%$ & $19.96\%$ & $28.22\%$ & $2.32\%$ \\
        \textsf{vaxnovax-retweet}      & $0.84\%$ & $1.13\%$ & $0.84\%$ & $1.13\%$ \\
        \textsf{leadersdebate-follow}  & $6.90\%$ & $6.32\%$ & $6.89\%$ & $6.10\%$ \\
        \textsf{leadersdebate-retweet} & $5.82\%$ & $7.93\%$ & $5.82\%$ & $7.93\%$ \\
        \textsf{russia-march-follow}   & $5.40\%$ & $5.22\%$ & $5.40\%$ & $5.24\%$ \\
        \textsf{russia-march-retweet}  & $2.58\%$ & $3.26\%$ & $2.58\%$ & $3.26\%$ \\
        \textsf{baltimore-follow}      & $12.96\%$ & $10.99\%$ & $12.39\%$ & $10.73\%$ \\
        \textsf{baltimore-retweet}     & $3.18\%$ & $5.03\%$ & $3.18\%$ & $5.03\%$ \\
        \textsf{beefban-follow}        & $11.31\%$ & $9.61\%$ & $11.20\%$ & $9.47\%$ \\
        \textsf{beefban-retweet}       & $3.98\%$ & $5.00\%$ & $3.98\%$ & $5.00\%$ \\
        \textsf{gunsense-follow}       & $6.48\%$ & $4.64\%$ & $6.09\%$ & $4.64\%$ \\
        \textsf{gunsense-retweet}      & $4.03\%$ & $5.82\%$ & $4.03\%$ & $5.82\%$ \\
        \bottomrule
    \end{tabular}
\end{table}

We now showcase
experimental results on
the real-world instance
\textsf{War} and the
instances of Garimella
et al.~\cite{garimella18}.
Analogous to Figures~\ref{fig:vax}, and
\ref{fig:vax-pol}
we show the optimization
of the MSE and the polarization.
%
Table~\ref{tab:rel-change} summarizes the
relative increase and decrease for both objectives
on 50 stooges.

\begin{figure}
    \centering
    
    \includegraphics[width=\figwidth]{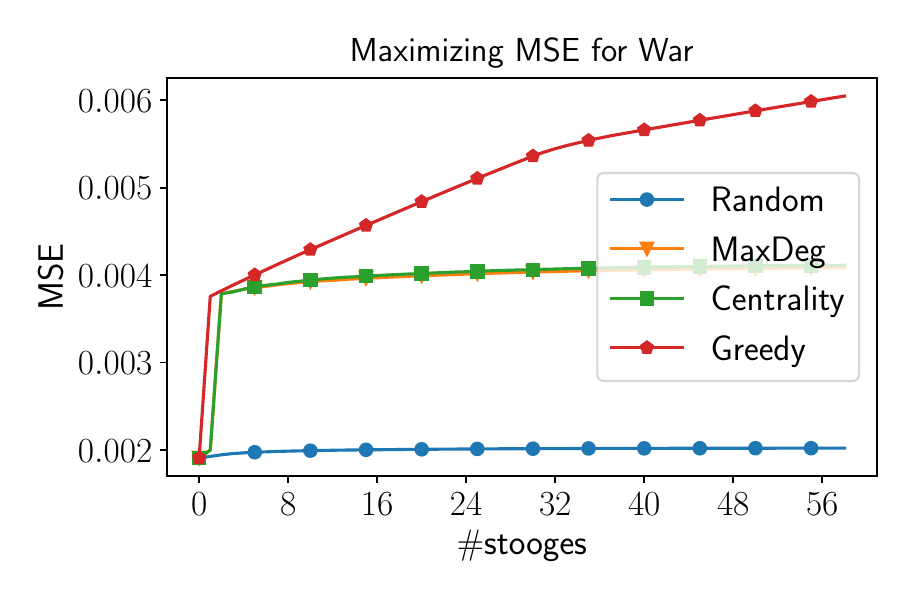}~
    \includegraphics[width=\figwidth]{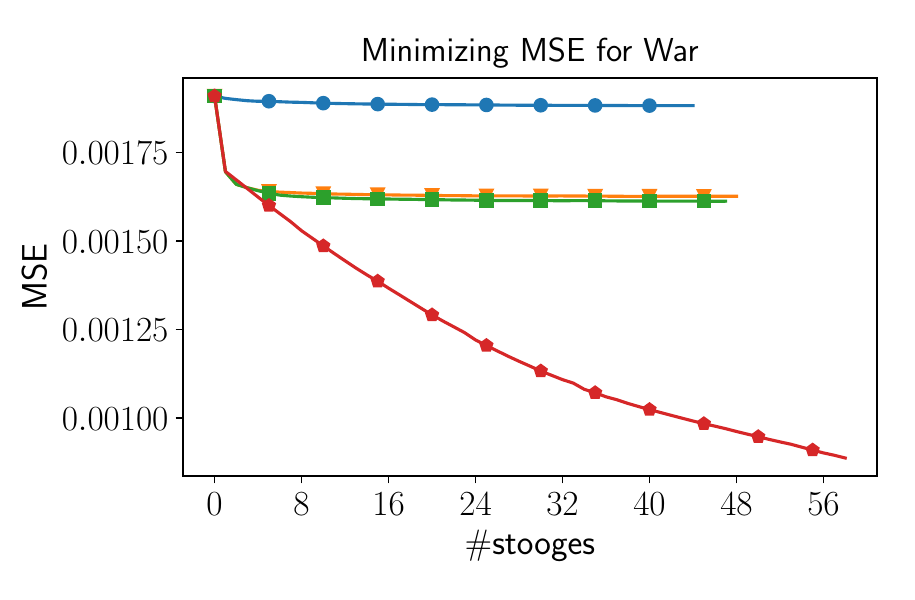}
    
    \includegraphics[width=\figwidth]{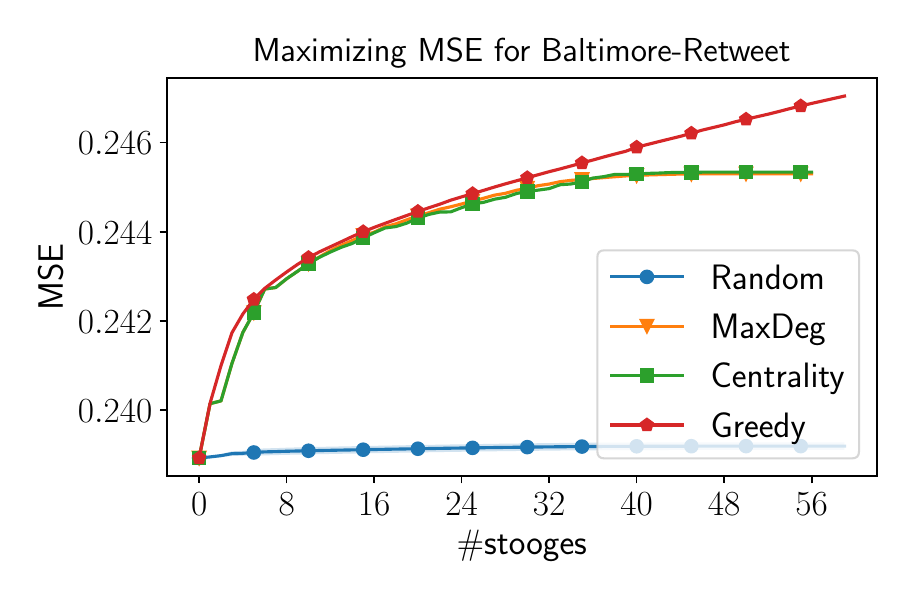}~
    \includegraphics[width=\figwidth]{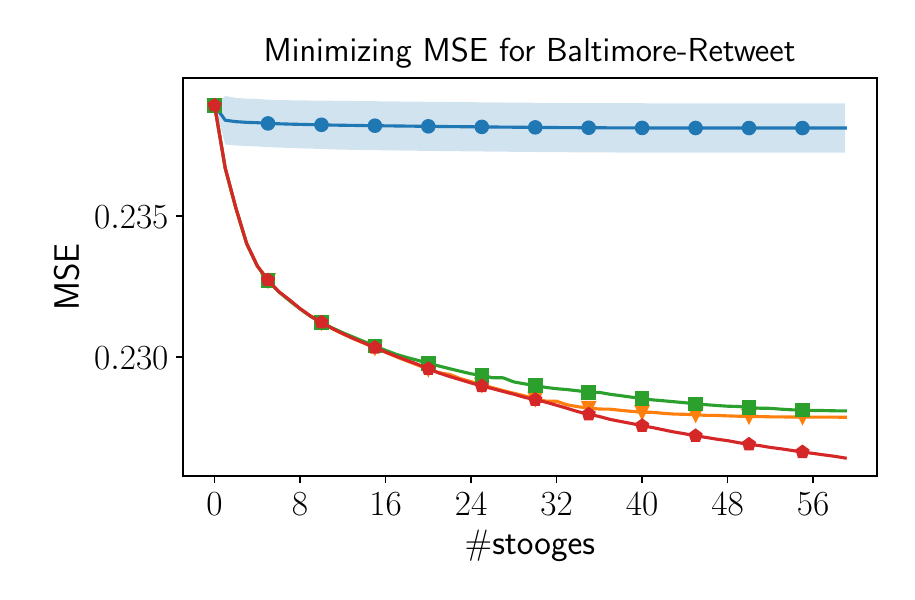}
    
    \includegraphics[width=\figwidth]{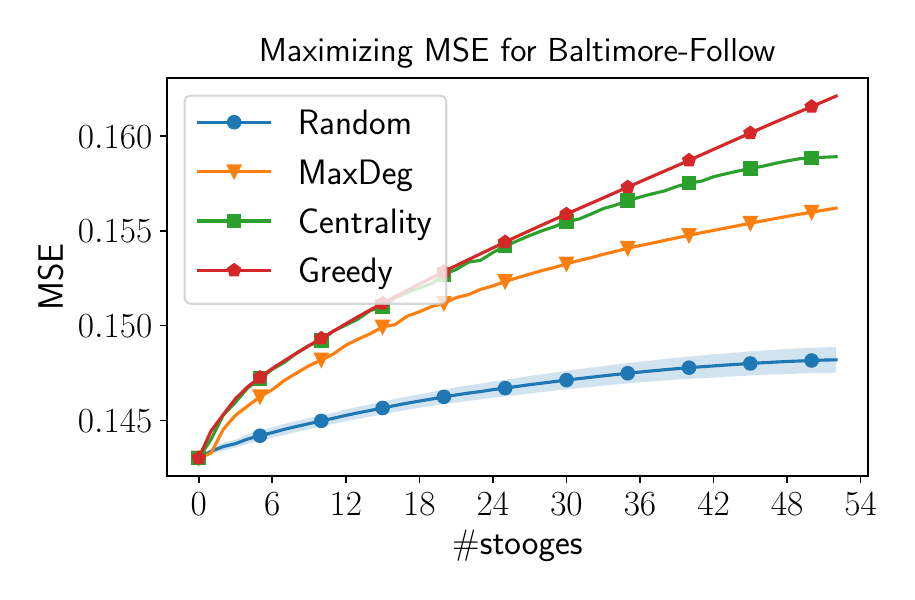}~
    \includegraphics[width=\figwidth]{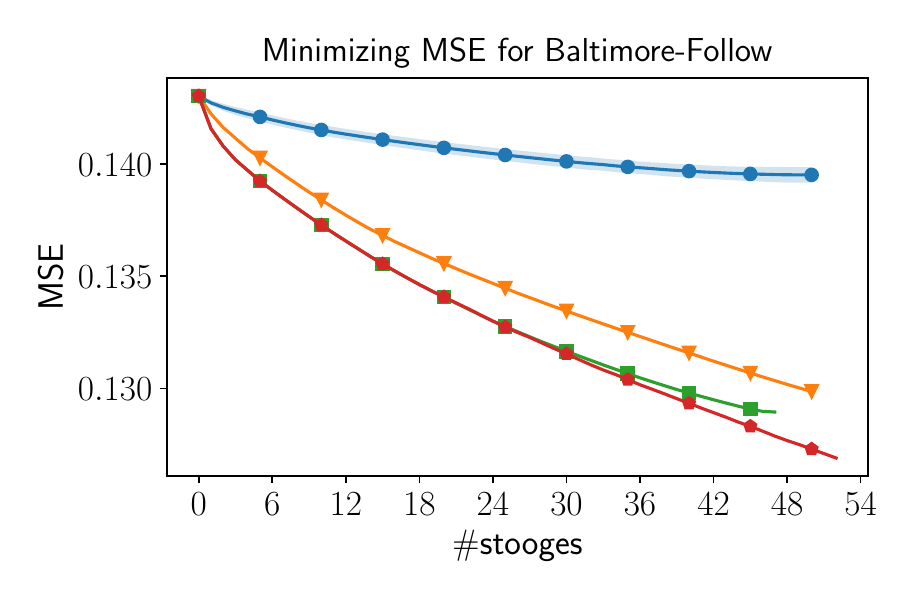}
    
    \caption{Maximization (left) and minimization (right) of the MSE on various
    real-world datasets, analogous to Figure~\ref{fig:synthetic-mse}}
    \label{fig:real-world-1-mse}
\end{figure}

\begin{figure}
    \centering
    
    \includegraphics[width=\figwidth]{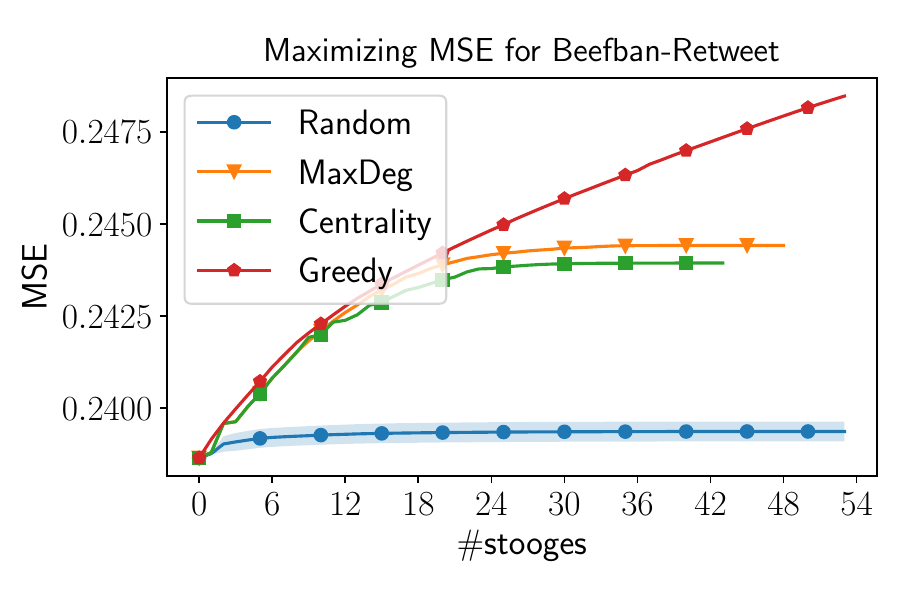}~
    \includegraphics[width=\figwidth]{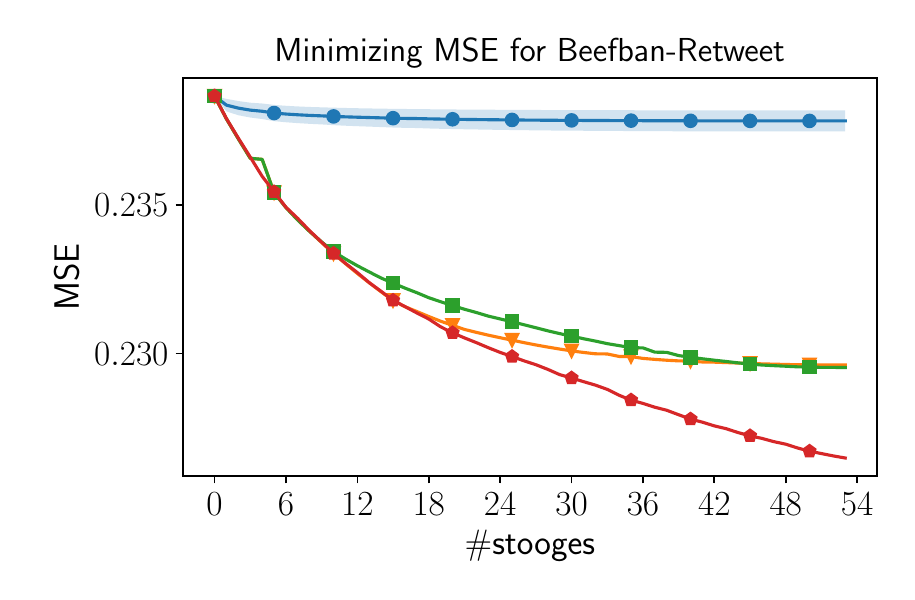}
    
    \includegraphics[width=\figwidth]{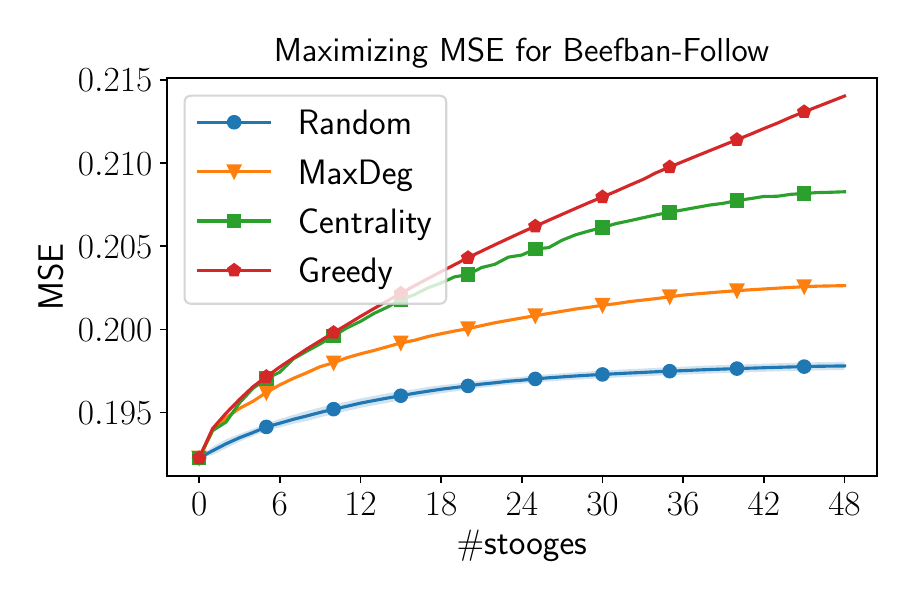}~
    \includegraphics[width=\figwidth]{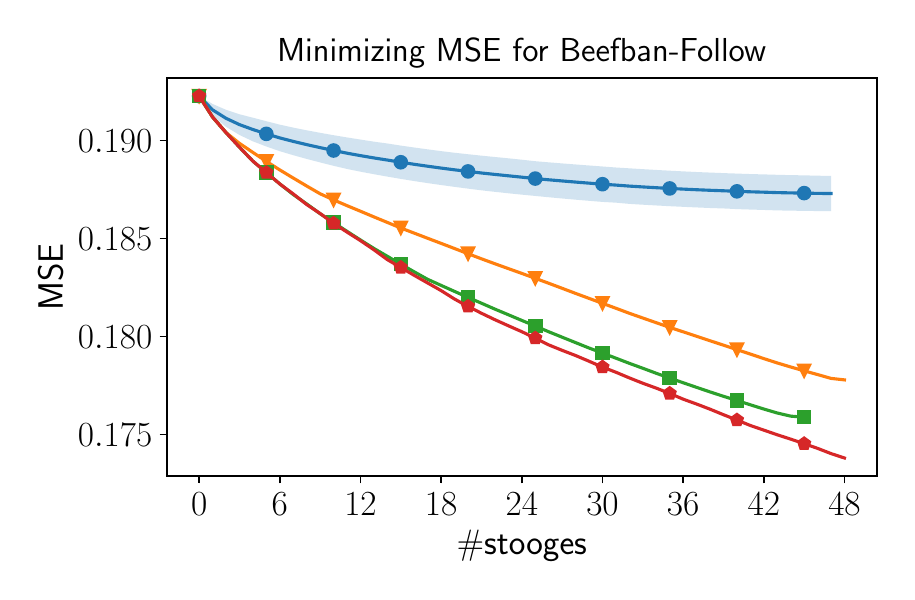}
    
    \includegraphics[width=\figwidth]{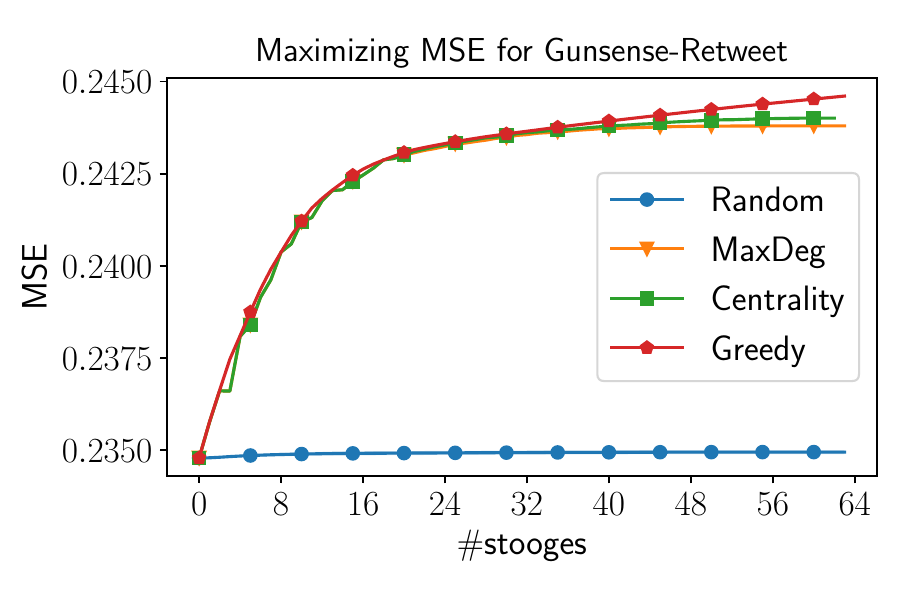}~
    \includegraphics[width=\figwidth]{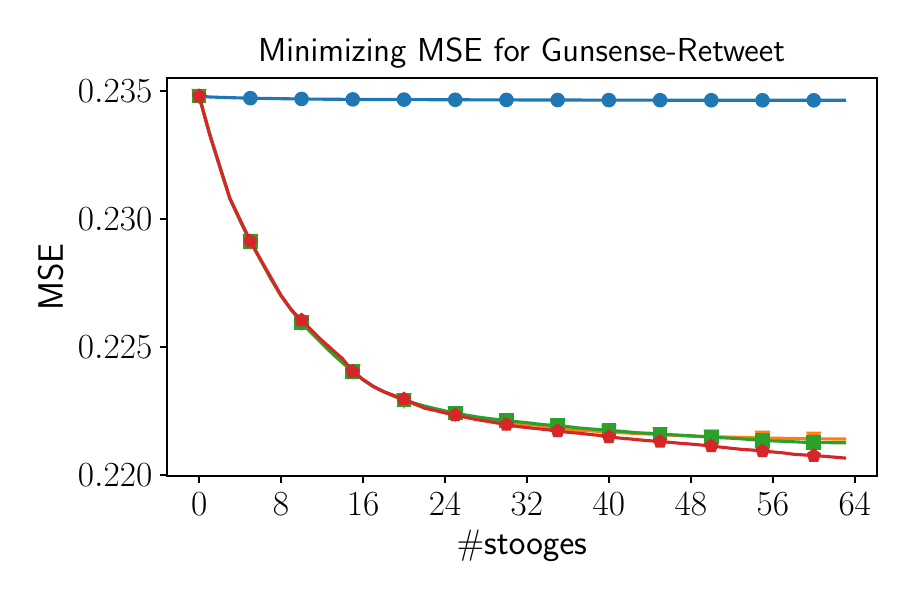}

    \includegraphics[width=\figwidth]{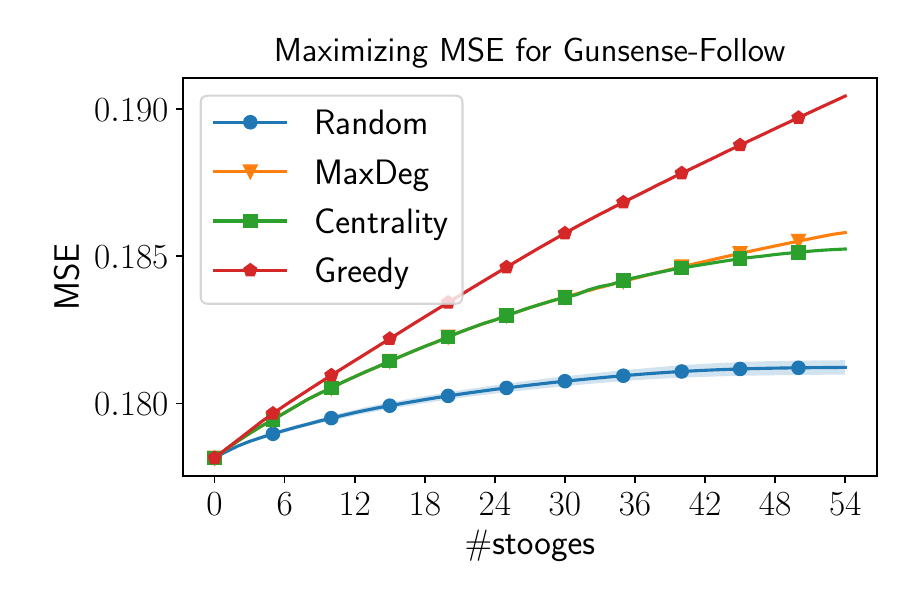}~
    \includegraphics[width=\figwidth]{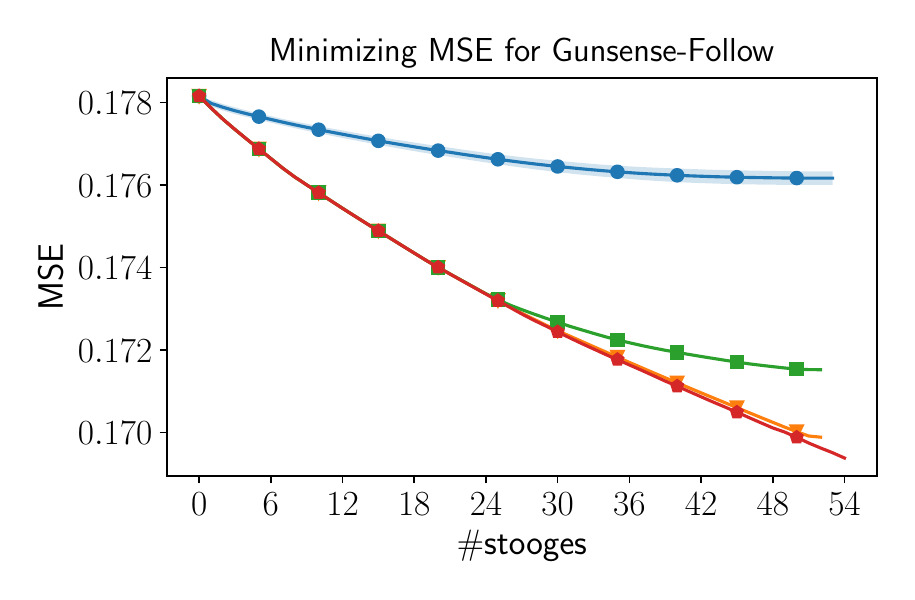}
    
    \caption{Maximization (left) and minimization (right) of the MSE on various
    real-world datasets, analogous to Figure~\ref{fig:synthetic-mse}}
    \label{fig:real-world-2-mse}
\end{figure}

\begin{figure}
    \centering

    \includegraphics[width=\figwidth]{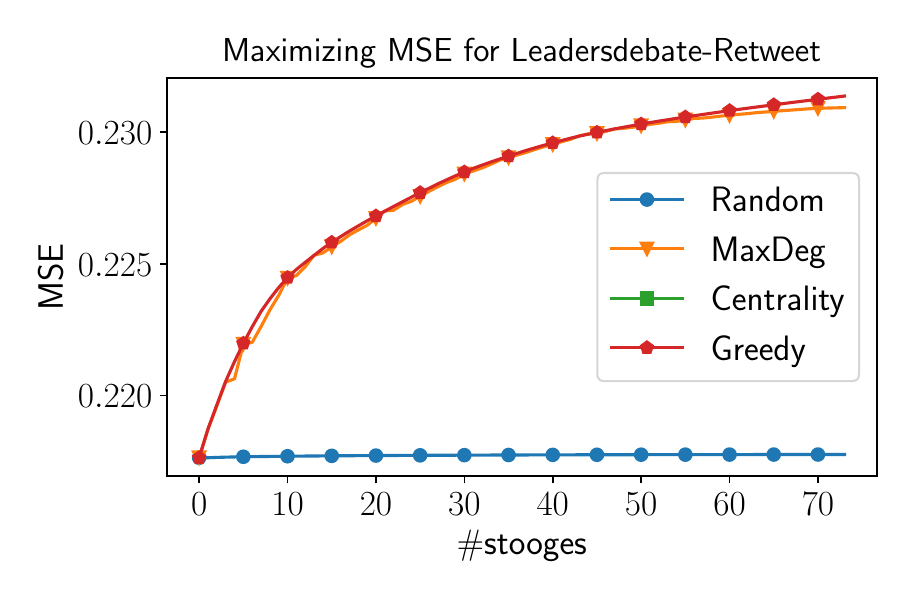}~
    \includegraphics[width=\figwidth]{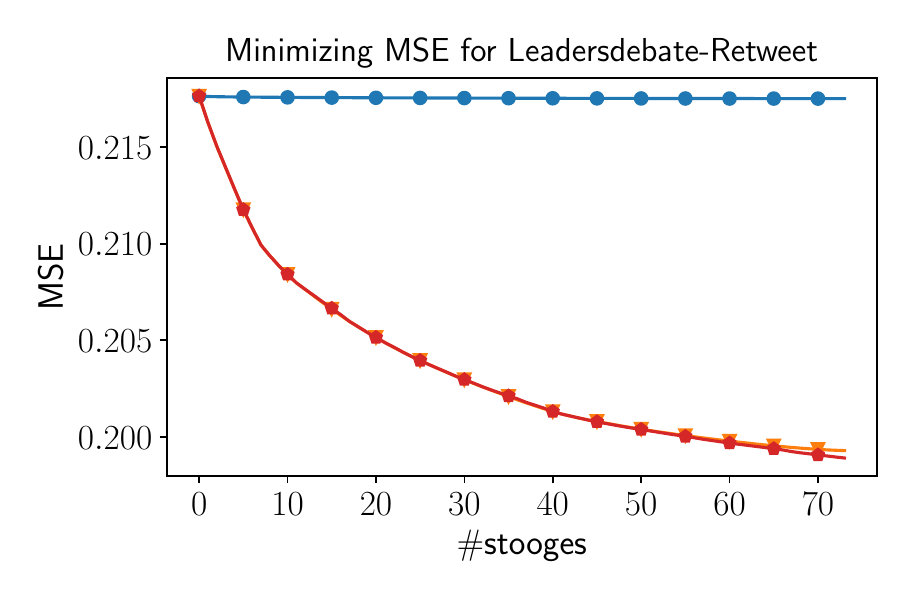}
    
    \includegraphics[width=\figwidth]{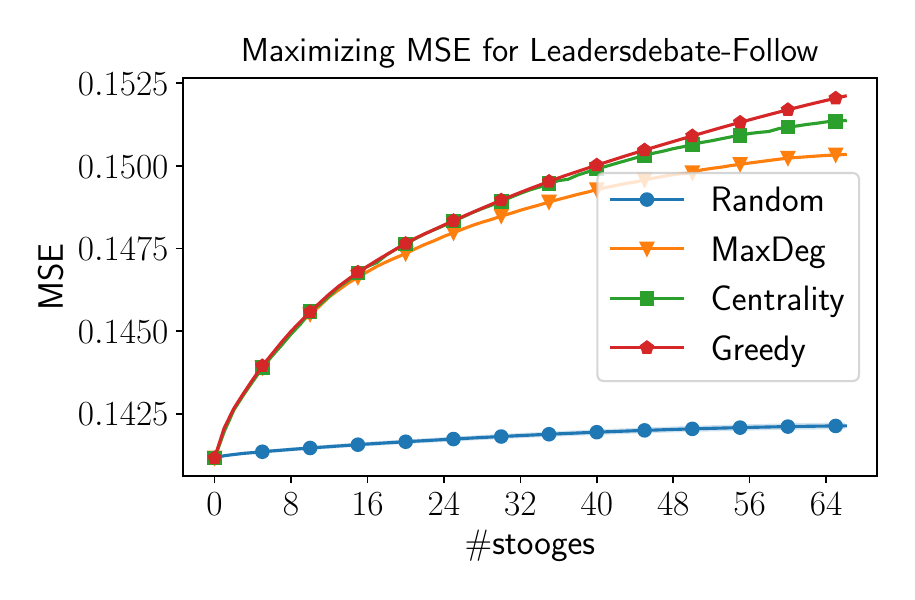}~
    \includegraphics[width=\figwidth]{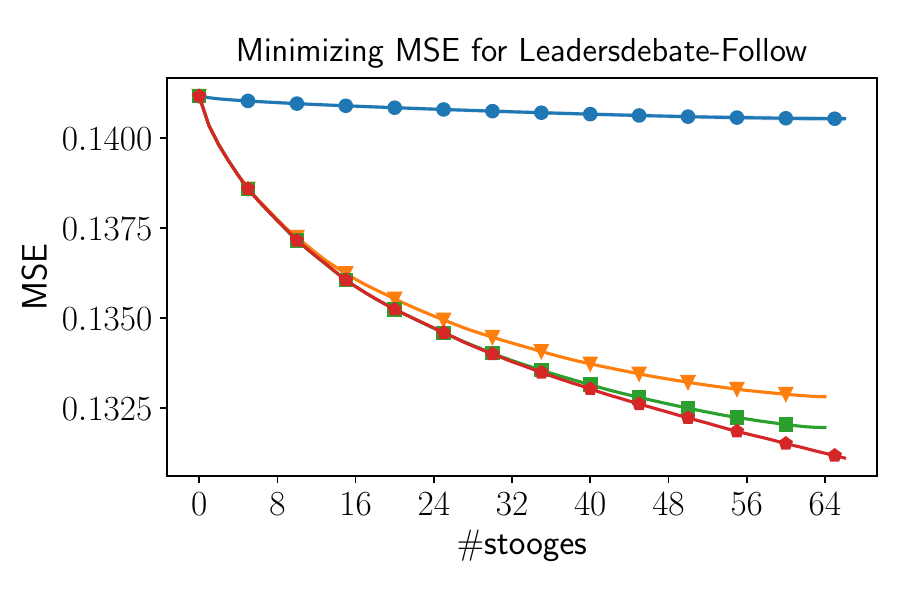}
    
    \includegraphics[width=\figwidth]{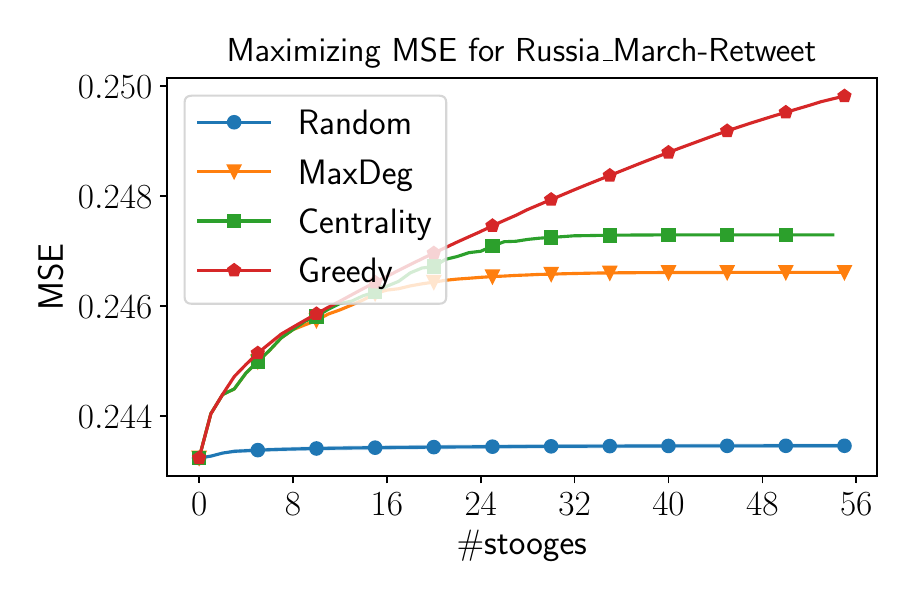}~
    \includegraphics[width=\figwidth]{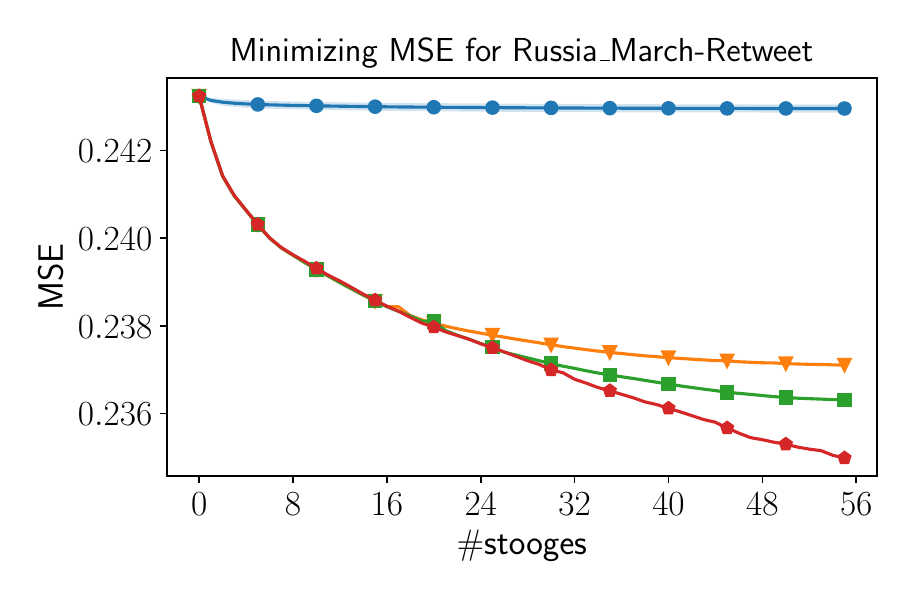}

    \includegraphics[width=\figwidth]{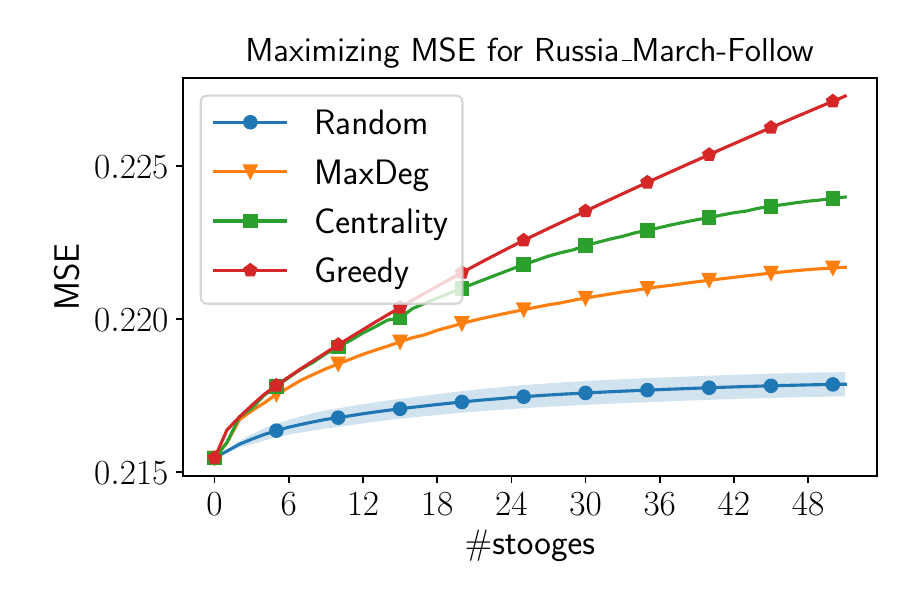}~
    \includegraphics[width=\figwidth]{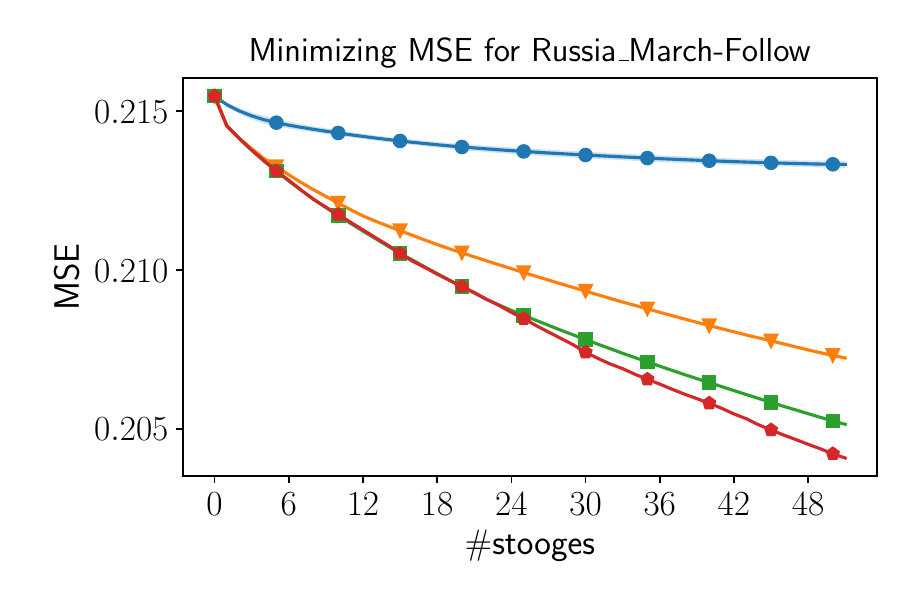}
    
    \caption{Maximization (left) and minimization (right) of the MSE on various
    real-world datasets, analogous to Figure~\ref{fig:synthetic-mse}}
    \label{fig:real-world-3-mse}
\end{figure}

\begin{figure}
    \centering
    
    \includegraphics[width=\figwidth]{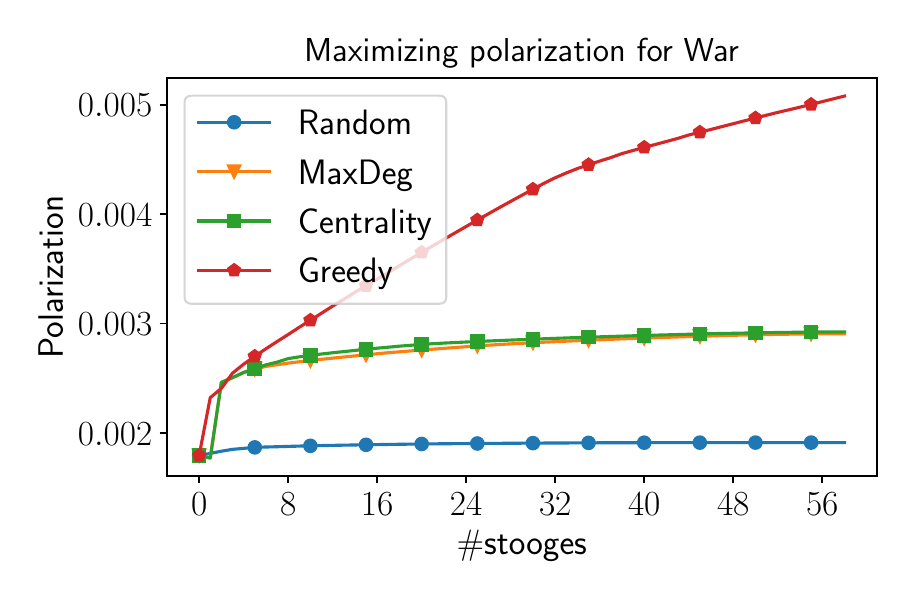}~
    \includegraphics[width=\figwidth]{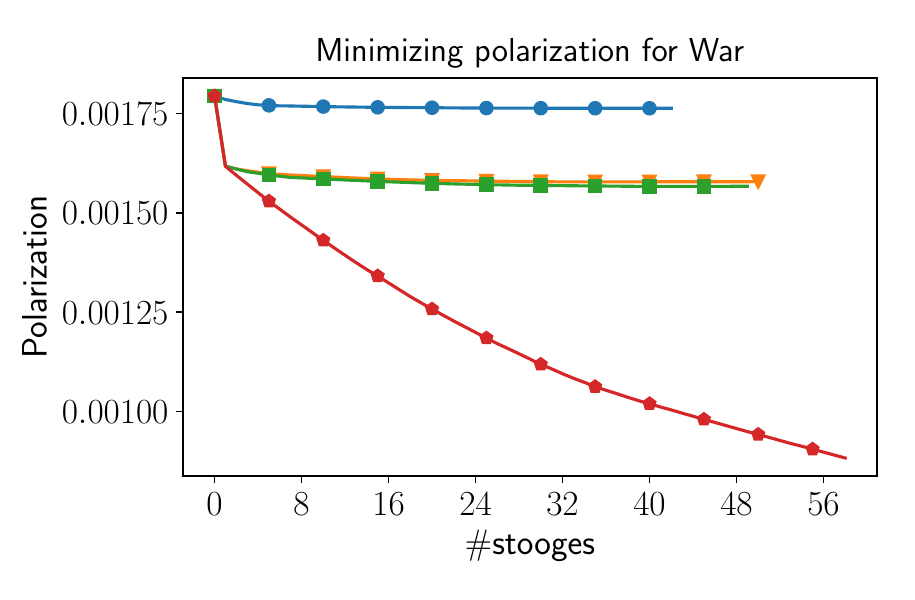}
    
    \includegraphics[width=\figwidth]{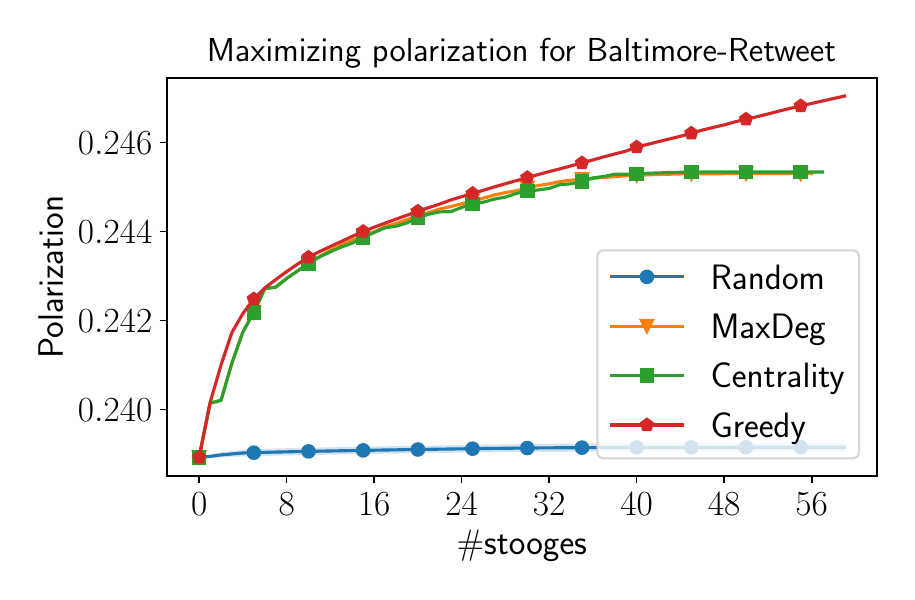}~
    \includegraphics[width=\figwidth]{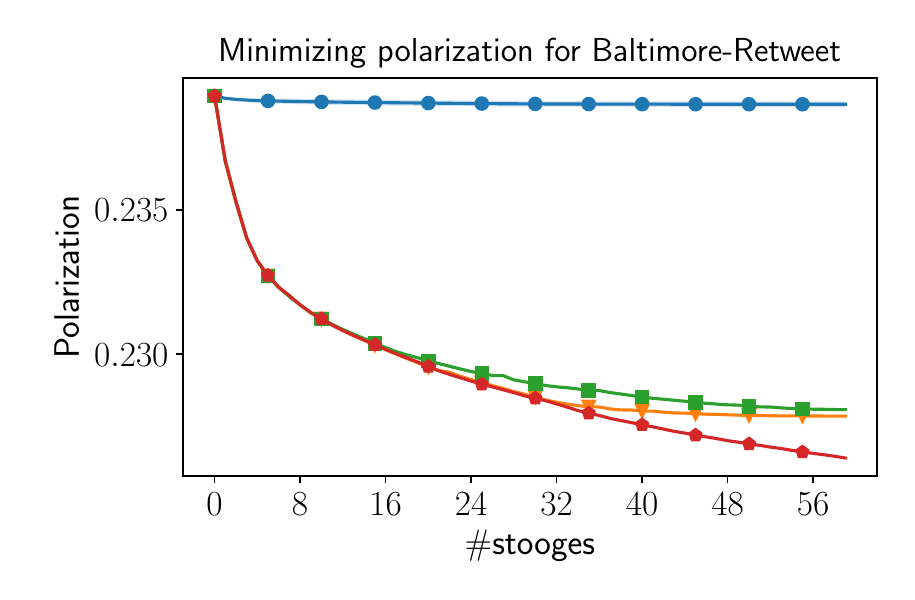}
    
    \includegraphics[width=\figwidth]{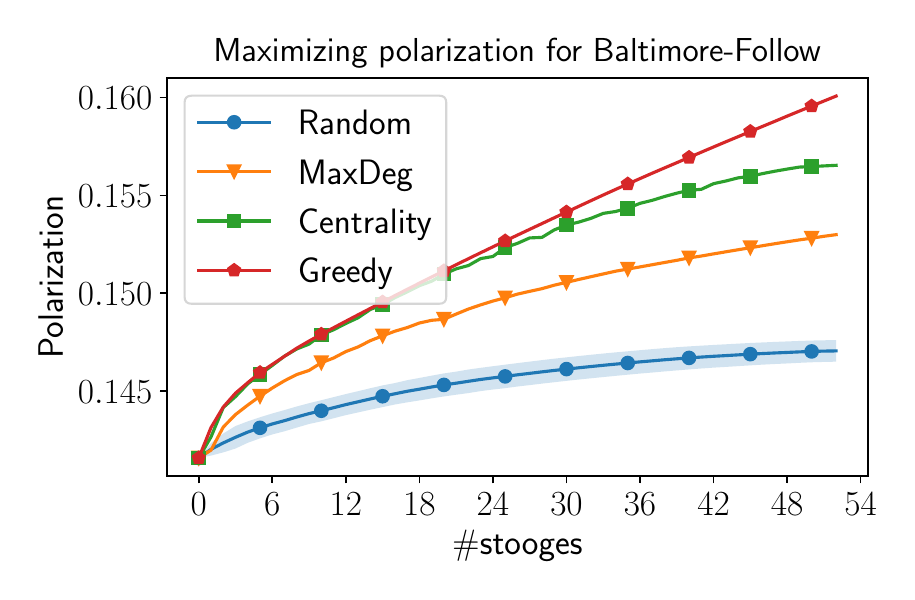}~
    \includegraphics[width=\figwidth]{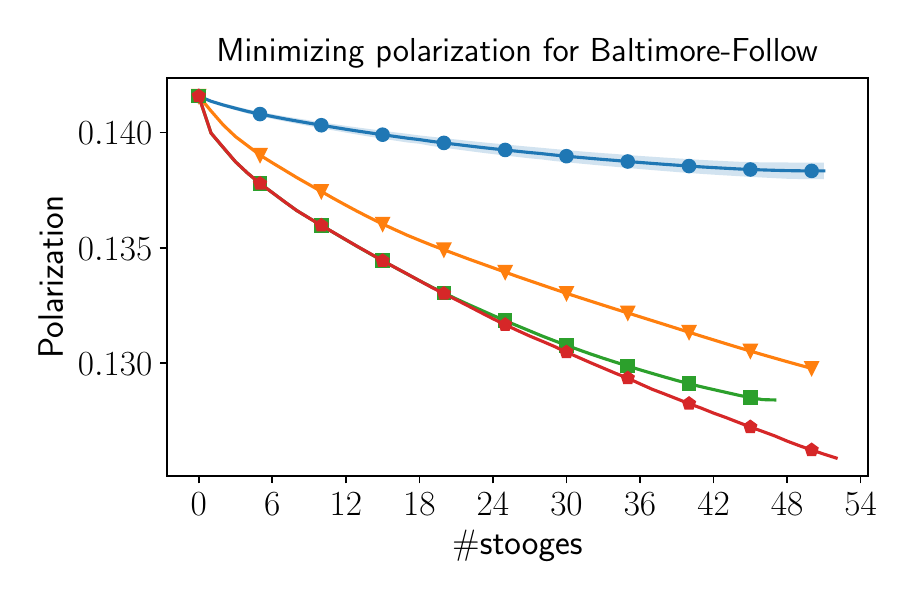}
    
    \includegraphics[width=\figwidth]{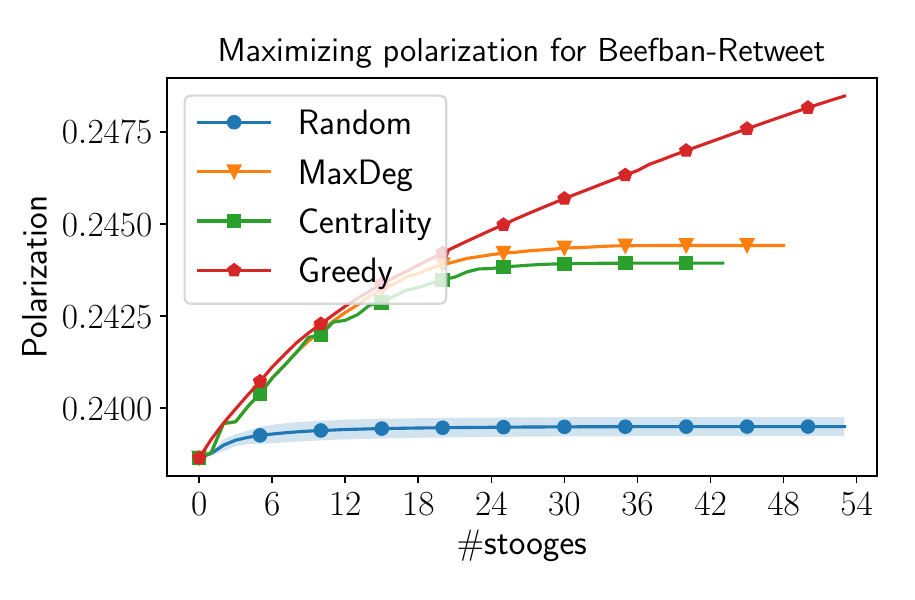}~
    \includegraphics[width=\figwidth]{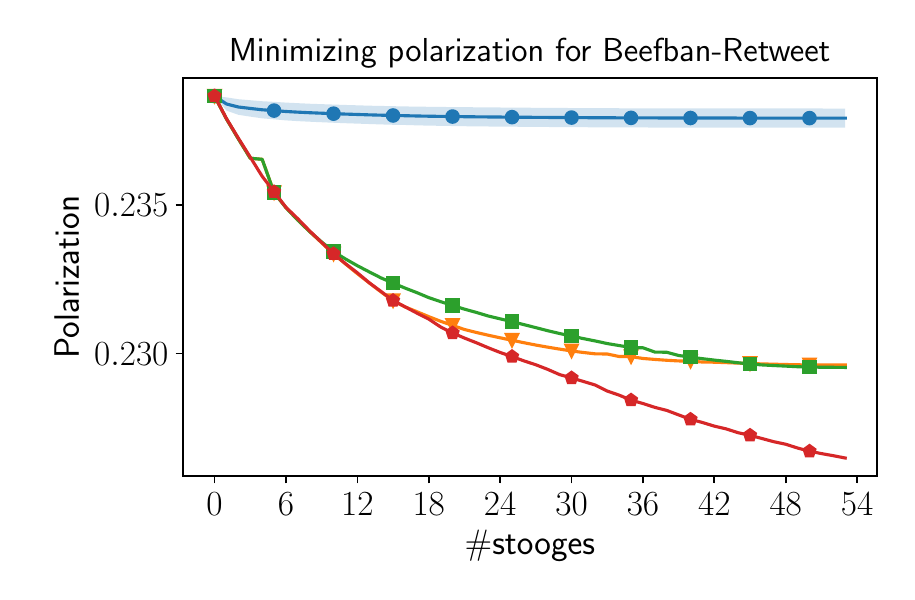}
    
    \caption{Maximization (left) and minimization (right) of the polarization on various
    real-world datasets, analogous to Figure~\ref{fig:synthetic-mse}}
    \label{fig:real-world-1-pol}
\end{figure}

\begin{figure}
    \centering
    
    \includegraphics[width=\figwidth]{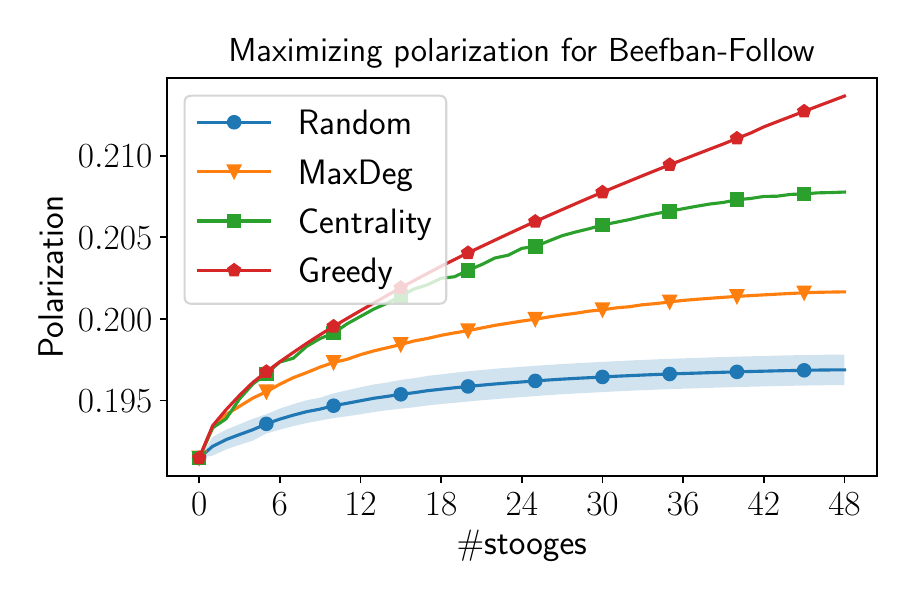}~
    \includegraphics[width=\figwidth]{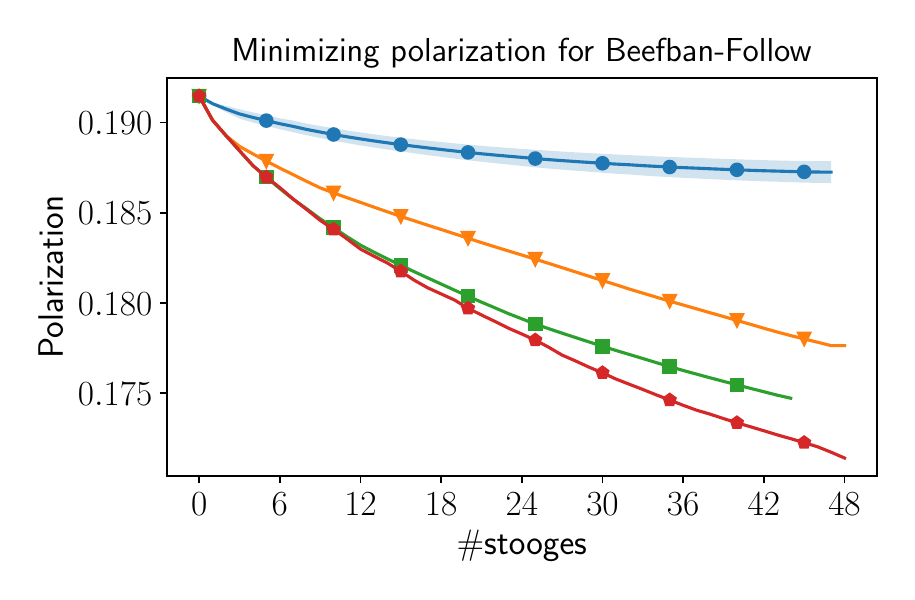}
    
    \includegraphics[width=\figwidth]{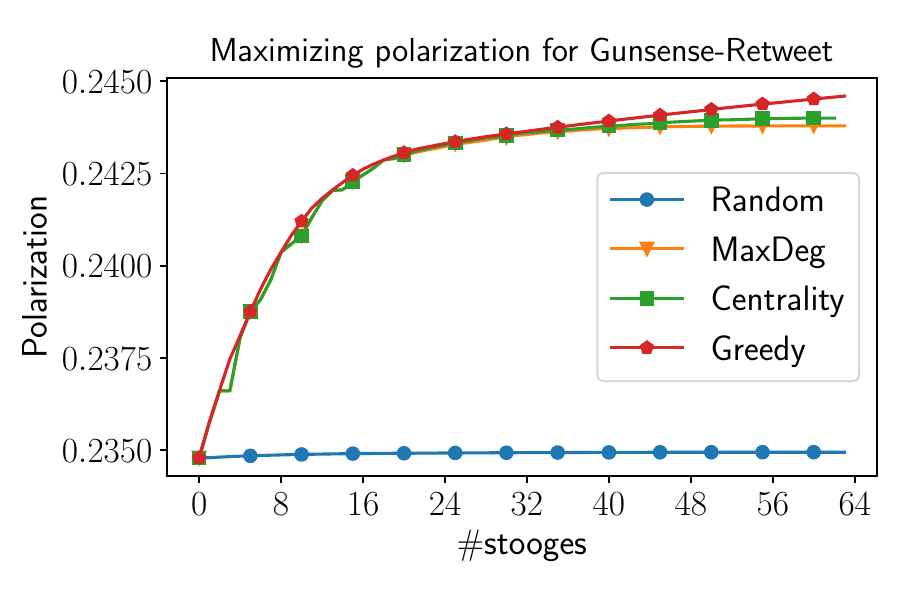}~
    \includegraphics[width=\figwidth]{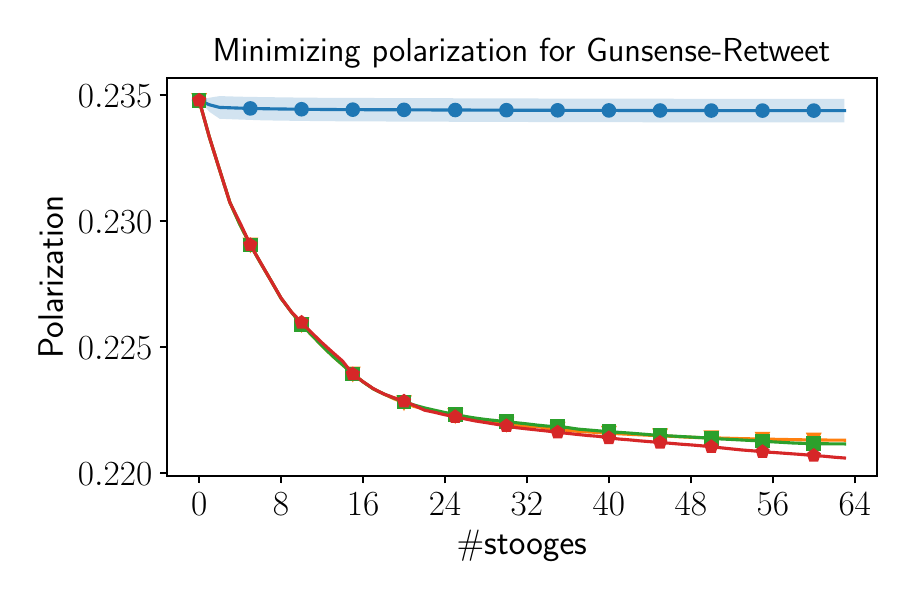}

    \includegraphics[width=\figwidth]{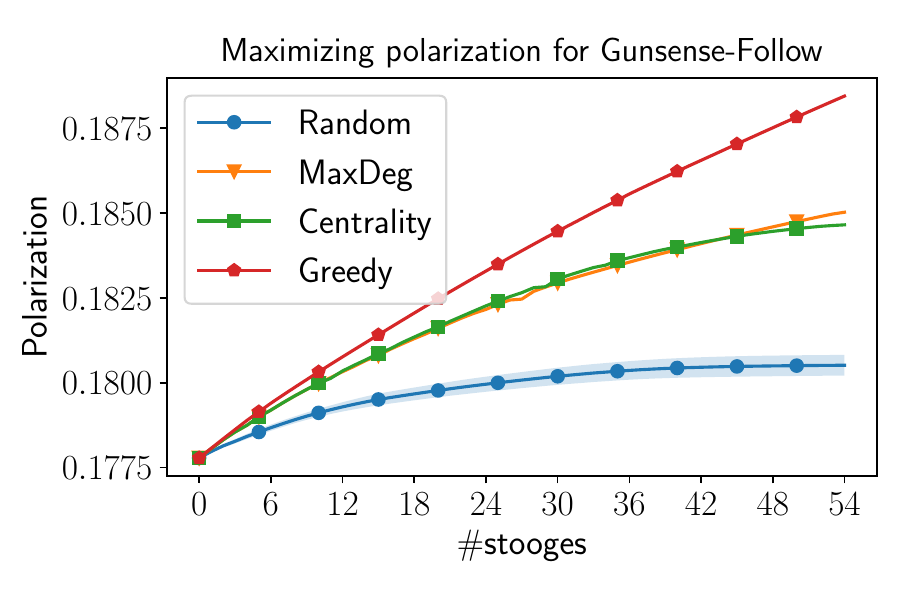}~
    \includegraphics[width=\figwidth]{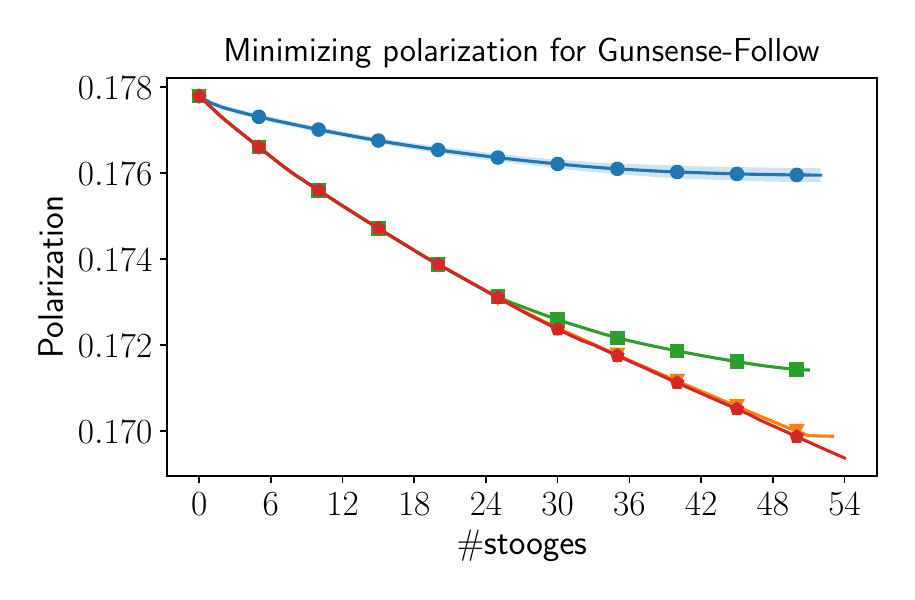}

    \includegraphics[width=\figwidth]{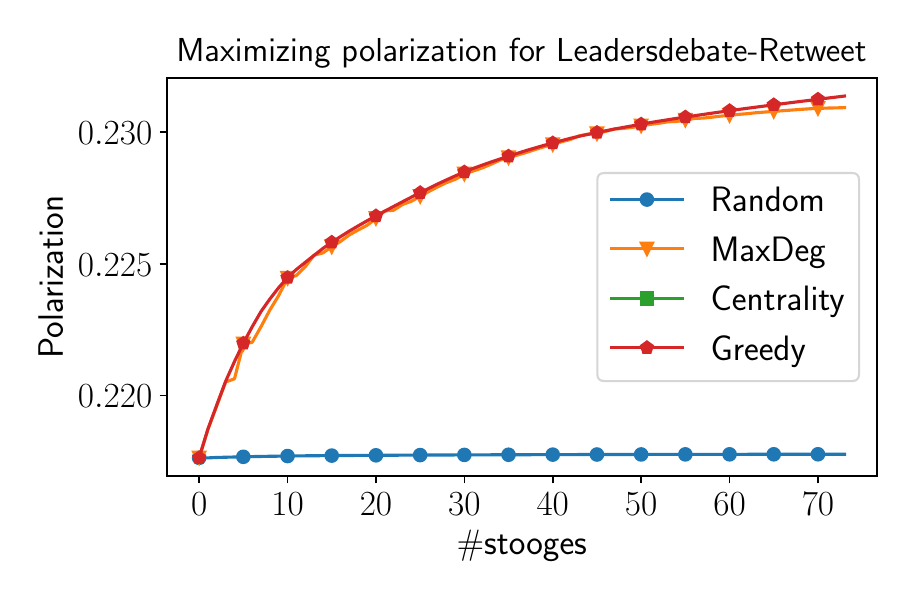}~
    \includegraphics[width=\figwidth]{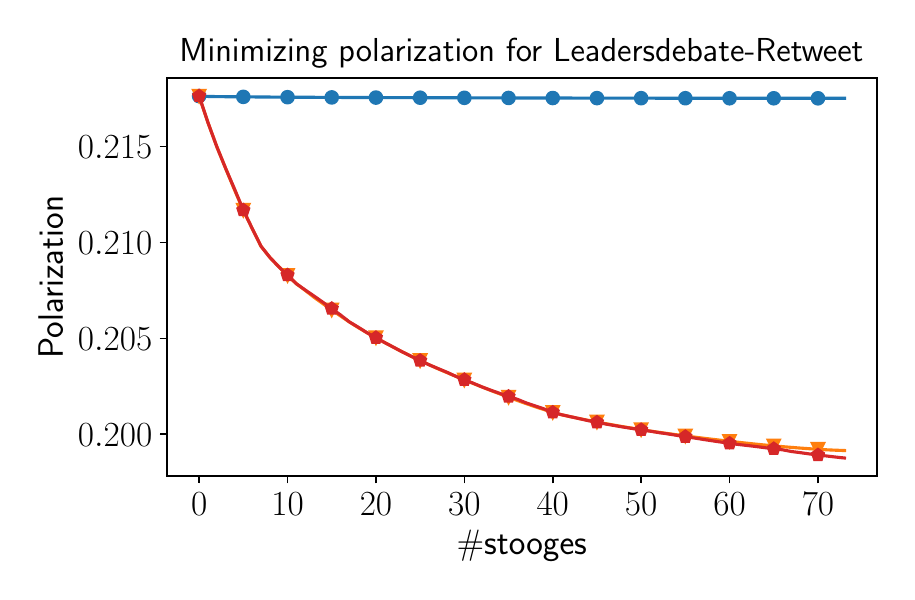}
    
    \caption{Maximization (left) and minimization (right) of the polarization on various
    real-world datasets, analogous to Figure~\ref{fig:synthetic-mse}}
    \label{fig:real-world-2-pol}
\end{figure}

\begin{figure}
    \centering
    
    \includegraphics[width=\figwidth]{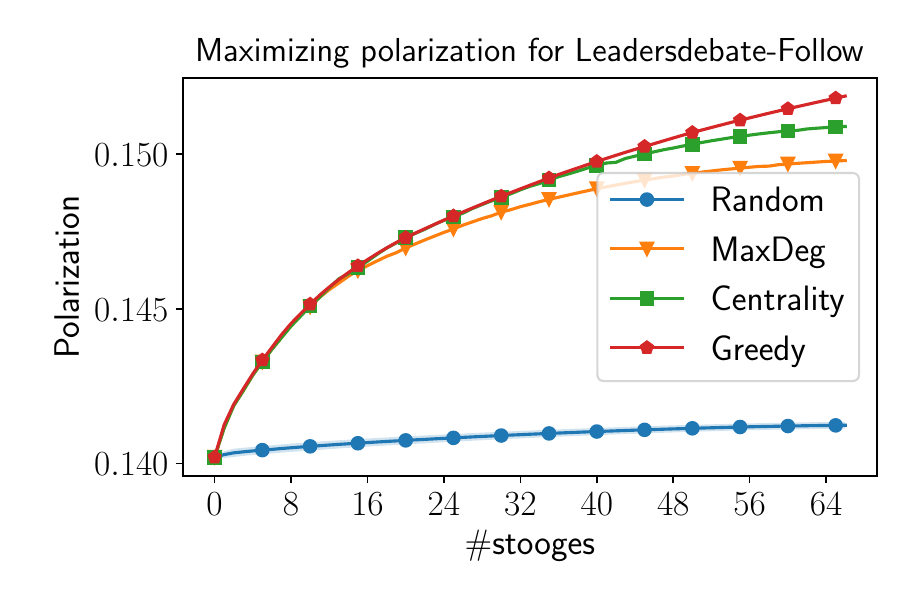}~
    \includegraphics[width=\figwidth]{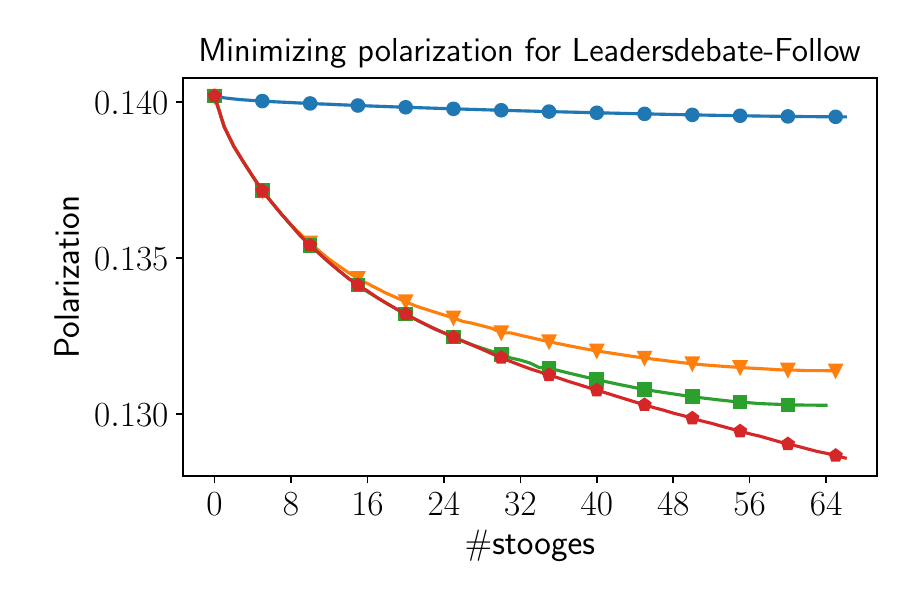}
    
    \includegraphics[width=\figwidth]{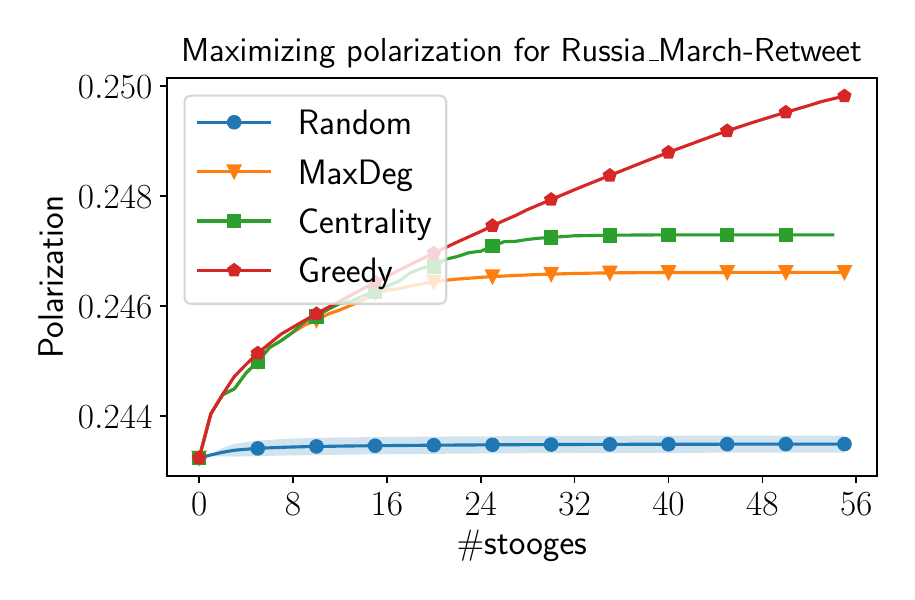}~
    \includegraphics[width=\figwidth]{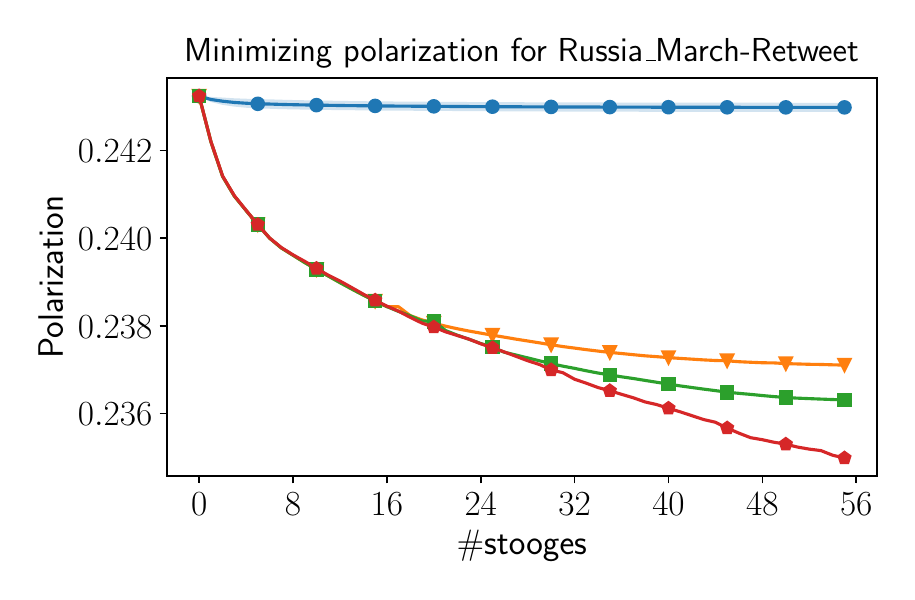}

    \includegraphics[width=\figwidth]{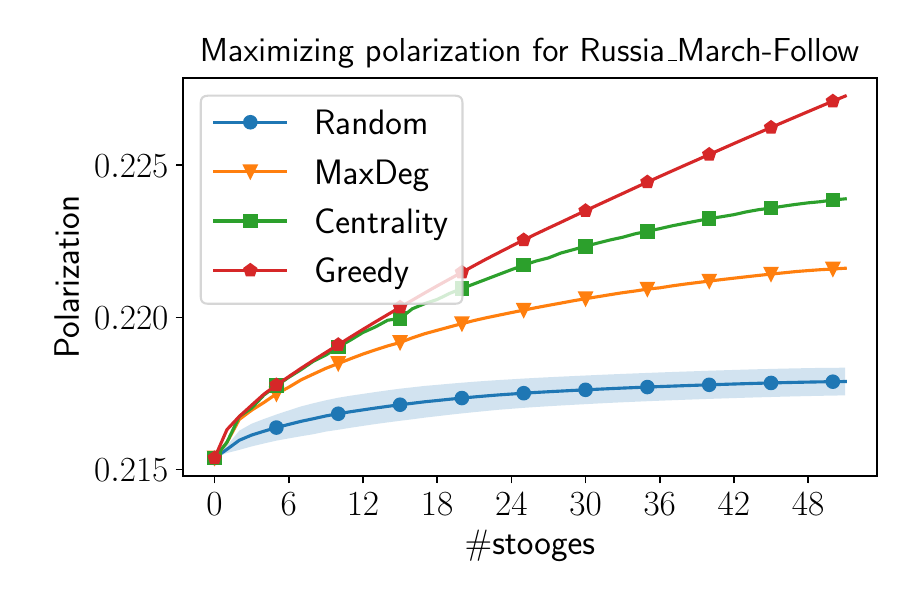}~
    \includegraphics[width=\figwidth]{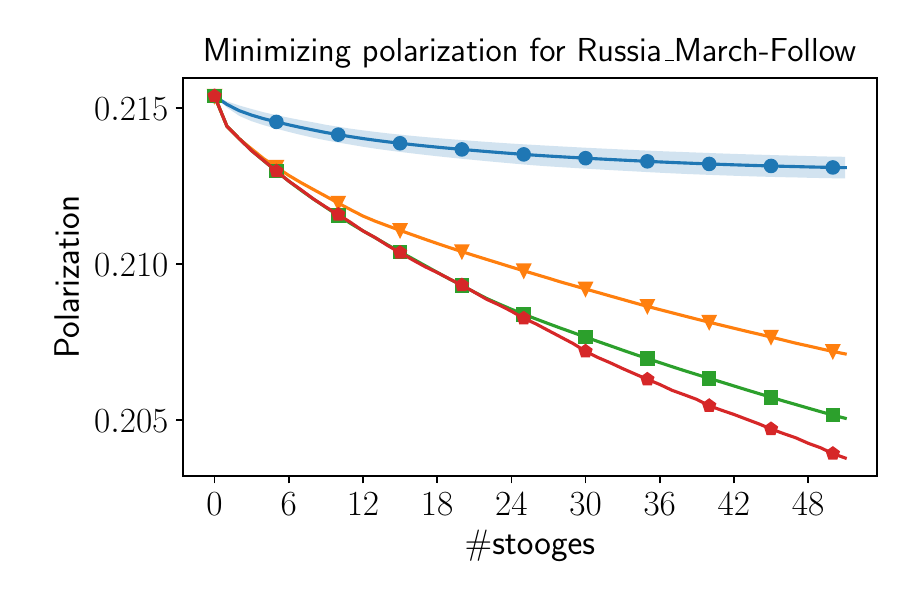}
    
    \caption{Maximization (left) and minimization (right) of the polarization on various
    real-world datasets, analogous to Figure~\ref{fig:synthetic-pol}}
    \label{fig:real-world-3-pol}
\end{figure}

\end{document}